%% file: main.tex
\pdfoutput=1
\documentclass[acmsmall,10pt,screen]{acmart}
\acmJournal{PACMPL}
\acmVolume{1}
\acmNumber{POPL} 
\acmArticle{1}
\acmYear{2022}
\acmMonth{1}
\acmDOI{} 
\startPage{1}

\renewcommand\footnotetextcopyrightpermission[1]{} 

\input{packages}

\input{macros}

\makeatletter
\g@addto@macro\bfseries{\boldmath}
\makeatother

\begin{document}

\title{The Decidability and Complexity of Interleaved Bidirected Dyck Reachability}

\author{Adam Husted Kjelstrøm}
\affiliation{
\institution{Aarhus University}            
\streetaddress{Aabogade 34}
\city{Aarhus}
\postcode{8200}
\country{Denmark}                    
}
\email{au640702@post.au.dk}          

\author{Andreas Pavlogiannis}
\affiliation{
\institution{Aarhus University}            
\streetaddress{Aabogade 34}
\city{Aarhus}
\postcode{8200}
\country{Denmark}                    
}
\email{pavlogiannis@cs.au.dk}          

\input{abstract}

\maketitle

\input{introduction}
\input{preliminaries}
\input{d1d1}

\input{dkd1}

\input{dkdk}
\input{experiments}

\input{related_work}
\input{conclusion}

\begin{acks}
We are grateful to Georg Zetzsche for comments on an earlier draft of the paper, and to the authors of \cite{Li2021} for sharing their dataset and assisting us on its use.
We also thank anonymous reviewers for their constructive feedback.
\end{acks}

\clearpage

\bibliography{bibliography}

\clearpage
\appendix
\input{appendix}

\end{document}

%% file: packages.tex
\usepackage[utf8]{inputenc} 
\usepackage{amsmath, amsthm}
\usepackage{thmtools}
\usepackage{thm-restate}
\usepackage[english]{babel}
\usepackage{pdfpages}
\usepackage{paralist}
\usepackage{stmaryrd}
\usepackage{graphicx}
\usepackage{subcaption}
\usepackage{caption}
\usepackage{longtable}
\usepackage{booktabs}
\usepackage{tabu}
\usepackage[referable]{threeparttablex}
\usepackage{varwidth}
\usepackage{multirow}
\usepackage{wrapfig}
\usepackage{listings}
\usepackage{microtype}
\usepackage{stackengine}
\usepackage{csquotes}
\usepackage{circuitikz}
\tikzset{>=stealth}
\ctikzset{tripoles/american and port/width=0.8, tripoles/american and port/height=.65}
\ctikzset{tripoles/american or port/width=0.8, tripoles/american or port/height=.65}

\usepackage{pifont}

\usepackage[ruled, linesnumbered]{algorithm2e}
\usepackage{parskip}
\usepackage{multicol}

\usepackage{xparse}

\usepackage{thm-restate}
\usepackage[capitalise]{cleveref}

\crefname{figure}{Figure}{Figure}
\crefformat{footnote}{#2\footnotemark[#1]#3}
\usepackage{tikz}
\usepackage{comment}

\usepackage{todonotes}
\usepackage{tikz}
\usetikzlibrary{positioning,arrows,automata,shapes,decorations,decorations.markings,calc, matrix,decorations.pathmorphing, patterns,backgrounds}

\usepackage{bm}
\usepackage{pgfplots}
\usepackage{enumerate,paralist}
\usepackage{pbox}

\usepackage{appendix}
\usepackage{calc}
\usepackage{fp}
\usepackage{multirow}
\usepackage{pbox}

\usepackage{etoolbox}

%% file: macros.tex
\theoremstyle{plain}
\newtheorem{remark}{Remark}

\DeclareMathAlphabet{\mathpzc}{OT1}{pzc}{m}{it}

\newcommand{\Nats}{\mathbb{N}}

\newcommand{\ov}{\overline}
\newcommand{\wh}{\widehat}

\newcommand{\Project}{\downharpoonright}
\newcommand{\True}{\mathsf{True}}
\newcommand{\False}{\mathsf{False}}

\newcommand{\Path}{\rightsquigarrow}
\newcommand{\DPath}[1]{\mathrel{\stackon[1pt]{$\rightsquigarrow$}{$\scriptscriptstyle#1$}}}
\newcommand{\DTo}[1]{\xrightarrow{#1}}
\newcommand{\DODOPath}{\DPath{\Dyck_1\odot \Dyck_1}}
\newcommand{\DKDOPath}{\DPath{\Dyck_k\odot \Dyck_1}}
\newcommand{\DKDKPath}{\DPath{\Dyck_k\odot \Dyck_k}}
\newcommand{\DKPath}{\DPath{\Dyck_k}}
\newcommand{\DOPath}{\DPath{\Dyck_1}}
\newcommand{\AlgoDKDOne}{\operatorname{Algo-\Dyck_k\odot \Dyck_1}}
\newcommand{\AlgoDODOne}{\operatorname{Algo-\Dyck_1\odot \Dyck_1}}

\newcommand{\NP}{\operatorname{NP}}
\newcommand{\NL}{\operatorname{NL}}

\newcommand{\NPH}{\operatorname{NP-hard}}

\newcommand{\Label}{\lambda}
\newcommand{\Paragraph}[1]{\smallskip\noindent{\bf #1}}
\newcommand{\SubParagraph}[1]{\smallskip\noindent{\em #1}}
\newcommand{\Alphabet}{\Sigma}
\newcommand{\Open}{\operatorname{Open}}
\newcommand{\Close}{\operatorname{Close}}
\newcommand{\OpenNonTerminal}{\mathcal{A}}
\newcommand{\CloseNonTerminal}{\ov{\mathcal{A}}}
\newcommand{\StartNonTerminal}{\mathcal{S}}

\newcommand{\OV}{\operatorname{OV}}

\newcommand{\StackHeight}{\operatorname{SH}}
\newcommand{\MaxStackHeight}{\operatorname{MaxSH}}
\newcommand{\MaxCounter}{\operatorname{MaxCnt}}

\newcommand{\CounterIndex}{\operatorname{CntInd}}

\newcommand{\Prefixes}{\operatorname{Pref}}

\newcommand{\Dyck}{\mathcal{D}}
\newcommand{\Language}{\mathcal{L}}

\newcommand{\Grammar}{\mathcal{G}}
\newcommand{\Produces}{\vdash}
\newcommand{\PTime}{\operatorname{PTime}}
\newcommand{\DUnion}{\uplus}
\newcommand{\Clause}{\mathcal{C}}
\newcommand{\Counter}{\operatorname{Cnt}}
\newcommand{\Stack}{\operatorname{Stk}}

\newcommand\numberthis{\addtocounter{equation}{1}\tag{\theequation}}

\usetikzlibrary{arrows.meta,calc,decorations.markings,math,arrows.meta,shapes.gates.logic.US,fit}

\preto\tabular{\setcounter{magicrownumbers}{0}}
\newcounter{magicrownumbers}

\pgfdeclarelayer{bg}    
\pgfsetlayers{bg,main}  

\definecolor{mybluecolor}{rgb}{0.2549019607843137, 0.2117647058823529, 0.9823529411764706}
\def \darkred {black!20!red}
\def \darkgreen {black!30!green}

\def \darkorange{orange!95!black!100}
\def \lightorange{orange!95!black!40}
\SetKw{Continue}{continue}

\def\ColorOne{mybluecolor}
\def\ColorTwo{\darkorange}
\def\ColorTwoLight{\lightorange}

\newcommand{\MyColorOne}[1]{\textcolor{\ColorOne}{#1}}
\newcommand{\MyColorTwo}[1]{\textcolor{\ColorTwo}{#1}}

\SetCommentSty{mycommfont}

\makeatletter
\newcommand*{\centerfloat}{%
\parindent \z@
\leftskip \z@ \@plus 1fil \@minus \textwidth
\rightskip\leftskip
\parfillskip \z@skip}
\makeatother

\definecolor{mGreen}{rgb}{0,0.6,0}
\definecolor{mGray}{rgb}{0.5,0.5,0.5}
\definecolor{mPurple}{rgb}{0.58,0,0.82}
\definecolor{backgroundColour}{rgb}{0.95,0.95,0.92}

\lstdefinestyle{CStyle}{
tabsize = 2, 
showstringspaces = false, 
backgroundcolor=\color{backgroundColour},   
numbers = left, 
commentstyle = \color{green}, 
keywordstyle = \color{blue}, 
stringstyle = \color{red}, 
rulecolor = \color{black}, 
basicstyle = \small \ttfamily , 
breaklines = true, 
numberstyle = \tiny,
language=Java
}

\newcounter{listtotal}\newcounter{listcntr}%
\NewDocumentCommand{\names}{o}{%
  \setcounter{listtotal}{0}\setcounter{listcntr}{0}%
  \renewcommand*{\do}[1]{\stepcounter{listtotal}}%
  \expandafter\docsvlist\expandafter{\namesarray}%
  \IfNoValueTF{#1}
    {\namesarray}
    {
     \renewcommand*{\do}[1]{\stepcounter{listcntr}\ifnum\value{listcntr}=#1\relax##1\fi}%
     \expandafter\docsvlist\expandafter{\namesarray}}%
}

%% file: abstract.tex
 \begin{abstract}
Dyck reachability is the standard formulation of a large domain of static analyses, as it achieves the sweet spot between precision and efficiency, and has thus been studied extensively.
\emph{Interleaved} Dyck reachability (denoted $\Dyck_k\odot \Dyck_k$) uses two Dyck languages for increased precision (e.g., context and field sensitivity) but is well-known to be undecidable.
As many static analyses yield a certain type of \emph{bidirected} graphs,
they give rise to \emph{interleaved bidirected Dyck reachability} problems.
Although these problems have seen numerous applications, their decidability and complexity has largely remained open.
In a recent work, Li \textit{et al}. made the first steps in this direction, showing that
(i)~$\Dyck_1\odot \Dyck_1$ reachability (i.e., when both Dyck languages are over a single parenthesis and act as counters) is computable in $O(n^7)$ time, while
(ii)~$\Dyck_k\odot \Dyck_k$ reachability is $\NPH$.
However, despite this recent progress, most natural questions about this intricate problem are open.

In this work we address the decidability and complexity of all variants of interleaved bidirected Dyck reachability.
First, we show that $\Dyck_1\odot \Dyck_1$ reachability can be computed in $O(n^3\cdot \alpha(n))$ time,
significantly improving over the existing $O(n^7)$ bound.
Second, we show that $\Dyck_k\odot \Dyck_1$ reachability (i.e., when one language acts as a counter) is decidable, in contrast to the non-bidirected case where decidability is open.
We further consider $\Dyck_k\odot \Dyck_1$ reachability where the counter remains linearly bounded.
Our third result shows that this bounded variant can be solved in $O(n^2\cdot \alpha(n))$ time,
while our fourth result shows that the problem has a (conditional) quadratic lower bound, and thus our upper bound is essentially optimal.
Fifth, we show that full $\Dyck_k\odot \Dyck_k$ reachability is undecidable.
This improves the recent $\NP$-hardness lower-bound, and shows that the problem is equivalent to the non-bidirected case.
Our experiments on standard benchmarks show that the new algorithms are very fast in practice, offering many orders-of-magnitude speedups over previous methods.
\end{abstract}

\ccsdesc[500]{Software and its engineering~Software verification and validation}
\ccsdesc[300]{Theory of computation~Theory and algorithms for application domains}
\ccsdesc[300]{Theory of computation~Program analysis}

\keywords{static analysis, CFL/Dyck reachability, bidirected graphs, complexity}  

%% file: introduction.tex
\section{Introduction}\label{sec:intro}

Static analyses are one of the most common approaches to program analysis.
They offer tunable approximations to program behavior, by representing the set of possible program executions in the semantics of a simpler model (e.g., a graph, or a set of constraints) that is amenable to efficient analysis.
Although the model typically over-approximates program behavior, and may thus contain unrealizable executions,
it can be used to make sound conclusions of correctness, and even raise warnings of potential erroneous behavior.
In essence, this approach casts the semantic question of program correctness to an algorithmic question about the model (e.g., ``determine the existence of a path in a graph'', or ``find the least fixpoint to a system of constraints'').

\Paragraph{Dyck reachability.}
The standard formalism of a plethora of static analyses is via language graph reachability~\cite{Reps97}.
Informally the nodes of a graph represent various program segments (e.g., variables), 
while edges capture relationships between those segments, such as program flow or data dependencies.
In order to refine such relationships, the edges are annotated with letters of some alphabet.
This formalism reduces the analysis question to a reachability question between nodes in the graph,
as witnessed by paths whose labels along the edges produce a string that belongs to a language $\Language$.
Although $\Language$ varies per application, it is almost always a form a visibly pushdown language~\cite{Alur2004},
yielding the problem commonly known as \emph{CFL/Dyck reachability}.

The Dyck reachability problem (denoted $\Dyck_k$) has applications to a very wide range of static analyses,
such as interprocedural data-flow analysis~\cite{Reps95}, slicing~\cite{Reps94}, 
shape analysis~\cite{Reps1995b}, impact analysis~\cite{Arnold96}, type-based flow analysis~\cite{Rehof01}, taint analysis~\cite{Huang2015},
data-dependence analysis~\cite{Tang2017},
alias/points-to analysis~\cite{Lhotak06,Zheng2008,Xu09}, and many others.
In practice, widely-used large-scale analysis tools, such as Wala~\cite{Wala} and Soot~\cite{Bodden12}, 
equip Dyck-reachability techniques to perform the analysis.
In all such cases, the balanced-parenthesis property of Dyck languages is ideal for expressing sensitivity of the analysis
in matching calling contexts, matching field accesses of composite objects, matching pointer references and dereferences etc.
This sensitivity improves precision in the obtained results and drastically reduces false positives.

\input{figures/example}

\Paragraph{Interleaved Dyck reachability and variants.}
When two types of sensitivity are in place simultaneously,
the underlying graph problem is lifted to \emph{interleaved} Dyck reachability over two Dyck languages (denoted $\Dyck_k\odot \Dyck_k$).
Intuitively, each language acts independently over its own alphabet to ensure sensitivity in one aspect (e.g., context vs field sensitivity).
This time, a path is accepted as a reachability witness iff the produced string wrt to each alphabet belongs to the respective language (see \cref{fig:example} for an illustration).
Unfortunately, interleaved Dyck reachability is well-known to be undecidable~\cite{Reps00}.
Nevertheless, the importance of the problem has lead to the development of various approximations that often work well in practice.
We refer to \cref{sec:related_work} for relevant literature in this area.

When a Dyck language is over a single parenthesis, the underlying stack acts as a counter.
This gives rise to variants of interleaved Dyck reachability, depending on whether one ($\Dyck_k\odot \Dyck_1$) or both ($\Dyck_1\odot \Dyck_1$) languages are over a single parenthesis.
Although these are less expressive models, they have rich modeling power and computational appeal.
For a $\Dyck_k\odot \Dyck_1$  example, $\Dyck_1$ can express that a queue is dequeued as many times as it is enqueued, and in the right order,
while $\Dyck_k$ is used as before to make the analysis context sensitive.
Similarly, $\Dyck_1\odot \Dyck_1$ can be used to identify an execution in which the program properly uses two interleaved but independent reentrant locks.

\Paragraph{Bidirected graphs.}
A Dyck graph is bidirected if every edge is present in both directions, and the two mirrored edges have complementary labels.
That is, $a\DTo{}b$ opens a parenthesis iff $b\DTo{}a$ closes the same parenthesis.
Bidirectedness turns reachability into an equivalence relation, much like the case of plain undirected graphs.
This happens because every path produces the same string when traversed backwards, and thus either both directions witness reachability or none does.
From a semantics perspective, bidirectedness is a standard approach to handle mutable heap data~\cite{Sridharan2006,Xu09,Lu2019,Zhang2017}
-- though it can sometimes be relaxed for read-only accesses~\cite{Milanova2020},
and the de-facto formulation of demand-driven points-to analyses~\cite{Sridharan2005,Zheng2008,Shang2012,Yan11,Vedurada2019}.
Bidirectedness is also used for CFL-reachability formulations of pointer analysis~\cite{Reps97},
and has also been used for simplifying the input graph~\cite{Li2020}.
Thus in all these cases, the analysis yields the problem of \emph{interleaved bidirected Dyck reachability}.

\Paragraph{Open questions.}
Given the many applications of interleaved bidirected Dyck reachability, 
it is surprising that very little has been known about its decidability and complexity.
This gap is even more pronounced if contrasted to the rich literature on the complexity of Dyck/CFL reachability (see \cref{sec:related_work}).
The recent work of~\cite{Li2021} makes an important step in this direction, by proving an upper bound of $O(n^7)$ for $\Dyck_1\odot \Dyck_1$, and also showing  that the $\Dyck_k\odot \Dyck_k$ case is at least $\NP$-hard.
However, most of the fundamental questions remain unanswered, such as the following.
\begin{compactenum}
\item Although the $O(n^7)$ bound is polynomial, it remains prohibitively large.
Is there a faster algorithm, e.g., one that operates in quadratic/cubic time that is common in static analyses?
\item What is the decidability and complexity for $\Dyck_k\odot \Dyck_1$?
As $\Dyck_k\odot \Dyck_1$ is more expressive than $\Dyck_1\odot \Dyck_1$,
it leads to better precision in static analyses, and thus it is very well motivated.
Moreover, the non-bidirected case is also known as pushdown vector addition systems, 
where the decidability of reachability has been open (see \cref{sec:related_work}).
Is bidirectedness sufficient to show decidability?
\item Given the undecidability of $\Dyck_k\odot \Dyck_k$ on general graphs,
does bidirectedness make the problem decidable?

\end{compactenum}

This work delivers several new results on the decidability and complexity of all variants of interleaved bidirected Dyck reachability, and paints the rich landscape of the problem.

\input{contributions}

%% file: figures/example.tex
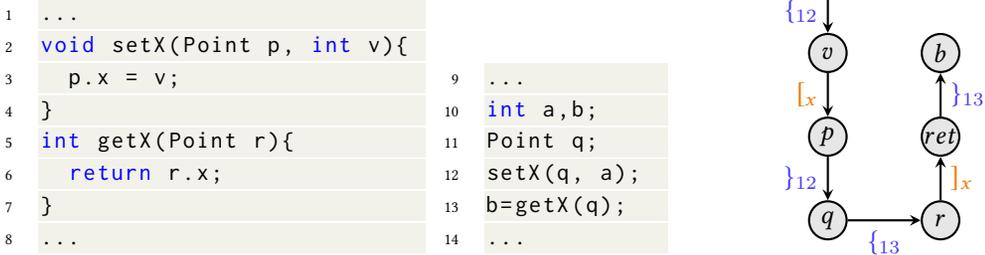
\begin{figure}
\begin{subfigure}[b]{0.37\textwidth}
\centering
\begin{lstlisting}[style=CStyle]
...
void setX(Point p, int v){
  p.x = v;    
}
int getX(Point r){
  return r.x;    
}
...
\end{lstlisting}
\end{subfigure}
\qquad
\begin{subfigure}[b]{0.175\textwidth}
\centering
\begin{lstlisting}[style=CStyle,firstnumber=9]
...
int a,b;
Point q;
setX(q, a);
b=getX(q);
...
\end{lstlisting}
\end{subfigure}
\qquad
\begin{subfigure}[b]{0.3\textwidth}
\newcommand{\xtstep}{0.9}
\newcommand{\ytstep}{0.5}
\newcommand{\xstep}{1.5}
\newcommand{\ystep}{1.1}
\centering
\begin{tikzpicture}[thick,
pre/.style={<-, thick, shorten >=0cm,shorten <=0cm},
post/.style={->, thick, shorten >=-0cm,shorten <=-0cm},
seqtrace/.style={->, line width=2},
postpath/.style={->, thick, decorate, decoration={zigzag,amplitude=1pt,segment length=2mm,pre=lineto,pre length=2pt, post=lineto, post length=4pt}},
node/.style={thick, draw=black, circle, inner sep = 0, minimum size=5mm, fill=gray!20},
]

\begin{scope}[shift={(0,0)}]

\node[node] (a) at (0*\xstep, 0*\ystep) {$a$};
\node[node] (v) at (0*\xstep, -1*\ystep) {$v$};
\node[node] (p) at (0*\xstep, -2*\ystep) {$p$};
\node[node] (q) at (0*\xstep, -3*\ystep) {$q$};
\node[node] (r) at (1*\xstep, -3*\ystep) {$r$};
\node[node] (ret) at (1*\xstep, -2*\ystep) {$ret$};
\node[node] (b) at (1*\xstep, -1*\ystep) {$b$};

\draw[post] (a) to node[left]{\MyColorOne{$\{_{12}$}} (v);
\draw[post] (v) to node[left]{\MyColorTwo{$[_{x}$}} (p);
\draw[post] (p) to node[left]{\MyColorOne{$\}_{12}$}} (q);
\draw[post] (q) to node[below]{\MyColorOne{$\{_{13}$}} (r);
\draw[post] (r) to node[right]{\MyColorTwo{$]_{x}$}} (ret);
\draw[post] (ret) to node[right]{\MyColorOne{$\}_{13}$}} (b);

\end{scope}

\end{tikzpicture}
\end{subfigure}
\caption{
\emph{(Left):} A program on which to perform context-sensitive and field-sensitive alias analysis.
\emph{(Right):} An interleaved Dyck graph where the curly braces (blue) model context sensitivity and the square brackets (orange) model field sensitivity.
The path $a\Path b$ produces two interleaved strings, \MyColorOne{$\{_{12}\}_{12}\{_{13}\}_{13}$} and \MyColorTwo{$[_{x}]_{x}$}.
As both strings are well-balanced, the path is a valid witness and thus $b$ may alias $a$.
}
\label{fig:example}
\end{figure}

%% file: contributions.tex
\subsection{Our contributions}\label{subsec:contributions}

In this section we state the main results of this paper, and put them in context with regards to existing literature (see \cref{tab:results} for a summary).
We consider as input a bidirected interleaved Dyck graph $G$ with $n$ nodes.

\Paragraph{Bidirected $\Dyck_1\odot \Dyck_1$ reachability.}
We start with the case of bidirected $\Dyck_1\odot \Dyck_1$ reachability.
The non-bidirected case falls into the class of vector addition systems with states.
In two dimensions, the problem is known to be in $\NL$ (and thus in $\PTime$)~\cite{Englert2016}.
The recent work of~\cite{Li2021} established an $O(n^7)$ bound for the bidirected case.
As our first theorem shows, the problem admits a must faster solution, namely, in essentially cubic time,
which is a common complexity bound in static analyses.

\input{tables/results}

\smallskip
\begin{restatable}{theorem}{thmd1d1cubicupperbound}\label{thm:d1d1_upper_bounds}
Bidirected $\Dyck_1\odot \Dyck_1$ reachability can be computed in $O(n^3\cdot \alpha(n))$ time, where $\alpha(n)$ is the inverse Ackermann function.
\end{restatable}

We obtain \cref{thm:d1d1_upper_bounds} by proving a boundedness property on paths witnessing reachability.
In particular, we show that wlog, along every such path \emph{both counters} remain quadratically bounded.
This strengthens the shallow-path property of~\cite{Li2021} which states that at any time \emph{one of the two counters} is bounded (though the bound is linear).

We next turn our attention to more expressive variants of interleaved Dyck reachability.

\Paragraph{Bidirected $\Dyck_k\odot \Dyck_1$ reachability.}
Here we address bidirected $\Dyck_k\odot \Dyck_1$ reachability.
The non-bidirected case is also known as pushdown vector addition systems in one dimension~\cite{Leroux2015},
for which the decidability of reachability is open.
Here we first show that bidirectedness suffices to make the problem decidable.

\smallskip
\begin{restatable}{theorem}{}\label{thm:dkd1_decidable}
Bidirected $\Dyck_k\odot \Dyck_1$ reachability is decidable.
\end{restatable}

%

We next focus on the problem under a form of witness bounding,
which is a standard technique for efficient static analyses (e.g., in terms of context bounding~\cite{Shivers1991,ESOP17}, or field-limiting~\cite{Deutsch1994}).
In all these cases, the respective bound is not guaranteed a-priori, but the analysis guarantees soundness up to that bound, which is a very useful property.
Here, we consider the case where reachability is witnessed by paths on which the counter stays linearly bounded, while the stack has no restrictions.
We establish the following theorem.

\smallskip
\begin{restatable}{theorem}{thmdkd1quadraticupperbound}\label{thm:dkd1_quadratic_upper_bound}
Bidirected $\Dyck_k\odot \Dyck_1$ reachability with $O(n)$-bounded counters can be computed in $O(n^2\cdot \alpha(n))$ time, where $\alpha(n)$ is the inverse Ackermann function.
\end{restatable}

Hence, restricting to a linear bound on the counter makes the problem efficiently solvable.
The next natural question is whether this restrictive setting admits even faster algorithms, e.g., can the problem be solved in $O(n)$ time?
We answer this question in negative.

\smallskip
\begin{restatable}{theorem}{thmdkd1quadraticupperbound}\label{thm:dkd1_quadratic_lower_bound}
For any fixed $\epsilon>0$,
bidirected $\Dyck_k\odot \Dyck_1$ reachability with $O(n)$-bounded counters cannot be solved in $O(n^{2-\epsilon})$ time under $\OV$.
\end{restatable}
Orthogonal Vectors ($\OV$) is a well-studied problem with a long-standing quadratic worst-case upper bound.
The corresponding hypothesis states that there is no sub-quadratic algorithm for the problem~\cite{Williams19}.
It is also known that the strong exponential time hypothesis (SETH) implies the Orthogonal Vectors hypothesis~\cite{Williams05}.
Thus, under \cref{thm:dkd1_quadratic_lower_bound}, the upper-bound of \cref{thm:dkd1_quadratic_upper_bound} is  (nearly) optimal.

\Paragraph{Bidirected $\Dyck_k\odot \Dyck_k$ reachability.}
Finally, we turn our attention to the general case of bidirected $\Dyck_k\odot \Dyck_k$ reachability.
The non-bidirected case is well-known to be undecidable~\cite{Reps00}, 
while it was recently shown~\cite{Li2021} that the bidirected case is $\NPH$, 
leaving its decidability open.
The following theorem resolves this open question.

\smallskip
\begin{restatable}{theorem}{thmdkdkquadraticupperbound}\label{thm:dkdk_undecidable}
Bidirected $\Dyck_k\odot \Dyck_k$ reachability is undecidable.
\end{restatable}

In terms of practical implications, \cref{thm:dkdk_undecidable} establishes the undecidability of popular pointer and alias analyses as those in~\cite{Sridharan2005,Zheng2008,Shang2012,Yan11,Vedurada2019}.

\Paragraph{Experimental results.}
We have implemented our algorithms for $\Dyck_1\odot \Dyck_1$ and $\Dyck_k\odot \Dyck_1$ reachability (\cref{thm:d1d1_upper_bounds} and \cref{thm:dkd1_quadratic_upper_bound}, respectively), 
and have evaluated them on standard benchmarks on a conventional laptop.
Each algorithm handles the whole benchmark set in less than 5 minutes.
On the other hand, the previous algorithm of~\cite{Li2021} for $\Dyck_1\odot \Dyck_1$ reachability was reported to handle the same benchmark set in more than 3 days when run on a bigger machine.
Thus, besides the theoretical improvements, our new algorithms confer a significant practical speedup, and are the first ones ready to be used in practice.

In the following sections we develop relevant notation and present details of the above theorems.
To improve readability, in the main paper we present algorithms, examples, and proofs of all theorems.
To highlight the main steps of the proofs, we also present all intermediate lemmas;  many lemma proofs, however, are relegated to the appendix.

%% file: tables/results.tex
\begin{table}
\centering
\renewcommand{\arraystretch}{1.5}
\caption{
Summary of our results for interleaved bidirected Dyck reachability.
}
\label{tab:results}
\begin{tabular}{|c||c|c|c|c|}
\hline
& $\Dyck_1\odot \Dyck_1$ & $\Dyck_k\odot \Dyck_1$ & $\Dyck_k\odot \Dyck_1$ (bounded counter) & $\Dyck_k\odot \Dyck_k$\\
\hline
\hline
Upper Bound & $O(n^3\cdot \alpha(n))$ & Decidable & $O(n^2\cdot \alpha(n))$ & - \\
\hline
Lower Bound & - & - & $\OV$-hard & Undecidable \\
\hline
\end{tabular}
\end{table}

%% file: preliminaries.tex
\section{Preliminaries}\label{sec:prelim}

\Paragraph{General notation.}
Given a natural number $k\in \Nats$, 
we denote by $[k]$ the set $\{1,\dots, k \}$.
Given a sequence of elements $\sigma=e_1,\dots, e_{\ell}$,
we denote by $|\sigma|=\ell$ the length of $\sigma$,
and denote by $\Prefixes(\sigma)$ the set of prefixes of $\sigma$, i.e.,
$\Prefixes(\sigma)=\{e_1,\dots, e_{\ell'}\colon \ell'\in [\ell] \}$.
Given two sequences $\sigma_1$, $\sigma_2$, we denote by $\sigma_1\circ \sigma_2$ their concatenation.
Finally, we lift concatenation to sets of sequences, i.e., for two sets of sequences $L_1$, $L_2$,
we have $L_1\circ L_2=\{ \sigma_1\circ \sigma_2\colon \sigma_i\in L_i \text{ for each } i\in[2] \}$.

\Paragraph{Dyck Languages.}
Given a natural number $k\in \Nats$, a Dyck alphabet $\Alphabet=\{\alpha_i, \ov{\alpha}_i\}_{i\in [k]}$
is a finite \emph{alphabet} of $k$ parenthesis symbols, where
$\Open(\Alphabet)=\{ \alpha_i \}_{i\in k}$ and $\Close(\Alphabet)=\{ \ov{\alpha}_i \}_{i\in k}$
are the sets of open parenthesis symbols and close parenthesis symbols of $\Alphabet$, respectively\footnote{For more clarity in our presentation, we use Greek letters such as $\alpha, \beta$ to represent opening parenthesis symbols, and $\ov{\alpha}, \ov{\beta}$ to represent their matching closing parentheses.}.
We denote by $\Dyck(\Alphabet)$ the Dyck language over $\Alphabet$,
defined as the language of strings generated by the following context-free grammar $\Grammar$:
\[
\StartNonTerminal \to \StartNonTerminal ~ \StartNonTerminal ~ | ~ \OpenNonTerminal_1 ~ \CloseNonTerminal_1 ~ | ~ \dots ~ | ~ \OpenNonTerminal_k ~ \CloseNonTerminal_k ~ | ~\epsilon \ ;
\qquad 
\OpenNonTerminal_i \to \alpha_i ~ \StartNonTerminal \ ;\qquad
\CloseNonTerminal_i \to  \ov{\alpha}_i
\]
Given a string $\sigma$ and a non-terminal symbol $X$ of $\Grammar$, we write $X\Produces s$
to denote that $X$ produces $s$.
The Dyck language over $\Alphabet$ is then defined as  $\Dyck(\Alphabet)=\{\sigma\in \Alphabet^{*}\colon \StartNonTerminal\Produces \sigma  \}$.
For example, $\alpha_1\alpha_2\ov{\alpha}_2 \alpha_3\ov{\alpha}_3\ov{\alpha}_1\in \Dyck(\Alphabet)$, but
$\alpha_1\alpha_2 \alpha_3\ov{\alpha}_3\ov{\alpha}_1 \ov{\alpha}_2 \not \in \Dyck(\Alphabet)$.
We typically ignore the alphabet $\Alphabet$ and write $\Dyck_k$ for the Dyck language over some implicit alphabet with $k$ different parenthesis symbols.
Finally, given a string $\sigma=\gamma_1\dots \gamma_m\in \Alphabet^*$, we let 
$\ov{\sigma}=\ov{\gamma}_1\dots \ov{\gamma}_m\in \Alphabet^*$,  
where $\ov{\gamma}_i$ is a closing parenthesis symbol if $\gamma_i$ is an opening parenthesis symbol, and vice versa.
For example, if $\sigma=\alpha_1\alpha_2\ov{\alpha}_2 \alpha_3\ov{\alpha}_3\ov{\alpha}_1$, then
$\ov{\sigma}=\ov{\alpha}_1\ov{\alpha}_2\alpha_2 \ov{\alpha}_3\alpha_3\alpha_1$.

\Paragraph{Interleaved Dyck languages.}
Given natural numbers $k_{1}, k_{2} \in \Nats$,
consider two alphabets $\Alphabet_{1}=\{\alpha_i, \ov{\alpha}_i\}_{i=1}^{k_{1}}$ and
$\Alphabet_{2}=\{\beta_i, \ov{\beta}_i\}_{i=1}^{k_{2}}$ with $\Alphabet_{1}\cap \Alphabet_{2}=\emptyset$, i.e., the two alphabets are disjoint. 
Given some word $w\in (\Alphabet_{1}\cup \Alphabet_{2})^*$, we denote by $w\Project \Alphabet$ the projection of $w$ on the alphabet $\Alphabet\in\{\Alphabet_{1}, \Alphabet_{2} \}$.
The interleaved Dyck language over the alphabet pair $(\Alphabet_{1}, \Alphabet_{2})$ is defined as 
\[
\Dyck(\Alphabet_{1})\odot \Dyck(\Alphabet_{2})=\left\{ \sigma\in (\Alphabet_{1}\cup \Alphabet_{2})^*\colon \sigma\Project\Alphabet_{1}\in \Dyck(\Alphabet_{1}) \text{ and }  \sigma\Project\Alphabet_{2}\in \Dyck(\Alphabet_{2})\right\}
\]
For example, we have $\alpha_1 \beta_1 \ov{\alpha}_1 \ov{\beta}_1\in \Dyck(\Alphabet_{1})\odot \Dyck(\Alphabet_{2})$.
Similarly as above, we typically ignore the alphabets $\Alphabet_{1}$ and $\Alphabet_{2}$, and write $\Dyck_{k_1}\odot \Dyck_{k_2}$ for the interleaved Dyck language over two implicit alphabets of sizes $k_1$ and $k_2$, with the understanding that the alphabets are disjoint.

\Paragraph{Graphs and language reachability.}
We consider labeled directed graphs $G=(V,E, \Alphabet)$ where $V$ is the set of nodes, $\Alphabet$ is an alphabet, and
$E\subseteq V\times V\times (\Alphabet\cup \{\epsilon \})$ is a set of edges (partially) labeled with letters from $\Alphabet$.
Given an edge $(u,v,\alpha)\in E$, we denote by $\Label(u,v,\alpha)=\alpha$ the label of the edge,
and often write $u\DTo{\alpha}v$ to denote the existence of such an edge.
A path $P=e_1,\dots, e_{\ell}$ is a sequence of edges. 
The label of $P$ is $\Label(P)=\Label(e_1)\dots \Label(e_{\ell})$, i.e., it is the concatenation of the labels of the edges along $P$.
We often write $P\colon u\Path v$ to denote a path from node $u$ to node $v$.
Naturally, we say that $v$ is reachable from $u$ if such a path exists.
Given some language $\Language\subseteq \Alphabet^*$, we say that $v$ is $\Language$-reachable from $u$ if there exists a path $P\colon u\Path v$ such that $\Label(P)\in \Language$, in which case we write $P\colon u\DPath{\Language}v$.

\SubParagraph{Notation on paths and cycles.}
A (simple) cycle is a path $C=e_{i_1}, e_{i_{i+1}},\dots,e_{i_j}$ such that the left endpoint of $e_{i_1}$ coincides with the right endpoint of  $e_{i_j}$ and no other node repeats in $C$.
A cyclic path is a path $P\colon u\Path u$ (though $P$ may also repeat nodes other than $u$, and thus not be a cycle).
Given a cyclic sub-path $P'$ of $P$, we denote by $P\setminus P'$ the path we obtain after removing $P'$ from $P$. 


\Paragraph{Dyck graphs.}
A labeled graph $G=(V,E,\Alphabet)$ is a Dyck graph if $\Alphabet$ is a Dyck alphabet.
Dyck reachability in $G$ is defined with respect to the corresponding Dyck language $\Dyck(\Alphabet)$.
A Dyck path $P\colon u\Path v$ is a path such that $\Label(P)\in \Dyck(\Alphabet)$, i.e., $P$ witnesses the reachability of $v$ from $u$ wrt the Dyck language $\Dyck_k$.

\Paragraph{Interleaved Dyck graphs.}
An interleaved Dyck graph is a Dyck graph $G=(V,E,\Alphabet)$ where $\Alphabet$ is implicitly partitioned into two disjoint Dyck alphabets $\Alphabet=\Alphabet_1\DUnion\Alphabet_2$.
Interleaved Dyck reachability in $G$ is defined wrt the respective interleaved Dyck language $\Dyck(\Alphabet_1)\odot \Dyck(\Alphabet_2)$.
That is, a path $P\colon u\Path v$ witnesses the reachability of $v$ from $u$ in $G$ iff $\Label(P)\Project \Alphabet_i\in \Dyck(\Alphabet_i)$ for each $i\in[2]$.
Given some $i\in[2]$, we will write $\Stack_i(P)$ to denote content of the $i$-th stack at the end of $P$,
with the natural interpretation that open parenthesis symbols push on the stack and close parenthesis symbols pop from the stack.
Moreover we use $\StackHeight_i(P)$ and $\MaxStackHeight_i(P)$ to denote the stack height and maximum stack height of $P$ wrt alphabet $\Alphabet_i$.
Formally,
\[
\StackHeight_i(P)= |\Label(P)\Project \Open(\Alphabet_i)| -  |\Label(P)\Project \Close(\Alphabet_i)|\qquad\text{and}\qquad
\MaxStackHeight_i(P)=\max_{P'\in \Prefixes(P)} \StackHeight_i(P')
\]
\cref{fig:graphs} shows an interleaved Dyck graph.
We will often project $G$ to some alphabet $\Alphabet_i$, for some $i\in[2]$, in which case we obtain a (non-interleaved) Dyck graph.
That is, the projection $G\Project \Alphabet_i$ is identical to $G$ where every edge label of $\Alphabet_{3-i}$ is replaced by $\epsilon$.

\SubParagraph{Special cases.}
We obtain two special cases of interleaved Dyck graphs when one, or both Dyck alphabets are binary (i.e., consisting of one opening-parenthesis symbol and the corresponding closing-parenthesis symbol).
When a Dyck alphabet is binary, the corresponding stack behaves as a counter.
To make this case explicit, if $|Alphabet_i|=2$ for some $i\in[2]$, we refer to the stack as a counter, and 
given a path $P$, we write $\Counter_i=\StackHeight_i(P)$ and $\MaxCounter_i(P)=\MaxStackHeight_i(P)$.

\input{figures/graphs}

\Paragraph{Bidirected graphs.}
Consider a labeled graph $G=(V,E,\Alphabet)$ where $\Alphabet$ is a Dyck alphabet.
We call $G$ bidirected if it satisfies the following condition, where we take $\ov{\epsilon}=\epsilon$.
\[
\forall u,v\in V, \alpha\in \{\epsilon\}\cup \Alphabet\colon (u,v,\alpha)\in E \text{ iff } (v,u,\ov{\alpha})\in E
\]
Bidirected graphs can be seen as a natural lifting of undirected graphs to the labeled setting where reachability is expressed wrt a Dyck language.
Indeed, for every path $P\colon u\Path v$, the reverse path $\ov{P}$ satisfies $\Label(\ov{P})=\Label(P)$.
Thus, Dyck reachability is an equivalence relation on bidirected graphs, much like plain reachability on undirected graphs.
Observe that the same holds when $\Alphabet$ is the disjoint union of two Dyck alphabets, and $G$ is an interleaved Dyck graph.
\cref{fig:graphs} shows a bidirected Dyck graph.
As we mostly deal with bidirected graphs in this paper, we will define them and depict them with every edge appearing in only one direction, with the inverse direction and label taken implicitly.

\Paragraph{Irreducible paths.}
We now introduce the notion of irreducible paths, which is helpful in analyzing bidirected graphs.
A path $P\colon u\Path v$ is \emph{reducible}, if there exist $j_1<j_2$ such that 
(i)~$P[j_1\colon j_2]$ is a cyclic sub-path of $P$, and 
(ii) $\Stack_i(P[:j_1])=\Stack_i(P[:j_2])$ for each $i\in[2]$.
In this case, we can simplify $P$ by the path $P'=P\setminus P[j_1\colon j_2]$, as $P'\colon u\Path v$ is a valid path and each stack has the same contents at the end of $P'$ as at the end of $P$.
Finally, we call $P$ \emph{irreducible} if it is not reducible.
Without loss of generality, throughout this work we consider that paths witnessing reachability are irreducible paths.


\Paragraph{Interleaved Dyck reachability problems.}
In this work we study interleaved Dyck reachability problems on bidirected Dyck graphs $G$, over a corresponding interleaved Dyck language $\Dyck_{k_1}\odot \Dyck_{k_2}$.
As standard in the literature, we take that $k_1,k_2$ are constant and independent of $G$.
We will distinguish cases when $k_1$ and $k_2$ are unrestricted, as well as special cases when $k_1=1$ and $k_2$ is unrestricted, and when $k_1=k_2=1$.
In the unrestricted case, (i.e., when $k_1, k_2>1$), for ease of presentation we assume that $k_1=k_2=k$, and write
$\Dyck_{k}\odot \Dyck_{k}$.

\begin{remark}[Sparsity of $G$]\label{rem:bounded_degree}
Without loss of generality, we assume that $G$ is sparse. i.e.,  $|E|=O(|V|)$.
Indeed, if any node $x$ has two  outgoing edges with the same closing label $x\DTo{\ov{\alpha}}y$ and $x\DTo{\ov{\alpha}}z$, then $z\DTo{\alpha}x\DTo{\ov{\alpha}}y$ and hence $y$ and $z$ are reachable from each other and thus can be merged.
By applying this merging process repeatedly, we arrive at a graph that is sparse, as every node has a bounded number of outgoing edges.
The total time of this process is nearly linear~\cite{Chatterjee18}.
\end{remark}

%% file: figures/graphs.tex
\begin{figure}
\newcommand{\xtstep}{0.9}
\newcommand{\ytstep}{0.5}
\newcommand{\xstep}{1.5}
\newcommand{\ystep}{1.1}

\centering
\begin{tikzpicture}[thick,
pre/.style={<-, thick, shorten >=0cm,shorten <=0cm},
post/.style={->, thick, shorten >=-0cm,shorten <=-0cm},
seqtrace/.style={->, line width=2},
postpath/.style={->, thick, decorate, decoration={zigzag,amplitude=1pt,segment length=2mm,pre=lineto,pre length=2pt, post=lineto, post length=4pt}},
node/.style={thick, draw=black, circle, inner sep = 0, minimum size=5mm, fill=gray!20},
]

\begin{scope}[shift={(0,0)}]

\node[node] (u) at (0*\xstep, 0*\ystep) {$u$};
\node[node] (v) at (2.5*\xstep, 1*\ystep) {$v$};
\node[node] (x) at (1*\xstep, 0.75*\ystep) {$x$};
\node[node] (y) at (2.1*\xstep, -0.7*\ystep) {$y$};

\draw[post, loop above, out=110, in=70, looseness=10] (u) to node[above]{\MyColorOne{$\alpha_1$}} (u);
\draw[post] (u) to node[above]{\MyColorTwo{$\beta_1$}} (x);
\draw[post, loop above, out=110, in=70, looseness=10] (x) to node[above]{\MyColorTwo{$\beta_2$}} (x);
\draw[post] (x) to node[above]{\MyColorOne{$\ov{\alpha}_2$}} (v);
\draw[post, bend left=30] (v) to node[right]{\MyColorOne{$\ov{\alpha}_1$}} (y);
\draw[post, bend left=0] (y) to node[left]{\MyColorTwo{$\ov{\beta}_1$}} (v);

\draw[post, bend left=0] (x) to node[right]{\MyColorOne{$\alpha_2$}} (y);
\draw[post, bend left=30] (y) to node[left]{\MyColorTwo{$\ov{\beta}_2$}} (x);

\end{scope}

\begin{scope}[shift={(10,1)}]

\node[node] (u) at (0*\xstep, 0*\ystep) {$u$};
\node[node] (y) at (-1*\xstep, -1*\ystep) {$x$};
\node[node] (v) at (1*\xstep, 0*\ystep) {$v$};
\node[node] (x) at (2*\xstep, -1*\ystep) {$y$};

\draw[post, bend left=20] (v) to node[above]{\MyColorOne{$\ov{\alpha}_2$}} (x);
\draw[post, bend left=20] (x) to node[below]{\MyColorOne{$\alpha_2$}} (v);

\draw[post, bend left=20] (u) to node[above]{} (v);
\draw[post, bend left=20] (v) to node[below]{} (u);

\draw[post, bend left=20] (u) to node[below]{\MyColorOne{$\ov{\alpha}_1$}} (y);
\draw[post, bend left=20] (y) to node[above]{\MyColorOne{$\alpha_1$}} (u);

\end{scope}

\end{tikzpicture}
\caption{
\emph{(Left):~}An interleaved Dyck graph.
We have $u\DODOPath v$, via the path $u\DTo{\alpha_1}u \DTo{\beta_1} x\DTo{\beta_2} x \DTo{\alpha_2} y \DTo{\ov{\beta}_2} x \DTo{\ov{\alpha}_2} v \DTo{\ov{\alpha}_1} y \DTo{\ov{\beta}_1}v$.
\emph{(Right):~} A bidirected Dyck graph.
We omit the $\epsilon$ label from the edges.
}
\label{fig:graphs}
\end{figure}
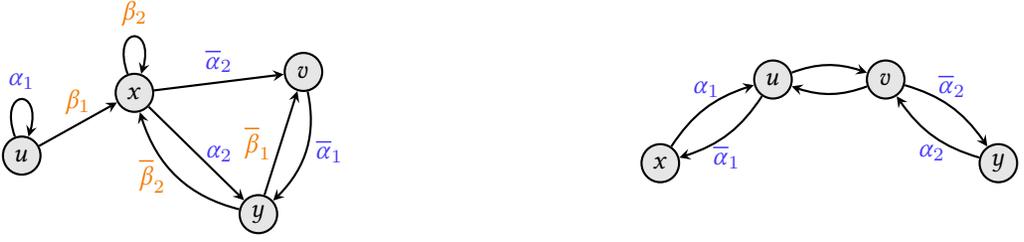

%% file: d1d1.tex
\section{A Fast Algorithm for $\Dyck_1\odot \Dyck_1$ Reachability}\label{sec:d1d1}

In this section we deal with $\Dyck_1\odot \Dyck_1$, i.e., when both Dyck languages are over a unary alphabet.

\Paragraph{Shallow paths.}
The key insight behind the algorithm of~\cite{Li2021} is reachable nodes exhibit a ``shallow-paths'' property, which states that wlog, along every witness path one of the two counters is linearly bounded.
We restate the property here.
\begin{restatable}[Shallow Paths~\cite{Li2021}]{lemma}{lemd1d1shallowpaths}\label{lem:d1d1shallow_paths}
For any nodes $u,v\in V$ such that $v$ is $\Dyck_1\odot \Dyck_1$-reachable from $u$, 
there exists a path $P\colon u\DODOPath v$ such that
for every prefix $P'\in \Prefixes(P)$, we have $\Counter_i(P')\leq 6\cdot n$ for some $i\in [2]$.
\end{restatable}
We call any path $P$ that satisfies the conditions of \cref{lem:d1d1shallow_paths} a \emph{shallow path}.
The shallow-path property was exploited in~\cite{Li2021} to compute $\Dyck_1\odot \Dyck_1$-reachability in $O(n^7)$ time.
The key insight behind our faster algorithm is a stronger property that we call \emph{bounded paths}.

\Paragraph{Bounded paths.}
The bounded-paths property states that wlog, each counter is \emph{always} bounded by $O(n^2)$ along any witness path.
Formally, we have the following lemma.
\begin{restatable}[Bounded Paths]{lemma}{lemd1d1boundedpaths}\label{lem:d1d1bounded_paths}
For any nodes $u,v\in V$ such that $v$ is $\Dyck_1\odot \Dyck_1$-reachable from $u$, 
there exists a witness path $P\colon u\DODOPath v$ such that $\MaxCounter_i(P)\leq 18\cdot n^2 + 6\cdot n$ for each $i\in[2]$, where $n$ is the number of nodes in $G$.
\end{restatable}

\input{figures/counter_bound}

We remark that the $O(n^2)$ bound of \cref{lem:d1d1bounded_paths} is tight:
\cref{fig:counter_bound} shows a simple example where in every $P\colon u\DODOPath v$ path, each counter reaches $\Theta(n^2)$.
In the following, we first present an efficient algorithm for $D_1\odot D_1$ based on \cref{lem:d1d1bounded_paths}, and then prove the lemma.

\input{d1d1_algorithm}

\input{d1d1_bounded_paths}

%% file: figures/counter_bound.tex
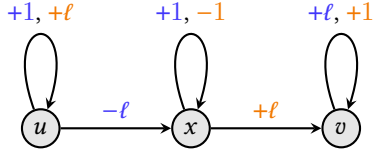
\begin{figure}
\newcommand{\xtstep}{0.9}
\newcommand{\ytstep}{0.5}
\newcommand{\xstep}{2}
\newcommand{\ystep}{0.9}

\centering
\begin{tikzpicture}[thick,
pre/.style={<-, thick, shorten >=0cm,shorten <=0cm},
post/.style={->, thick, shorten >=-0cm,shorten <=-0cm},
seqtrace/.style={->, line width=2},
postpath/.style={->, thick, decorate, decoration={zigzag,amplitude=1pt,segment length=2mm,pre=lineto,pre length=2pt, post=lineto, post length=4pt}},
node/.style={thick, draw=black, circle, inner sep = 0, minimum size=5mm, fill=gray!20},
]

\begin{scope}[shift={(0,0)}]

\node[node] (u) at (0*\xstep, 0*\ystep) {$u$};
\node[node] (x) at (1*\xstep, 0*\ystep) {$x$};
\node[node] (v) at (2*\xstep, 0*\ystep) {$v$};

\draw[post] (u) to node[above]{\MyColorOne{$-\ell$}} (x);
\draw[post] (x) to node[above]{\MyColorTwo{$+\ell$}} (v);
\draw[post, loop above, out=110, in=70, looseness=20] (u) to node[above]{\MyColorOne{$+1$}, \MyColorTwo{$+\ell$}} (u);
\draw[post, loop above, out=110, in=70, looseness=20] (v) to node[above]{\MyColorOne{$+\ell$}, \MyColorTwo{$+1$}} (v);
\draw[post, loop above, out=110, in=70, looseness=20] (x) to node[above]{\MyColorOne{$+1$}, \MyColorTwo{$-1$}} (x);

\end{scope}

\end{tikzpicture}
\caption{
A simple graph where every $P\colon u\DODOPath v$ path has both counters reaching $\Omega(n^2)$.
Edges indicate paths with the respective effect on the counters.
In every such path each counter reaches value $\ell^2$, where we take $\ell=(n-3)/4=\Theta(n)$.
}
\label{fig:counter_bound}
\end{figure}

%% file: d1d1_algorithm.tex
\subsection{Algorithm for $\Dyck_1\odot \Dyck_1$ Reachability}\label{subsec:d1d1_algo}

The bounded-paths property implies a straightforward algorithm to solve $\Dyck_1\odot \Dyck_1$ reachability.
Let $c$ be the counter bound of the bounded-paths lemma.
We  flatten the graph on one counter, by associating each node with all possible values of that counter up to this bound, while replacing all edge labels of that counter with $0$ (i.e., no effect).
This transformation reduces bidirected $\Dyck_1\odot \Dyck_1$ reachability to bidirected $\Dyck_1$ reachability,
which is solved in almost linear-time on bidirected graphs~\cite{Chatterjee18}.
See \cref{algo:d1d1} for details.


\input{algorithms/algo_d1d1}


\begin{proof}[Proof of \cref{thm:d1d1_upper_bounds}]
We start with the correctness.
For any path $P\colon u\Path v$ in $G$ with $\MaxCounter_1(P)\leq c$, there exists a path $P'\colon (u,0)\Path (v, \ell)$ in $G'$, such that $\ell=\Counter_1(P)$ and $\Counter_2(P)=\Counter_2(P')$.
By \cref{lem:d1d1bounded_paths}, if a node $v$ is $\Dyck_1\odot\Dyck_1$-reachable from a node $u$ in $G$, then there exists a witness path $P$ with $\MaxCounter_j(P)\leq c$ for each $j\in[2]$.
Thus there exists a corresponding path $P'\colon (u,0)\Path (v,0)$ in $G'$, which means that $(u,0)$ and $(v,0)$ are in the same reachability class of $G'$.
The correctness hence follows from the correctness of the algorithm of~\cite{Chatterjee18} for $\Dyck_1$-reachability.

Regarding the complexity,
by \cref{lem:d1d1bounded_paths}, the graph $G'$ has $O(n\cdot c)$ nodes.
By \cref{rem:bounded_degree}, the graph $G$ is sparse, which implies that $G'$ is also sparse, i.e., $G'$ has $O(n\cdot c)$ edges.
By~\cite{Chatterjee18}, solving $\Dyck_{1}$-reachability on $G'$ requires $O(n\cdot c\cdot \alpha (n))=O(n^3\cdot \alpha(n))$ time.
The desired result follows.
\end{proof}

%% file: algorithms/algo_d1d1.tex
\smallskip
\begin{algorithm}
\small
\DontPrintSemicolon
\SetInd{0.3em}{0.3em}
\caption{$\AlgoDODOne$}\label{algo:d1d1}
\SetKwProg{Fn}{Function}{:}{}
\KwIn{
A bidirected $\Dyck_{1}\odot \Dyck_{1}$ graph $G=(V,E, \Alphabet=\Alphabet_1\DUnion\Alphabet_2)$ of $n$ nodes
}
\KwOut{
All-pairs $\Dyck_{1}\odot \Dyck_{1}$ reachability in $G$
}
\BlankLine
Let $c\gets 18\cdot n^2 + 6\cdot n$\\
Construct the bidirected $\Dyck_1$ graph $G'=(V', E', \Alphabet')$ as follows\\
\begin{minipage}{.92\linewidth}
\begin{compactenum}
\item The node set is $V'=V\times (\{0\}\cup [c])$
\item The edge set is such that, for every edge $(u,v,\gamma)\in E$ and $j\in \{0\}\cup [c]$, we have an edge $((u,j_1), (u,j_2), \delta)\in E'$, as follows
\begin{compactenum}
\item If $\gamma\in \Alphabet_1\cup \{\epsilon\}$, then $j_2=j_1$ and $\delta=\gamma$
\item If $\gamma\in \Open(\Alphabet_{2})$ and $j_1<c$, then $j_2=j_1+1$ and $\delta=\epsilon$
\item If $\gamma\in \Close(\Alphabet_{2})$ and $j_1>0$, then $j_2=j_1-1$ and $\delta=\epsilon$\\
\end{compactenum}
\end{compactenum}
\end{minipage}

Solve all-pairs $\Dyck_k$ reachability on $G'$ using~\cite{Chatterjee18}\\
\Return{that $u\DODOPath v$ in $G$ iff $(u,0) \DOPath (v,0)$ in $G'$} 
\end{algorithm}

%% file: d1d1_bounded_paths.tex
\subsection{Bounded Paths in $\Dyck_1\odot \Dyck_1$ Reachability}\label{subsec:d1d1_bounded_paths}
We now turn our attention to the proof of the bounded paths lemma.
This section is somewhat technical and can be be skipped at first, as later sections do not depend on it.
We remark that our main goal is to show a quadratic upper-bound on each counter along a path that witnesses
 $\Dyck_1\odot \Dyck_1$ reachability.
As such, our proof aims at readability, without necessarily establishing the smallest constant factor in this bound.

\Paragraph{Counter indexes.}
Given a path $P$, an integer $i\in[2]$ and a natural number $c\in \Nats$, we define the \emph{counter indexes} of counter $i$ of $P$ on counter value $c$ as the set of indexes of $P$ in which $P$ attains value $c$ on counter $i$.
Formally, we have
\[
\CounterIndex_i^c(P)=\{j\in \{0\} \cup [|P|]\colon \Counter_i(P[\colon j])=c \}
\]
To make our exposition simpler, we focus on the boundedness property of \cref{lem:d1d1bounded_paths} on the first counter.
The case of the second counter is completely symmetric,
while it will become clear in the final step of the proof that both properties can be guaranteed simultaneously (i.e., both counters satisfying the bound of \cref{lem:d1d1bounded_paths}).

\Paragraph{Matching pairs.}
Consider a path $P\colon u\DODOPath v$ in $G$, and let $i_{\max}$ be the first point where $P$ attains its maximum value on the first counter.
That is, let
$i_{max}=\min(\CounterIndex_1^a(P))$, where $a=\MaxCounter_1(P)$.
For every $c\in \{0\}\cup[\MaxCounter_1(P)]$, we define the indexes
\[
l_{c}=\max(\{i\colon i\leq i_{\max} \text{ and } \Counter_1(P[\colon i])=c \})  \text{ and }
r_{c}=\min(\{i\colon i\geq i_{\max} \text{ and }  \Counter_1(P[\colon i])=c \})
\]
i.e., $l_c$ (resp., $r_c$) is the last index before $i_{\max}$ (resp., the first index after $i_{\max}$) in which the first counter of $P$ has value $c$.
Let $x_c$ (resp., $y_c$) be the last node  of $P[\colon l_c]$ (resp., $P[\colon r_c])$, and we call $(x_c, y_c)$ a \emph{matching pair}.
The following lemma states that if a shallow path $P$ reaches a large enough value on the first counter,
then ``on its way up'' it must go through the same pair $(p,q)$ twice, where $p$ is a node and $q$ is the value of the second counter, while the same holds ``on its way down''.
It is a straightforward application of the pigeonhole principle (see \cref{fig:path_deflation} for an illustration).

\begin{restatable}{lemma}{lemmatchingpaircounter}\label{lem:matching_pair_counter}
Let $\gamma=12\cdot n^2 + 6\cdot n$ and $\delta=18\cdot n^2 + 6\cdot n + 1$.
Consider any shallow path $P\colon u\DODOPath v$ such that $\MaxCounter_1(P)\geq \delta$.
Then there exist matching pairs $(x_{c_j}, y_{c_j})_{1\leq j \leq 4}$ such that the following hold.
\begin{compactenum}
\item\label{item:mathcing_pair_counter1} $\gamma \leq c_1 < c_2\leq \delta$ and $\gamma \leq c_4 < c_3\leq  \delta$.
\item\label{item:mathcing_pair_counter2} $x_{c_1}=x_{c_2}$ and $y_{c_3}=y_{c_4}$.
\item\label{item:mathcing_pair_counter3} $\Counter_{2}(P[\colon l_{c_1}])=\Counter_{2}(P[\colon l_{c_2}])$ and 
$\Counter_{2}(P[\colon r_{c_3}])=\Counter_{2}(P[\colon r_{c_4}])$.
\end{compactenum}
\end{restatable}

\input{figures/path_deflation}

\Paragraph{Path deflation.}
Using \cref{lem:matching_pair_counter} we describe a process we call \emph{path deflation}.
Informally, it states that if the first counter exceeds $\delta$ in $P$,
we can construct a path $Q$ in which the second counter stays the same as in $P$,
but the first counter reaches the global maximum of $P$ one less time in $Q$ than in $P$ (hence the path has been deflated).
In particular, we construct $Q$ as
\[
Q=P_1\circ P_3 \circ P_4 \circ \ov{P}_3 \circ P_2 \circ P_3 \circ P_5
\]
where the various sub-paths $P_i$ are defined based on \cref{lem:matching_pair_counter} and are also illustrated in \cref{fig:path_deflation}.
In more detail, the first counter reaches its global maximum in $P$ for the first time while traversing the sub-path $P_3$, while
$P_2$ starts and ends in the same node $x$ with the second counter having the same value $\alpha$, and
$P_4$ starts and ends in the same node $y$ with the second counter having the same value $\beta$.
Thus, we can skip $P_2$ and instead reach and traverse $P_4$, without affecting the second counter.
Skipping $P_2$ avoids increasing the first counter, while traversing $P_4$ decreases the first counter.
Traversing $\ov{P}_3$ (i.e., traversing $P_3$ backwards) brings us back to node $x$ with the second counter having the same value $\alpha$.
At this point we can traverse $P_2$ and $P_3$ and then proceed with $P_5$, i.e., proceed as in $P$ but skip the already traversed sub-path $P_4$.
This rearrangement has the effect that none of the traversals of $P_3$ in $Q$ reaches the global maximum of $P$.
Indeed, in the first traversal the counter has decreased by  $c_2-c_1>0$;
in the second traversal it has decreased by $c_2-c_1 + c_3-c_4>0$; and
in the third traversal it has decreased by a $c_3-c_4>0$.
Moreover, as $c_1, c_4\geq \gamma$ while $c_2-c_1<\gamma/2$ and $c_3-c_4<\gamma/2$, while the first counter has decreased, it remains non-negative in these traversals, and thus $Q$ is a valid path.
The following lemma states this property formally.

\begin{restatable}{lemma}{lempathdeflation}\label{lem:path_deflation}
Let $\delta=18\cdot n^2 + 6\cdot n + 1$.
Assume that there exists a shallow path $P\colon u\DODOPath v$ such that $\MaxCounter_1(P)\geq \delta$.
Then there exists a shallow path $Q\colon u\DODOPath v$ such that the following hold.
\begin{compactenum}
\item\label{item:path_deflation1} 
$\MaxCounter_2(Q)=\MaxCounter_2(P)$.
\item\label{item:path_deflation2}
$|\CounterIndex_1^c(Q)|<|\CounterIndex_1^c(P)|$, where $c=\MaxCounter_1(P)$.
\end{compactenum}
\end{restatable}

Note that if the path $Q$ of \cref{lem:path_deflation} also has $\MaxCounter_1(Q)\geq \delta$,
we can apply the lemma recursively on $Q$.
This recursive process is guaranteed to terminate as the new path that comes out of the lemma either has a smaller global maximum on the first counter, or has the same global maximum but this maximum is attained one time less.
Hence we have a finite number of applications of the lemma.
Given this observation, we are now ready to conclude the proof of the bounded-paths lemma.

\begin{proof}[Proof of \cref{lem:d1d1bounded_paths}]
Consider that there is a path $P\colon u\DODOPath v$.
Due to \cref{lem:d1d1shallow_paths}, we may assume wlog that $P$ is a shallow path.
If $\MaxCounter_1(P)\leq 18\cdot n^2 + 6\cdot n$, we proceed with the second counter.
Otherwise, we apply repeatedly \cref{lem:path_deflation} until we arrive at a path $Q'$ with 
$\MaxCounter_2(Q')=\MaxCounter_2(P)$ and 
$\MaxCounter_1(Q')\leq 18\cdot n^2+6\cdot n$.
Note that $Q'$ remains a shallow path throughout this process.
Finally, we follow the same process for $Q'$ instead of $P$, and with the two counters swapped.
In the end, we arrive at a path $Q$ with $\MaxCounter_i(Q)\leq 18\cdot n^2 + 6\cdot n$.
The desired result follows.
\end{proof}

%% file: figures/path_deflation.tex
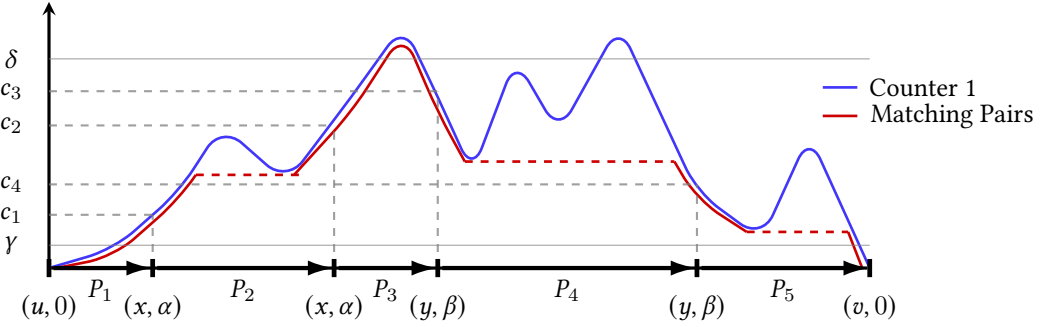
\begin{figure}
\newcommand{\xtstep}{0.5}
\newcommand{\ytstep}{0.2}
\newcommand{\xstep}{1.15}
\newcommand{\ystep}{0.25}
\def\ybias{0.1}
\def\xbias{0.2}
\def\opac{0.75}
\def\counterthickness{1}
\def\opac{0.75}

\newcommand{\MySubPath}[5]{
\node[] (1#5) at (#1) {};
\node[] (2#5) at (#2) {};
\draw[-{Latex[width=2mm]}, ultra thick, draw=#4, shorten >=-4pt,shorten <=-2pt] (1#5) to node[below]{\textcolor{#4}{#3}} (2#5);
}

\centering
\begin{tikzpicture}[thick,
pre/.style={<-, thick, shorten >=0cm,shorten <=0cm},
post/.style={-{Latex[width=3mm]}, thick, shorten >=-0cm,shorten <=-0cm},
seqtrace/.style={->, line width=2},
postpath/.style={->, thick, decorate, decoration={zigzag,amplitude=1pt,segment length=2mm,pre=lineto,pre length=2pt, post=lineto, post length=4pt}},
node/.style={thick, draw=black, circle, inner sep = 0, minimum size=4.75mm, fill=gray!20},
marked/.style={thick, draw=black, fill=black, rectangle, inner sep = 0, minimum width=1pt, minimum height=6pt},
]

\begin{scope}[shift={(0,0)}]

\node[] (1) at (0*\xstep, 0*\ystep) {};
\node[marked] (u) at (0*\xstep, 0*\ystep) {};
\node[marked] (v) at (9.5*\xstep, 0*\ystep) {};
\node[] (xc) at (1.5*\xstep, 0*\ystep) {};
\node[] (xcp) at (3.5*\xstep, 0*\ystep) {};
\node[] (ycp) at (4.5*\xstep, 0*\ystep) {};
\node[] (yc) at (7.5*\xstep, 0*\ystep) {};

\node[marked] (xc1) at (1.2*\xstep, 0*\ystep) {};
\node[marked] (xc2) at (3.3*\xstep, 0*\ystep) {};
\node[marked] (yc3) at (4.5*\xstep, 0*\ystep) {};
\node[marked] (yc4) at (7.5*\xstep, 0*\ystep) {};

\draw[post, -, very thick, shorten <=-1mm] (1) to (v);

\node[below=\ybias of u] {$(u,0)$};
\node[below=\ybias of v] {$(v,0)$};
\node[below=\ybias of xc1] {$(x,\alpha)$};
\node[below=\ybias of xc2] {$(x, \alpha)$};
\node[below=\ybias of yc3] {$(y, \beta)$};
\node[below=\ybias of yc4] {$(y, \beta)$};

\MySubPath{u}{xc1}{$P_1$}{black}{P1}
\MySubPath{xc1}{xc2}{$P_2$}{black}{P2}
\MySubPath{xc2}{yc3}{$P_3$}{black}{P3}
\MySubPath{yc3}{yc4}{$P_4$}{black}{P4}
\MySubPath{yc4}{v}{$P_5$}{black}{P5}

\begin{pgfonlayer}{bg}

\def\xshift{0.1}
\def\yshift{0.1}
\def\pairscolor{\darkred}

\begin{scope}[shift={(-1.2,0)}]
\draw[-, line width=\counterthickness, draw=\ColorOne] (10*\xstep, 9.5*\ystep) to (10.4*\xstep, 9.5*\ystep);
\node[ align=left] at (11.25*\xstep + -0.1, 9.5*\ystep) {Counter~1};
\draw[-, line width=\counterthickness, draw=\pairscolor] (10*\xstep, 8*\ystep) to (10.4*\xstep, 8*\ystep);
\node[] at (11.15*\xstep+0.4, 8*\ystep) {Matching Pairs};
\end{scope}

\draw[-, line width=\counterthickness, rounded corners=10, draw=\ColorOne] (0*\xstep, 0*\ystep) to ($ (xc) + (-0.75*\xstep, 1*\ystep) $) to ($ (xc) + (0, 4*\ystep) $) to ($ (xc) + (0.5*\xstep, 7.5*\ystep) $) to ($ (xc) + (1.25*\xstep, 4.5*\ystep) $) to ($ (xcp) + (0*\xstep, 9*\ystep) $) to ($ (xcp) + (0.6*\xstep, 12.8*\ystep) $) to ($ (ycp) + (0*\xstep, 9*\ystep) $) to ($ (ycp) + (0.4*\xstep, 5*\ystep) $) to ($ (ycp) + (0.9*\xstep, 11*\ystep) $) to ($ (ycp) + (1.4*\xstep, 7.1*\ystep) $) to ($ (ycp) + (2.1*\xstep, 12.8*\ystep) $) to ($ (yc) + (0*\xstep, 4*\ystep) $) to ($ (yc) + (0.75*\xstep, 1.5*\ystep) $) to ($ (yc) + (1.3*\xstep, 7*\ystep) $) to ($ (v) + (0*\xstep, 0*\ystep) $);

\draw[-, line width=\counterthickness, rounded corners=10, draw=\pairscolor] (0*\xstep+\xshift, 0*\ystep) to ($ (xc) + (-0.75*\xstep, 1*\ystep-\yshift) $) to ($ (xc) + (0, 4*\ystep-\yshift) $) to ($ (xc) + (0.2*\xstep, 4.5*\ystep + \yshift) $);

\draw[-, line width=\counterthickness, rounded corners=10, draw=\pairscolor, dashed] ($ (xc) + (0.2*\xstep, 4.5*\ystep + \yshift) $) to  ($ (xc) + (1.25*\xstep+2*\xshift, 4.5*\ystep + \yshift) $);

\draw[-, line width=\counterthickness, rounded corners=10, draw=\pairscolor] ($ (xc) + (1.25*\xstep+\xshift, 4.5*\ystep + \yshift) $) to ($ (xcp) + (0*\xstep, 9*\ystep-2*\yshift) $) to ($ (xcp) + (0.6*\xstep, 12.8*\ystep-\yshift) $) to ($ (ycp) + (0*\xstep, 9*\ystep-2*\yshift) $) to ($ (ycp) + (0.4*\xstep-\xshift, 5*\ystep+1.5*\yshift) $);

\draw[-, line width=\counterthickness, rounded corners=10, draw=\pairscolor, dashed] ($ (ycp) + (0.4*\xstep-\xshift, 5*\ystep+1.5*\yshift) $) to ($ (yc) + (0*\xstep-3*\xshift, 5*\ystep+1.5*\yshift) $);



\draw[-, line width=\counterthickness, rounded corners=10, draw=\pairscolor]  ($ (yc) + (0*\xstep-3*\xshift, 5*\ystep+1.5*\yshift) $) to ($ (yc) + (0*\xstep, 4*\ystep - \yshift) $) to ($ (yc) + (0.75*\xstep-2*\xshift, 1.5*\ystep + \yshift) $);

\draw[-, line width=\counterthickness, rounded corners=10, draw=\pairscolor, dashed]   ($ (yc) + (0.75*\xstep-2*\xshift, 1.5*\ystep + \yshift) $) to ($ (yc) + (1.75*\xstep-0*\xshift, 1.5*\ystep + \yshift) $);

\draw[-, line width=\counterthickness, rounded corners=10, draw=\pairscolor]   ($ (yc) + (1.75*\xstep-0*\xshift, 1.5*\ystep + \yshift) $) to ($ (v) + (-0\xshift, 0*\ystep + 0*\yshift) $);

\node[] (gamma) at (0*\xstep-\xtstep, 1.2*\ystep) {$\gamma$};
\node[] (delta) at (0*\xstep-\xtstep, 11*\ystep) {$\delta$};
\node[] (c1) at (0*\xstep-\xtstep, 2.8*\ystep) {$c_1$};
\node[] (c2) at (0*\xstep-\xtstep, 7.5*\ystep) {$c_2$};
\node[] (c3) at (0*\xstep-\xtstep, 9.3*\ystep) {$c_3$};
\node[] (c4) at (0*\xstep-\xtstep, 4.4*\ystep) {$c_4$};

\draw[->, very thick] (0*\xstep, 0*\ystep) to (0*\xstep, 14*\ystep);

\draw[-, dashed, gray, thick, opacity=\opac] (xc1|- 0,0) to (xc1|- c1);
\draw[-, dashed, gray, thick, opacity=\opac] (c1 -| 0,0) to (xc1|- c1);

\draw[-, dashed, gray, thick, opacity=\opac] (xc2|- 0,0) to (xc2|- c2);
\draw[-, dashed, gray, thick, opacity=\opac] (c2 -| 0,0) to (xc2|- c2);

\draw[-, dashed, gray, thick, opacity=\opac] (yc3|- 0,0) to (yc3|- c3);
\draw[-, dashed, gray, thick, opacity=\opac] (c3 -| 0,0) to (yc3|- c3);

\draw[-, dashed, gray, thick, opacity=\opac] (yc4|- 0,0) to (yc4|- c4);
\draw[-, dashed, gray, thick, opacity=\opac] (c4 -| 0,0) to (yc4|- c4);

\draw[-, gray , opacity=\opac] (gamma -| 0,0) to (v|- gamma);

\draw[-, gray, opacity=\opac] (delta -| 0,0) to (v|- delta);


\end{pgfonlayer}

\end{scope}

\end{tikzpicture}
\caption{
Illustration of \cref{lem:matching_pair_counter}.
The x-axis shows the decomposition of the path $P\colon u\DODOPath v$ into paths $P=P_1\circ P_2 \circ P_3 \circ P_4 \circ P_5$.
The pairs $(p,q)$ show the node $p$ and the value $q$ of the second counter along various segments of $P$.
The y-axis shows the value of the first counter along $P$ (blue) as well as its value on the matching pairs of $P$ (red).
}
\label{fig:path_deflation}
\end{figure}

%% file: dkd1.tex
\section{Upper and Lower Bounds for $\Dyck_k\odot \Dyck_1$ Reachability}\label{sec:dkd1}

In this section we address $\Dyck_k\odot \Dyck_1$ reachability, i.e., where one of the two stacks behaves as a counter.
In \cref{subsec:dkd1_upper_bound} we prove the decidability of the problem (\cref{thm:dkd1_decidable})
In \cref{subsec:dkd1_bounded} we turn our attention to the counter-bounded version of the problem, and prove
 \cref{thm:dkd1_quadratic_upper_bound} and \cref{thm:dkd1_quadratic_lower_bound}.

\input{dkd1_general}
\input{dkd1_bounded}

%% file: dkd1_general.tex
\input{dkd1_upper_bound}

%% file: dkd1_upper_bound.tex
\subsection{Decidability of $\Dyck_k\odot \Dyck_1$ Reachability}\label{subsec:dkd1_upper_bound}

Consider an interleaved bidirected Dyck graph $G=(V,E,\Alphabet=\Alphabet_k\DUnion \Alphabet_1)$ and two nodes $u,v\in V$.
We show that, deciding whether there is a path $P\colon u\DKDKPath v$ is decidable.
Our decidability proof makes use of a relaxed notion of reachability in vector addition systems, known as coverability.

\Paragraph{Coverability.}
Consider two nodes $u,v\in V$ and a non-negative integer $c\in \Nats$.
The node $u$ \emph{covers} $(v,c)$ if there is a valid path $P\colon u\Path v$ with $\Stack(P)=\epsilon$ and $\Counter(P)\geq c$.
In other words, we can reach $v$ from $u$ via a path that is balanced wrt the stack, but the counter might be positive (instead of zero).
The coverability problem is to decide whether $u$ covers $(v,c)$, and is known to be decidable even on non-bidirected interleaved Dyck graphs (phrased in the language of pushdown vector addition systems \cite{Leroux2015}).
We start with a straightforward lemma.

\begin{restatable}{lemma}{lemreachimpliescover}\label{lem:reach_implies_cover}
If $u\DKDOPath v$ then $u$ covers $(v,0)$ and $v$ covers $(u,0)$.
\end{restatable}

\Paragraph{Intuitive description.}
\cref{lem:reach_implies_cover} states a necessary condition for the reachability of $v$ from $u$.
Hence, as a first step, in order to decide reachability, we may verify that $u$ covers $(v,0)$ and $v$ covers $(u,0)$ using the procedure of~\cite{Leroux2015}.
Next, we can similarly check if $u$ covers $(v,1)$ and $v$ covers $(u,1)$.
Clearly, if $u$ does not cover $(v,1)$, since $u$ covers $(v,0)$, we have that $v$ is reachable from $u$ and we are done.
We arrive at a similar conclusion if $v$ does not cover $(u,1)$.
But what if $u$ covers $(v,1)$ and $v$ covers $(u,1)$?
The key insight is that, using the paths $P_u\colon u\Path v$ and $P_v\colon v\Path u$ that witness the corresponding coverability relationships, we can derive a bound on the length of the shortest path $L\colon u\DKDOPath v$ that may witness reachability.
\cref{algo:dkd1} presents this algorithm.

\input{algorithms/algo_dkd1}

The key idea behind the bounded length of $L$ stems from the fact that, by traversing the cyclic path $P_u\circ P_v$ repeatedly, we can increase the value of the counter arbitrarily without increasing the maximum stack height.
Now, starting from $u$ with a large enough counter and an empty stack, we can reach $v$ with the same counter value and an empty stack, while the maximum stack height of this $u\Path v$ path is bounded.
Now repeating the process symmetrically from $v$, we end up with a path $T\colon u\DKDOPath v$ in which the stack height remains bounded.
The final step is to show that wlog, any path that has bounded stack height also has bounded length, thus arriving at the path $L$.
As there are finitely many paths of bounded length, the decidability follows.
We use the remaining of this section to develop this intuition precisely.

\Paragraph{Matching pairs.}
We revisit the notion of matching pairs from \cref{subsec:d1d1_bounded_paths}, and adapt it to our current setting where we have a stack instead of a counter.
In particular, consider a path $P\colon u\DKPath v$ in $G\Project \Alphabet_k$, and let $i_{\max}$ be the first point where $P$ attains its maximum stack height.
For every $h\in [\MaxStackHeight(P)]$, let
\[
l_{h}=\max(\{i\colon i\leq i_{\max} \text{ and } \StackHeight(P[\colon i])=h \}) \quad \text{and}\quad
r_{h}=\min(\{i\colon i\geq i_{\max} \text{ and }  \StackHeight(P[\colon i])=h \})
\]
i.e., $l_h$ (resp., $r_h$) is the last index before $i_{\max}$ (resp., the first index after $i_{\max}$) in which the stack of $P$ has size $h$.
Let $x_h$ (resp., $y_h$) be the last node  of $P[\colon l_h]$ (resp., $P[\colon r_h])$, and we call $(x_h, y_h)$ a \emph{matching pair}.
We have the following straightforward lemma.

\begin{restatable}{lemma}{lemmatchingpair}\label{lem:matching_pair}
Consider a path $P\colon u\DKPath v$ in $G\Project \Alphabet_k$.
If $\MaxStackHeight(P)\geq n^2$, then $P$ has two matching pairs $(x_h, y_h)$ and $(x_{h'}, y_{h'})$ such that
(i)~$h< h'\leq n^2$, (ii)~$x_h=x_{h'}$, and 
(iii)~$y_h=y_{h'}$.
\end{restatable}

Based on \cref{lem:matching_pair} we prove the following lemma.

\begin{restatable}{lemma}{lemboundedstackheight}\label{lem:bounded_stack_height}
Consider a path $P\colon u\DKPath v$ in $G\Project \Alphabet_k$.
There is a path $Q\colon u\DKPath v$ in $G\Project \Alphabet_k$ such that
(i)~$\MaxStackHeight(Q)\leq 2\cdot n^2$, and 
(ii)~$\Counter(Q)=\Counter(P)$.
\end{restatable}

\input{figures/bounded_stack_height}

The key idea behind \cref{lem:bounded_stack_height} is as follows.
Assume that $\MaxStackHeight(P)\geq 2\cdot n^2$.
We apply \cref{lem:matching_pair} to obtain two matching pairs $(x_h, y_h)$ and $(x_{h'}, y_{h'})$.
We decompose $P$ as $P=P_1\circ P_2\circ P_3 \circ P_4 \circ P_5$ as follows (see \cref{fig:bounded_stack_height}).
\begin{compactenum}
\item $P_1$ is the prefix of $P$ up to $x_h$.
\item $P_2$ is the sub-path $x_h\Path x_{h'}$ of $P$.
\item $P_3$ is the sub-path $x_{h'}\Path y_{h'}$ of $P$.
\item $P_4$ is the sub-path $y_{h'}\Path y_{h}$ of $P$.
\item $P_5$ is the suffix of $P$ from $x_h$ on.
\end{compactenum}
Moreover, we let $K$ be a shortest path $K\colon x_h\DKPath y_h$, and it is known that $\MaxStackHeight(K)\leq n^2$~\cite{Pierre1992}.
We rearrange $P$ to obtain the path
$R=P_1\circ P_2\circ K \circ P_4 \circ \ov{K} \circ P_3\circ P_5$.
Now $R$ does not reach the maximum stack height that $P$ reaches while traversing $P_3$,
as, in the beginning of $P_3$, the stack height has decreased from $h'$ to $h$.
Moreover, all the way up to $P_3$, the stack height of $R$ is (strictly) below $h'+n^2\leq 2\cdot n^2$.
Finally, since $R$ traverses every edge of $P$ exactly once, the two paths executed in $G$ have the same counter value at the end.
\cref{lem:bounded_stack_height} is obtained by repeated applications this process until we end up with a path $Q$ as stated in the lemma.

Note that the above process only concerns the stack height. Indeed, if we consider the counter along $Q$, then $Q$ is not necessarily a valid path as the counter may become negative because of the repeated rearrangements.
Assume, however, that we also have a cyclic path $F_u\colon u\Path u$ that ends with an empty stack and increases the counter by a positive amount.
Then we may prefix $Q$ by iterating $F_u$ a number of $\ell_u$ times until the counter becomes large enough to stay non-negative along $Q$.
Now the corresponding path $F_u^{\ell_u}\circ Q$ is a valid path, but no longer witnesses the $\Dyck_k\odot \Dyck_1$ reachability of $v$ from $u$, as the counter at the end of the path equals $\Counter(F_u^{\ell_u})>0$.
However, if we have a similar path $F_v\colon v\Path v$ from $v$, we can follow the process symmetrically on the side of $v$.
In the end, we construct a path $T$ that is a reachability witness as
\[
T=F_u^{\ell_u}\circ Q \circ F_v^{\ell_v}\circ \ov{Q} \circ \ov{F}^{\ell_u}_u \circ Q\circ \ov{F}^{\ell_v}_v
\]
The crucial observation is that each of the above sub-paths starts and ends with an empty stack.
Hence, the maximum stack height of $T$ is the maximum among the maximum stack heights of these sub-paths.
We thus arrive at the following lemma.

\begin{restatable}{lemma}{lemboundedstackheightwitness}\label{lem:bounded_stack_height_witness}
Assume that there are cyclic paths $F_u\colon u\Path u$ and $F_v\colon v\Path v$ such that $\Stack(F_u)=\Stack(F_v)=\epsilon$, and $\Counter(F_u), \Counter(F_v)>0$.
Let $\zeta=\max(\MaxStackHeight(F_u), \MaxStackHeight(F_v))$.
If $u\DKDOPath v$, then there is a path $T\colon u\DKDOPath v$ such that
$\MaxStackHeight(T)\leq \max(\zeta, 2\cdot n^2)$.
\end{restatable}

In particular, the bound $\zeta$ comes from the stack height of the circular paths $F_u$ and $F_v$,
while the bound $2\cdot n^2$ comes from the stack height of $Q$, following~\cref{lem:bounded_stack_height}.
At this point, a natural question is whether a bound on the maximum stack height of reachability witnesses bounds the search space for a witness.
The following lemma shows that such a bound implies a bound on the length of the witness, and thus answers this question in positive.

\begin{restatable}{lemma}{lemboundedlengthpath}\label{lem:bounded_length_path}
Assume that there is a path $P\colon u\DKDOPath v$ with $\MaxStackHeight(P)\leq \delta$, for some $\delta\in \Nats$.
Then there is a path $Q\colon u\DKDOPath v$ with $|Q|\leq n^3\cdot k^{3\cdot \delta}$.
\end{restatable}

We finally have all the ingredients to prove \cref{thm:dkd1_decidable}.

\begin{proof}[Proof of \cref{thm:dkd1_decidable}.]
We argue about the correctness of \cref{algo:dkd1}.
Clearly, if $u$ does not cover $(v,0)$ or $v$ does not cover $(u,0)$, due to \cref{lem:reach_implies_cover} the algorithm returns $\False$ correctly in \cref{line:algo_dkd1_return_false}.
On the other hand, if $u$ does not cover $(v,1)$, then the path that witnesses the coverability of $(v,0)$ from $u$ also witnesses reachability.
Similarly if $v$ does not cover $(u,1)$, and thus the algorithm returns $\True$ correctly in \cref{line:algo_dkd1_return_true}.

Now consider the case that $u$ covers $(v,1)$ and $v$ covers $(u,1)$, and let $P_u\colon u\Path v$ and $P_v\colon v\Path u$ be the corresponding witness paths.
Then we have paths $F_u=P_u\circ P_v$ and $F_v=P_v\circ P_u$ with $\Stack(F_u)=\Stack(F_v)=\epsilon$ and $\Counter(F_u), \Counter(F_v)\geq 0$.
By \cref{lem:bounded_stack_height_witness}, if $v$ is $\Dyck_k\odot\Dyck_1$-reachable from $u$ then there is a path $T\colon u\Path v$ such that
$\MaxStackHeight(T)\leq\delta= \max(\zeta, 2\cdot n^2)$.
Finally, \cref{lem:bounded_length_path} applies, and thus there is  a path $Q\colon u\DKDOPath v$ with $|Q|\leq n^3\cdot k^{3\cdot \delta}$.
\cref{algo:dkd1} iterates over all such paths $Q$ in \cref{line:algo_dkd1_loop}, and tests whether any of them witnesses the reachability of $v$ from $u$, returning $\True$ if such a path is found, and $\False$ otherwise.
The desired result follows.
\end{proof}

%% file: algorithms/algo_dkd1.tex
\smallskip
\begin{algorithm}
\small
\DontPrintSemicolon
\SetInd{0.3em}{0.3em}
\caption{$\AlgoDKDOne$}\label{algo:dkd1}
\SetKwProg{Fn}{Function}{:}{}
\KwIn{
An interleaved bidirected Dyck graph $G=(V,E, \Alphabet=\Alphabet_2\DUnion\Alphabet_1)$, two nodes $u,v\in V$
}
\KwOut{True iff $u\DKDKPath v$}
\BlankLine
\uIf{$u$ does not cover $(v,0)$ or $v$ does not cover ($u, 0)$}{
\Return{$\False$}\label{line:algo_dkd1_return_false}
}
\uIf{$u$ does not cover $(v,1)$ or $v$ does not cover ($u, 1)$}{
\Return{$\True$}\label{line:algo_dkd1_return_true}
}
Let $P_u\gets$ a witness path for the coverability of $(v,1)$ from $u$\\
Let $P_v\gets$ a witness path for the coverability of $(u,1)$ from $v$\\
Let $\zeta\gets \max(\MaxStackHeight(P_u), \MaxStackHeight(P_v))$\\
Let $\delta\gets \max(\zeta, 2\cdot n^2)$\\
\ForEach{path $L\colon u\Path v$ with $|P|\leq n^3\cdot k^{3\cdot \delta}$}{\label{line:algo_dkd1_loop}
\uIf{$L$ witnesses $u\DKDKPath v$}{
\Return{$\True$}
}
}
\Return{$\False$}
\end{algorithm}

%% file: figures/bounded_stack_height.tex
\begin{figure}
\newcommand{\xtstep}{0.5}
\newcommand{\ytstep}{0.2}
\newcommand{\xstep}{1.15}
\newcommand{\ystep}{0.25}
\def\ybias{0.1}
\def\xbias{0.2}
\def\opac{0.75}
\def\counterthickness{1}
\def\opac{0.75}

\newcommand{\MySubPath}[5]{
\node[] (1#5) at (#1) {};
\node[] (2#5) at (#2) {};
\draw[-{Latex[width=2mm]}, ultra thick, draw=#4, shorten >=-4pt,shorten <=-2pt] (1#5) to node[below]{\textcolor{#4}{#3}} (2#5);
}

\centering
\begin{tikzpicture}[thick,
pre/.style={<-, thick, shorten >=0cm,shorten <=0cm},
post/.style={-{Latex[width=3mm]}, thick, shorten >=-0cm,shorten <=-0cm},
seqtrace/.style={->, line width=2},
postpath/.style={->, thick, decorate, decoration={zigzag,amplitude=1pt,segment length=2mm,pre=lineto,pre length=2pt, post=lineto, post length=4pt}},
node/.style={thick, draw=black, circle, inner sep = 0, minimum size=4.75mm, fill=gray!20},
marked/.style={thick, draw=black, fill=black, rectangle, inner sep = 0, minimum width=1pt, minimum height=6pt},
]

\begin{scope}[shift={(0,0)}]

\node[] (1) at (0*\xstep, 0*\ystep) {};
\node[marked] (u) at (0*\xstep, 0*\ystep) {};
\node[marked] (v) at (9.5*\xstep, 0*\ystep) {};
\node[marked] (xc) at (1.5*\xstep, 0*\ystep) {};
\node[marked] (xcp) at (3.5*\xstep, 0*\ystep) {};
\node[marked] (ycp) at (4.5*\xstep, 0*\ystep) {};
\node[marked] (yc) at (7.5*\xstep, 0*\ystep) {};

\draw[post, -, very thick, shorten <=-1mm] (1) to (v);

\node[below=\ybias of u] {$u$};
\node[below=\ybias of v] {$v$};
\node[below=\ybias of xc] {$x_h$};
\node[below=\ybias of xcp] {$x_{h'}$};
\node[below=\ybias of yc] {$y_h$};
\node[below=\ybias of ycp] {$y_{h'}$};

%
%
%
%
\MySubPath{u}{xc}{$P_1$}{black}{P1}
\MySubPath{xc}{xcp}{$P_2$}{black}{P2}
\MySubPath{xcp}{ycp}{$P_3$}{black}{P3}
\MySubPath{ycp}{yc}{$P_4$}{black}{P4}
\MySubPath{yc}{v}{$P_5$}{black}{P5}

\begin{pgfonlayer}{bg}

\def\xshift{0.1}
\def\yshift{0.1}
\def\pairscolor{\darkred}

\begin{scope}[shift={(-1.2,0)}]
\draw[-, line width=\counterthickness, draw=\ColorOne] (10*\xstep, 9.5*\ystep) to (10.4*\xstep, 9.5*\ystep);
\node[ align=left] at (11.25*\xstep+0.14, 9.5*\ystep) {Stack Height};
\draw[-, line width=\counterthickness, draw=\pairscolor] (10*\xstep, 8*\ystep) to (10.4*\xstep, 8*\ystep);
\node[] at (11.15*\xstep+0.4, 8*\ystep) {Matching Pairs};
\end{scope}

\draw[-, line width=\counterthickness, rounded corners=10, draw=\ColorOne] (0*\xstep, 0*\ystep) to ($ (xc) + (-0.75*\xstep, 1*\ystep) $) to ($ (xc) + (0, 4*\ystep) $) to ($ (xc) + (0.5*\xstep, 7.5*\ystep) $) to ($ (xc) + (1.25*\xstep, 4.5*\ystep) $) to ($ (xcp) + (0*\xstep, 9*\ystep) $) to ($ (xcp) + (0.6*\xstep, 12.8*\ystep) $) to ($ (ycp) + (0*\xstep, 9*\ystep) $) to ($ (ycp) + (0.4*\xstep, 5*\ystep) $) to ($ (ycp) + (0.9*\xstep, 11*\ystep) $) to ($ (ycp) + (1.4*\xstep, 7.1*\ystep) $) to ($ (ycp) + (2.1*\xstep, 12.8*\ystep) $) to ($ (yc) + (0*\xstep, 4*\ystep) $) to ($ (yc) + (0.75*\xstep, 1.5*\ystep) $) to ($ (yc) + (1.3*\xstep, 7*\ystep) $) to ($ (v) + (0*\xstep, 0*\ystep) $);

\draw[-, line width=\counterthickness, rounded corners=10, draw=\pairscolor] (0*\xstep+\xshift, 0*\ystep) to ($ (xc) + (-0.75*\xstep, 1*\ystep-\yshift) $) to ($ (xc) + (0, 4*\ystep-\yshift) $) to ($ (xc) + (0.2*\xstep, 4.5*\ystep + \yshift) $);

\draw[-, line width=\counterthickness, rounded corners=10, draw=\pairscolor, dashed] ($ (xc) + (0.2*\xstep, 4.5*\ystep + \yshift) $) to  ($ (xc) + (1.25*\xstep+2*\xshift, 4.5*\ystep + \yshift) $);

\draw[-, line width=\counterthickness, rounded corners=10, draw=\pairscolor] ($ (xc) + (1.25*\xstep+\xshift, 4.5*\ystep + \yshift) $) to ($ (xcp) + (0*\xstep, 9*\ystep-2*\yshift) $) to ($ (xcp) + (0.6*\xstep, 12.8*\ystep-\yshift) $) to ($ (ycp) + (0*\xstep, 9*\ystep-2*\yshift) $) to ($ (ycp) + (0.4*\xstep-\xshift, 5*\ystep+1.5*\yshift) $);

\draw[-, line width=\counterthickness, rounded corners=10, draw=\pairscolor, dashed] ($ (ycp) + (0.4*\xstep-\xshift, 5*\ystep+1.5*\yshift) $) to ($ (yc) + (0*\xstep-3*\xshift, 5*\ystep+1.5*\yshift) $);



\draw[-, line width=\counterthickness, rounded corners=10, draw=\pairscolor]  ($ (yc) + (0*\xstep-3*\xshift, 5*\ystep+1.5*\yshift) $) to ($ (yc) + (0*\xstep, 4*\ystep - \yshift) $) to ($ (yc) + (0.75*\xstep-2*\xshift, 1.5*\ystep + \yshift) $);

\draw[-, line width=\counterthickness, rounded corners=10, draw=\pairscolor, dashed]   ($ (yc) + (0.75*\xstep-2*\xshift, 1.5*\ystep + \yshift) $) to ($ (yc) + (1.75*\xstep-0*\xshift, 1.5*\ystep + \yshift) $);

\draw[-, line width=\counterthickness, rounded corners=10, draw=\pairscolor]   ($ (yc) + (1.75*\xstep-0*\xshift, 1.5*\ystep + \yshift) $) to ($ (v) + (-0\xshift, 0*\ystep + 0*\yshift) $);

\node[] at (0*\xstep-\xtstep, 4*\ystep) {$h$};
\node[] at (0*\xstep-\xtstep, 9*\ystep) {$h'$};

\draw[->, very thick] (0*\xstep, 0*\ystep) to (0*\xstep, 14*\ystep);

\draw[-, dashed, gray, thick, opacity=\opac] (xc|- 0,0) to (xc|- 0,4*\ystep);
\draw[-, dashed, gray, thick, opacity=\opac] (yc|- 0,0) to (yc|- 0,4*\ystep);
\draw[-, dashed, gray, thick, opacity=\opac] (xcp|- 0,0) to (xcp|- 0,9*\ystep);
\draw[-, dashed, gray, thick, opacity=\opac] (ycp|- 0,0) to (ycp|- 0,9*\ystep);

\draw[-, dashed, gray, thick, opacity=\opac] (yc|- 0,4*\ystep) to (0,4*\ystep);
\draw[-, dashed, gray, thick, opacity=\opac] (ycp|- 0,9*\ystep) to (0,9*\ystep);

\draw[->, bend left=50, thick, dashed, shorten >=1mm,shorten <=1mm] (xcp) to node[above]{$K$} (ycp);
\draw[<-, bend right=10, thick, dashed, rounded corners, shorten >=1mm,shorten <=1mm] ($ (xcp) + (0, 0*\ystep) $) to ($ (xcp) + (3*\xbias, -3*\ystep) $) to node[below]{$\ov{K}$} ($ (yc) + (-3*\xbias, -3*\ystep) $) to (yc);


\end{pgfonlayer}

\end{scope}

\end{tikzpicture}
\caption{
Illustration of the path rearrangement behind \cref{lem:bounded_stack_height}.
}
\label{fig:bounded_stack_height}
\end{figure}
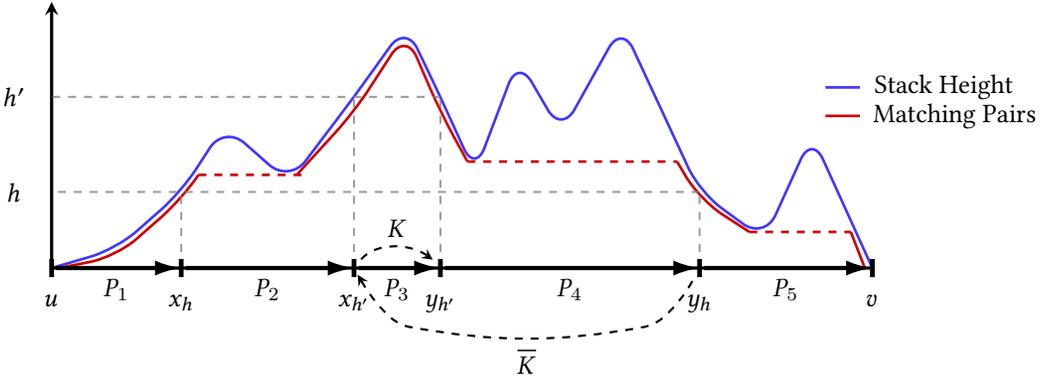

%% file: dkd1_bounded.tex
\subsection{Bounded-Counter $\Dyck_k\odot \Dyck_1$ Reachability}\label{subsec:dkd1_bounded}

In this section we focus on a bounded version of $\Dyck_k\odot\Dyck_1$ reachability.
The goal is to determine all nodes $u,v$ such that $v$ is $\DKDOPath$ reachable from $u$ via a path $P$ such that $\MaxCounter(P)=O(n)$.
We first describe the algorithm for the quadratic upper bound (\cref{thm:dkd1_quadratic_upper_bound}),
and then prove that the bound is optimal (\cref{thm:dkd1_quadratic_lower_bound}).

\Paragraph{Upper bound.}
The algorithm for the upper bound follows the counter-flattening idea of \cref{thm:d1d1_upper_bounds}.
In particular, we use an algorithm identical to \cref{algo:d1d1}, with the exception that the counter bound is $c=O(n)$.
Note that the algorithm of~\cite{Chatterjee18} solves bidirected $\Dyck_{k}$ reachability, i.e., when $k\geq 2$ in general, and thus applies to this setting as well.

\begin{proof}[Proof of \cref{thm:dkd1_quadratic_upper_bound}]
The correctness follows straightforwardly.
For the complexity, observe that the graph $G'$ has $O(n\cdot c)=O(n^2)$ nodes.
By \cref{rem:bounded_degree}, the graph $G$ is bounded-degree, which implies that $G'$ is also bounded-degree, and hence $G'$ has $O(n^2)$ edges.
By~\cite{Chatterjee18}, solving $\Dyck_{k}$-reachability on $G'$ requires $O(n^2\cdot \alpha (n))$ time.
The desired result follows.
\end{proof}

The simple algorithm behind \cref{thm:dkd1_quadratic_upper_bound} leads to the natural question of whether a more involved algorithm can achieve a better bound, i.e., below quadratic.
We next establish \cref{thm:dkd1_quadratic_lower_bound}, showing that, in fact, no algorithm can bring the complexity below quadratic, under standard complexity-theoretic hypotheses.
To show this, we establish a fine-grained reduction from the problem of orthogonal vectors.

\Paragraph{Orthogonal vectors ($\OV$).}
The input to $\OV$ is two sets $X,Y$, each containing $n$ vectors in $\{0,1\}^{D}$, for some dimension $D=\omega(\log n)$.
The task is to determine if there exists a pair $(x_i,y_j)\in X\times Y$ that is orthogonal, i.e., for each $\ell\in [D]$, we have $x_i[\ell]\cdot y_j[\ell] = 0$.
The respective hypothesis states that the problem cannot be solved in time $O(n^{2-\epsilon})$, for any fixed $\epsilon>0$~\cite{Williams19}.

\input{figures/ov_simple}

\Paragraph{Reduction.}
Given an instance $(X,Y)$ of $\OV$,
we construct an interleaved and bidirected Dyck graph $G=(V,E,\Alphabet=\Alphabet_1\DUnion\{+1,-1 \})$ such that for two distinguished nodes $u,v\in V$ we have a path $P\colon u\DKDOPath v$ with $\MaxCounter(P)=O(n)$ iff there exists an orthogonal pair $(x_i,y_j)\in X\times Y$.

\SubParagraph{Intuition.}
Before presenting the construction we provide some intuition.
To develop some insight, we first consider a simple reduction of $\OV$ to plain Dyck reachability (i.e., we don't have a counter) for non-bidirected graphs (see \cref{fig:ov_simple}).
Starting from $u$, we have $n$ parallel paths to an intermediate node $q$.
Along the $i$-th such path, we push on the stack the encoding of the vector $x_i$ using the symbols $\gamma_0$ and $\gamma_1$ in the straightforward way.
A similar encoding for the vectors in $Y$ suffices, where we match the contents of the stack by traversing a path that corresponds to a vector $y_j$.
If $y_{j}[\ell]=1$, then we can only move forward if $x_i[\ell]=0$, thus we have a labeled transition that pops $\gamma_0$ from the stack.
If $y_{j}[\ell]=0$, then we expect either $x_i[\ell]=0$ or $x_i[\ell]=1$, thus we have two parallel labeled transitions, the first popping $\gamma_0$ and the second popping $\gamma_1$.

The above construction works for Dyck reachability on non-bidirected graphs, but fails for bidirected graphs.
The issue is that, if $y_{j}[\ell]=0$, then we can follow the edge that pops $\gamma_1$, and then follow the other edge (i.e., the one that pops $\gamma_0$) in inverse direction, which has the effect of pushing $\gamma_0$ on the stack.
But now the stack encodes a vector $x$ that is not part of $X$, which affects correctness.
Note, however, that the above construction reduces $\OV$ to Dyck reachability, as opposed to interleaved Dyck reachability
(i.e., it does not make use of the counter).

To alleviate the bidirectedness problem in interleaved Dyck reachability, we make use of the counter to force that certain edges are traversed only one way.
Similarly as before, to encode $y_j[\ell]$, we have an edge popping $\gamma_0$.
However, when $y_j[\ell]=0$, in order to pop $\gamma_1$ we have to traverse a counter gadget that increases the counter by $2^{\ell}$.
This implies that this path cannot be traversed backwards, as this would require to decrease the counter by $2^{\ell}$, which is a value that cannot have been accumulated so far (the current counter value can be at most $\sum_{l<\ell}2^l=2^{\ell}-1$).
In the end, we self-loop on $v$ to reduce the counter to $0$.
See \cref{fig:ov} for an illustration.

\input{figures/ov}

\SubParagraph{Formal construction.}
We now present the formal construction.
\begin{compactenum}
\item We have special nodes $\{u,w,v\}\cup \bigcup_{i\in [n], j\in [D+1]} \{x_i^{\ell}, y_i^{\ell} \}$, as well as $n$ copies of the counter gadgets $H_{2^{\ell}}$, for every $\ell\in [D]$.
\item For every $i\in[n]$, we have edges $u\DTo{}x_i^1$, $x_i^{D+1}\DTo{}w$, $w\DTo{}y_i^{D+1}$ and $y_i^{1}\DTo{}v$.
\item For every $i\in[n]$ and $\ell\in [D]$, we have $x_i^\ell\DTo{\gamma_j}x_i^{\ell+1}$, where $j=x_i[\ell]$.
\item For every $i\in[n]$ and $\ell\in [D]$, we have $y_{i+1}^\ell\DTo{\ov{\gamma}_0}y_i^{\ell}$, where $j=y_i[\ell]$.
Moreover, if $y_i[\ell]=0$, we have an edge $y_i^{\ell}\DTo{\ov{\gamma}_1}p_i^{\ell}$ and an edge $q_i^{\ell}\DTo{}y_{i+1}$,
where $p_i^{\ell}$ and $q_i^{\ell}$ are the entry and exit nodes, respectively, of the $i$-th copy of the counter gadget $H_{2^{\ell}}$.
\item Finally, we have a self-loop $v\DTo{-1}v$.
\end{compactenum}

\Paragraph{Correctness.}
We now turn our attention to the correctness of the construction.
It is clear that if there are $x_i$ and $y_j$ that are orthogonal, we have a path $P\colon u\DKDOPath v$ by traversing the corresponding nodes $x_i^{\ell}$ and $y_{j}^{\ell}$, and end the path by looping on $v$.
Now assume that there exists a (irreducible) path $P\colon u\DKDOPath v$.
Observe that the path goes from $u$ to $w$ by pushing on the stack the contents of a vector $x_i$.
At this point the stack dictates how the path can continue from $y_j^{\ell}$ to $y_{j}^{\ell+1}$.
Moreover, the value of the counter while the path is at $y_j^{\ell}$ is bounded by $\sum_{l<\ell}2^{l}=2^{\ell}-1$.
Hence, although the path can transition back from $y_j^{\ell}$ to $y_{j}^{\ell-1}$, it can only do so by pushing $\gamma_0$ on the stack.
This has the effect that if the path ever returns to $w$, its stack will encode a vector (possibly not in $X$) that has at least $1$ in the coordinates $\ell$ for which $x_i[\ell]=1$.
It follows that when $P$ reaches $v$, it has traversed the nodes $x_i^{\ell}$ and $y_{j}^{\ell}$, for some $i,j\in[n]$, such that $x_i$ and $y_j$ are orthogonal.
We thus have the following lemma.

\begin{restatable}{lemma}{lemovcorrectness}\label{lem:ov_correctness}
There is a path $P\colon u\DKDOPath v$ iff there is an orthogonal pair $(x_i, y_j)\in X\times Y$.
\end{restatable}

We are now ready to conclude \cref{thm:dkd1_quadratic_lower_bound}.

\begin{proof}[Proof of \cref{thm:dkd1_quadratic_lower_bound}]
\cref{lem:ov_correctness} proves the correctness of the construction, so it remains to argue about the complexity.
The number of nodes in $G$ is $O(n\cdot D^2)$, as we have $O(n\cdot D)$ nodes $x_i^{\ell}$ and $y_i^{\ell}$, while each counter gadget $H_{2^{\ell}}$ uses $\ell$ nodes.
Since $\ell\leq D$, we have $O(n\cdot D^2)$ nodes in total.
Finally, observe that $G$ is sparse, and thus there are $O(n\cdot D^2)$ edges as well.
The desired result follows.
\end{proof}

%% file: figures/ov_simple.tex
\begin{figure}
\newcommand{\xtstep}{0.9}
\newcommand{\ytstep}{0.5}
\newcommand{\xstep}{1.5}
\newcommand{\ystep}{0.7}

\centering
\begin{tikzpicture}[thick,
pre/.style={<-, thick, shorten >=0cm,shorten <=0cm},
post/.style={->, thick, shorten >=-0cm,shorten <=-0cm},
seqtrace/.style={->, line width=2},
postpath/.style={->, thick, decorate, decoration={zigzag,amplitude=1pt,segment length=2mm,pre=lineto,pre length=2pt, post=lineto, post length=4pt}},
node/.style={thick, draw=black, circle, inner sep = 0, minimum size=4.75mm, fill=gray!20},
font={\small},
]

\begin{scope}[shift={(0,0)}]

\node[node] (u) at (0*\xstep, 0*\ystep) {$u$};
\node[node] (x11) at (1*\xstep, 1*\ystep) {};
\node[node] (x12) at (2*\xstep, 1*\ystep) {};
\node[node] (x13) at (3*\xstep, 1*\ystep) {};
\node[node] (x21) at (1*\xstep, -1*\ystep) {};
\node[node] (x22) at (2*\xstep, -1*\ystep) {};
\node[node] (x23) at (3*\xstep, -1*\ystep) {};
\node[node] (q) at (4*\xstep, 0*\ystep) {$w$};
\node[node] (y11) at (5*\xstep, 1*\ystep) {};
\node[node] (y12) at (6*\xstep, 1*\ystep) {};
\node[node] (y13) at (7*\xstep, 1*\ystep) {};
\node[node] (y21) at (5*\xstep, -1*\ystep) {};
\node[node] (y22) at (6*\xstep, -1*\ystep) {};
\node[node] (y23) at (7*\xstep, -1*\ystep) {};

\node[] at (2*\xstep, 2.5*\ystep) {$x_1=11$};
\node[] at (2*\xstep, -2.5*\ystep) {$x_2=01$};

\node[node] (t) at (8*\xstep, 0*\ystep) {$v$};

\draw[post,bend left=30] (u) to (x11);
\draw[post,bend right=30] (u) to (x21);
\draw[post,bend left=30] (q) to (y11);
\draw[post,bend right=30] (q) to (y21);

\draw[post,bend left=30] (x13) to (q);
\draw[post, bend right=30] (x23) to (q);
\draw[post,bend left=30] (y13) to (t);
\draw[post, bend right=30] (y23) to (t);

\draw[post] (x11) to node[above]{$\MyColorOne{\gamma_1}$} (x12);
\draw[post] (x12) to node[above]{$\MyColorOne{\gamma_1}$} (x13);

\draw[post] (x21) to node[below]{$\MyColorOne{\gamma_0}$} (x22);
\draw[post] (x22) to node[below]{$\MyColorOne{\gamma_1}$} (x23);

\draw[post, bend right=20] (y11) to node[below]{$\MyColorOne{\ov{\gamma}_0}$} (y12);
\draw[post, bend right=20] (y12) to node[below]{$\MyColorOne{\ov{\gamma}_0}$} (y13);
\draw[post, bend left=20] (y12) to node[above]{$\MyColorOne{\ov{\gamma}_1}$} (y13);

\draw[post, bend right=20] (y21) to node[below]{$\MyColorOne{\ov{\gamma}_0}$} (y22);
\draw[post, bend left=20] (y21) to node[above]{$\MyColorOne{\ov{\gamma}_1}$} (y22);
\draw[post, bend right=20] (y22) to node[below]{$\MyColorOne{\ov{\gamma}_0}$} (y23);

\node[] at (6*\xstep, 2.5*\ystep) {$y_1=01$};
\node[] at (6*\xstep, -2.5*\ystep) {$y_2=10$};

\end{scope}

\end{tikzpicture}
\caption{
A simple reduction from $\OV$ to Dyck reachability on non-bidirected graphs.
}
\label{fig:ov_simple}
\end{figure}
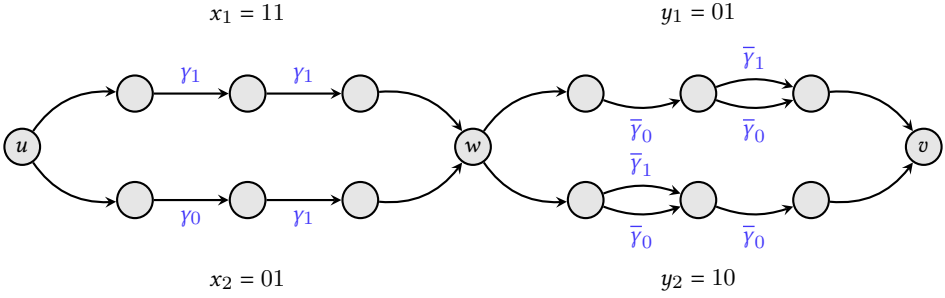

%% file: figures/ov.tex
\begin{figure}
\newcommand{\xtstep}{0.9}
\newcommand{\ytstep}{0.5}
\newcommand{\xstep}{1.3}
\newcommand{\ystep}{0.65}
\newcommand{\xxstep}{0.5}

\centering
\begin{tikzpicture}[thick,
pre/.style={<-, thick, shorten >=0cm,shorten <=0cm},
post/.style={->, thick, shorten >=-0cm,shorten <=-0cm},
seqtrace/.style={->, line width=2},
postpath/.style={->, thick, decorate, decoration={zigzag,amplitude=1pt,segment length=2mm,pre=lineto,pre length=2pt, post=lineto, post length=4pt}},
node/.style={thick, draw=black, circle, inner sep = 0, minimum size=4.75mm, fill=gray!20},
counter/.style={rectangle, draw=black, inner sep=0, minimum width=25, minimum height=5mm, fill=gray!0, rounded corners},
font={\small},
]

\begin{scope}[shift={(0,0)}]

\node[node] (u) at (0*\xstep, 0*\ystep) {$u$};
\node[node] (x11) at (1*\xstep, 1*\ystep) {};
\node[node] (x12) at (2*\xstep, 1*\ystep) {};
\node[node] (x13) at (3*\xstep, 1*\ystep) {};
\node[node] (x21) at (1*\xstep, -1*\ystep) {};
\node[node] (x22) at (2*\xstep, -1*\ystep) {};
\node[node] (x23) at (3*\xstep, -1*\ystep) {};
\node[node] (q) at (4*\xstep, 0*\ystep) {$w$};
\node[node] (y11) at (5*\xstep, 1*\ystep) {};
\node[node] (y12) at (6*\xstep+\xxstep, 1*\ystep) {};
\node[node] (y13) at (7*\xstep+2*\xxstep, 1*\ystep) {};
\node[node] (y21) at (5*\xstep, -1*\ystep) {};
\node[node] (y22) at (6*\xstep+\xxstep, -1*\ystep) {};
\node[node] (y23) at (7*\xstep+2*\xxstep, -1*\ystep) {};
\node[counter] (counter11) at (5.5*\xstep+\xxstep/2, 2.5*\ystep) {\MyColorTwo{$H_{2^1}$}};
\node[counter] (counter12) at (6.5*\xstep+1.5*\xxstep, 2.5*\ystep) {\MyColorTwo{$H_{2^2}$}};
\node[counter] (counter21) at (5.5*\xstep+\xxstep/2, -2.5*\ystep) {\MyColorTwo{$H_{2^1}$}};
\node[counter] (counter22) at (6.5*\xstep+1.5*\xxstep, -2.5*\ystep) {\MyColorTwo{$H_{2^2}$}};


\node[] at (2*\xstep, 2.5*\ystep) {$x_1=11$};
\node[] at (2*\xstep, -2.5*\ystep) {$x_2=01$};

\node[node] (t) at (8*\xstep+\xxstep, 0*\ystep) {$v$};

\draw[post,bend left=30] (u) to (x11);
\draw[post,bend right=30] (u) to (x21);
\draw[post,bend left=30] (q) to (y11);
\draw[post,bend right=30] (q) to (y21);

\draw[post,bend left=30] (x13) to (q);
\draw[post, bend right=30] (x23) to (q);
\draw[post,bend left=30] (y13) to (t);
\draw[post, bend right=30] (y23) to (t);

\draw[post] (x11) to node[above]{$\MyColorOne{\gamma_1}$} (x12);
\draw[post] (x12) to node[above]{$\MyColorOne{\gamma_1}$} (x13);

\draw[post] (x21) to node[below]{$\MyColorOne{\gamma_0}$} (x22);
\draw[post] (x22) to node[below]{$\MyColorOne{\gamma_1}$} (x23);

\draw[post, bend right=0] (y12) to node[below]{$\MyColorOne{\ov{\gamma}_1}$} (y13);
\draw[post, bend right=0] (y11) to node[left]{$\MyColorOne{\ov{\gamma}_0}$} (counter11);
\draw[post, bend right=0] (counter11) to node[]{} (y12);
\draw[post, bend right=0] (y12) to node[left]{$\MyColorOne{\ov{\gamma}_0}$} (counter12);
\draw[post, bend right=0] (counter12) to node[]{} (y13);

\draw[post, bend right=0] (y21) to node[above]{$\MyColorOne{\ov{\gamma}_1}$} (y22);
\draw[post, bend left=0] (y21) to node[left]{$\MyColorOne{\ov{\gamma}_0}$} (counter21);
\draw[post, bend left=0] (counter21) to node[below]{} (y22);
\draw[post, bend left=0] (y22) to node[left]{$\MyColorOne{\ov{\gamma}_0}$} (counter22);
\draw[post, bend left=0] (counter22) to node[below]{} (y23);

%

\draw[post, loop right, out=25, looseness=10] (t) to node[right]{\MyColorTwo{$-1$}} (t);

\node[] at (6*\xstep+\xxstep, 3.5*\ystep) {$y_1=01$};
\node[] at (6*\xstep+\xxstep, -3.5*\ystep) {$y_2=10$};

\end{scope}

\end{tikzpicture}
\caption{
Reduction from $\OV$ to $\Dyck_k\odot \Dyck_1$-reachability with linearly bounded counter.
A path $P_x\colon u\Path w$ pushes on the stack an encoding of the vector $x\in X$ using symbols $\gamma_0, \gamma_1$.
A path $P_y\colon w\Path v$ pops from the stack any vector that is orthogonal with $y$.
When $P_y$ pops a $\gamma_0$ in the $\ell$-th coordinate, it increases the counter by $2^{\ell}$.
In the end, the counter is emptied by self-looping on $v$.
}
\label{fig:ov}
\end{figure}

%% file: dkdk.tex
\section{Undecidability of $D_k\odot D_k$ Reachability}\label{sec:dkdk}

Finally, in this section we prove the undecidability of general $D_k\odot D_k$ reachability on bidirected graphs (\cref{thm:dkdk_undecidable}).
Our reduction is from $D_k\odot D_k$ reachability on non-bidirected graphs, which is known to be undecidable~\cite{Reps00}.

\Paragraph{Reduction.}
Consider a non-bidirected graph $G=(V,E,\Alphabet)$ where $\Alphabet=\Alphabet_1\DUnion\Alphabet_2$,
and the problem of  $D_k\odot D_k$ reachability wrt the alphabets $\Alphabet_1$ and $\Alphabet_2$ on two nodes $u$ and $v$ of $G$.
We assume wlog that $G$ is not a multi-graph, i.e., every pair of nodes has a unique edge between them.
We construct a bidirected graph $G'=(V', E', \Alphabet')$ where $V\subset V'$ and $\Alphabet'=\Alphabet'_1\DUnion\Alphabet'_2$,
and such that a node $t\in V'$ is $D_k\odot D_k$-reachable from node $s\in V'$, wrt the alphabets $\Alphabet'_1$ and $\Alphabet'_2$, iff $v$ is reachable from $u$ in $G$.
To keep the exposition simple, the size of $\Alphabet'$ is proportional to the size of $G$.
Standard constructions can turn $\Alphabet'$ to constant, though this comes at the expense of increasing the graph size by a logarithmic factor (see, e.g.,~\cite{Chistikov2021}).
Since we show undecidability, this increase is not a concern.

\Paragraph{Intuitive description.}
We start with the intuition behind the reduction, while we refer to \cref{fig:undecidability} for illustrations on a small example.
Given $G$ and nodes $u$ and $v$, we construct $G'$ such that there is a path $P\colon u\Path v$ in $G$
iff there exists a path $P'\colon s\Path t$ in $G'$, as follows.
Initially, the path $P'$ ``guesses'' the path $P$ edge-by-edge in reverse order, while it
(i)~pushes every guess in the first stack, and
(ii)~uses the second stack to simulate the behavior of the second stack along $P$ in $G$ (in reverse), mirrored under the parenthesis complementation operator $\ov{\gamma}$
(i.e., if the edge in $G$ pushes $\gamma$, the simulating edge will pop $\gamma$).
Note that complementing the letters as found in the reverse order of $P$, in fact simulates the second stack in the forward order of $P$.

\input{figures/undecidability}

After this part is done, the second stack is empty, while the first stack contains the sequence of edges of $P$ in top-down order (i.e., the last edge of $P$ is at the top of the stack).
At this point, $P'$ will traverse the remaining of the graph $G'$ and verify that the path $P$ stored in first stack is a valid path of $G$, while using the second stack to simulate the effect of $P$ on the first stack.
As a preliminary step, $P'$ is forced to push a special symbol  $\nu$ on the second stack as it takes a transition from $s$ to $u$.
This is a simple trick to deal with bidirectedness, in the sense that, if in the future $P'$ decides to come back to $s$ and alter the guessed path, it can only do so by popping $\nu$.
This will imply that the second stack is empty and the first stack contains a suffix of the guessed $P$, and thus $P'$ is reducible.

In order to verify the guess of $P'$, the graph $G'$ is constructed as follows.
For every edge $x\DTo{\gamma}y$ of $G$, we have a two-edged path $x\DTo{(\ov{x,y}), \delta}y $ in $G'$.
If $\gamma\in \Alphabet_1$, i.e., the edge in $G$ manipulates the first stack, $\delta$ is identical to $\gamma$ but affects the second stack of $G'$.
On the other hand, if $\gamma\in \Alphabet_2$, then $\delta=\epsilon$.
Intuitively, we can traverse the path $x\DTo{(\ov{x,y}), \delta}y $ in $G'$ iff the top-most edge guess in the first stack matches the edge $x\DTo{\gamma}y$ of $G$.
By traversing this path, if $\gamma$ manipulates the first stack in $G$, we use $\delta=\gamma$ to manipulate the second stack in $G'$.
On the other hand, if $\gamma$ manipulates the second stack in $G$ we ignore its effect ($\delta=\epsilon$), as the second stack has already been simulated in the first phase where we guessed the path $P$ while self-looping on $s$.

The key idea in the above construction is that while self-looping over $s$, the path $P'$ can only guess the edges of $P$ in the direction they appear in $G$.
Although $P'$ can traverse edges backwards in $G'$ (due to bidirectedness),
this will generally result in pushing more edge symbols $e_1,\dots, e_{\ell}$ on the first stack.
Because of the letter $\nu$, the only way to pop these symbols from the stack is for $P'$ to reverse the sub-path that added these symbols, which makes $P'$ reducible.

\Paragraph{Construction.}
We now proceed with the detailed construction.
For ease of presentation, we specify some parts of $G'$ as paths, meaning unique sequences of labeled edges, without explicit reference to all the nodes and edges in the path.
We begin with the alphabets
\[
\Alphabet'_1 = \bigcup_{x\DTo{\gamma}y\in E}\{ (x,y), (\ov{x,y}) \} \qquad\qquad
\Alphabet'_2 = \Alphabet_1\cup \Alphabet_2\cup \{\nu, \ov{\nu} \}
\]
Let $\Alphabet' = \Alphabet'_1\DUnion \Alphabet'_2$.
We have $V\subset V'$, while $G'$ contains some additional nodes and edges, defined by the following paths.
For every edge $x\DTo{\gamma}y\in E$, we construct a cyclic path $s\DTo{(x,y), \delta} s$,
where $\delta=\ov{\gamma}$ if $\gamma\in \Alphabet_2$ and $\delta=\epsilon$ otherwise.
We also have an edge $s\DTo{\nu}u$.
For every edge $x\DTo{\gamma}y$ of $G$, we have a path $x\DTo{(\ov{x,y}), \delta}y $ in $G'$,
where $\delta=\gamma$ if $\gamma\in \Alphabet_1$ and $\delta=\epsilon$ otherwise.
Finally, we have an edge $v\DTo{\ov{\nu}}t$.

\Paragraph{Correctness.}
It remains to argue about the correctness of the above construction,
i.e., we have $u\DKDKPath v$ in $G$ iff $s\DKDKPath t$ in $G'$.

\begin{proof}[Proof of \cref{thm:dkdk_undecidable}]
We argue separately about completeness and soundness.

\SubParagraph{Completeness.}
Assume that there is a path $P\colon u\DKDKPath v$ in $G$.
We construct a path $P'\colon s\DKDKPath t$ in $G'$ as follows.
Let $P=e_1,\dots, e_{m}$ be the sequence of edges traversed by $P$, without their labels.
For each $i\in [m]$, we take the self-loop path $s\DTo{e_{m-i+1}, \delta}s$.
Afterwards, we traverse the edge $s\DTo{\nu}u$, and then we repeatedly traverse the path $x\DTo{(\ov{x,y}), \delta}y$, where $(x,y)$ is the current top symbol on the first stack.
Finally, we traverse the edge $v\DTo{\ov{\nu}} t$.

Let $P'_1$ be the prefix of $P'$ right before we traverse the edge $s\DTo{\nu}u$.
Observe that 
\[
\Label(P'_1)\Project \Alphabet'_2= \ov{\Label(\ov{P})\Project \Alphabet_2} = \Label(P)\Project \Alphabet_2
\]
and hence $P'_1$ is valid on the second stack.
From this point, It is straightforward to verify by induction that $P'$ is a valid path with $\Stack_1(P')=\Stack_2(P')=\epsilon$, and hence $P'\colon s\DKDKPath t$.

\SubParagraph{Soundness.}
Assume that there is a path $P'\colon s\DKDKPath t$ in $G'$.
Let $P'_1$ be the prefix of $P'$ right before the last time that $P'$ traverses the edge $s\DTo{\nu} u$.
Observe that $\Stack_2(P'_1)=\epsilon$.
Indeed, at that point, the second stack cannot contain any symbol from $\Alphabet_2$, as these symbols are not popped in the suffix of $P'$ that succeeds $P'_1$.
Moreover, the second stack cannot contain any symbol from $\Alphabet_1$, as these symbols can only be pushed in the stack on top of $\nu$, which prevents $P'_1$ to traverse the edge $s\DTo{\nu}v$ in reverse.

We now argue that at the end of $P'_1$, the first stack encodes a path $P\colon u\DKDKPath v$ in $G$.
Indeed, let $P'_2$ be the suffix of $P'$ such that $P'=P'_1\circ P'_2$.
Since $P'$ is an irreducible path, $P'_2$ is also irreducible.
This implies that $P'_2$ never pushes an element on the first stack.
Hence, the whole of $P'_2$ except the last edge $v\DTo{\ov{\nu}} t$ matches the content of the first stack at the end of $P'_1$.
This implies that $\Stack_1(P'_1)$ encodes a path $P\colon u\Path v$ in $G$ such that $\Label(P)\Project \Alphabet_1\in \Dyck(\Alphabet_1)$, i.e., the label of $P$ produces a valid Dyck string wrt the first alphabet.
Finally, since $\Stack_2(P'_1)=\epsilon$, it follows that  $\Label(P)\Project \Alphabet_2\in \Dyck(\Alphabet_2)$, and thus the label of $P$ produces a valid Dyck string wrt the second alphabet as well.
Hence $P$ witnesses the $\Dyck_k\odot\Dyck_k$ reachability of $v$ from $u$.
\end{proof}

%% file: figures/undecidability.tex
\begin{figure}
\newcommand{\xtstep}{0.9}
\newcommand{\ytstep}{0.5}
\newcommand{\xstep}{1.5}
\newcommand{\ystep}{1.8}

\centering
\begin{tikzpicture}[thick,
pre/.style={<-, thick, shorten >=0cm,shorten <=0cm},
post/.style={->, thick, shorten >=-0cm,shorten <=-0cm},
seqtrace/.style={->, line width=2},
postpath/.style={->, thick, decorate, decoration={zigzag,amplitude=1pt,segment length=2mm,pre=lineto,pre length=2pt, post=lineto, post length=4pt}},
node/.style={thick, draw=black, circle, inner sep = 0, minimum size=4.75mm, fill=gray!20},
counter/.style={rectangle, draw=black, inner sep=0, minimum width=27, minimum height=5mm, fill=gray!0, rounded corners},
]

\begin{scope}[shift={(0,0)}]

\node[node] (u) at (0*\xstep, 0*\ystep) {$u$};
\node[node] (x) at (1*\xstep, 0*\ystep) {$x$};
\node[node] (y) at (0*\xstep, -1*\ystep) {$y$};
\node[node] (v) at (1*\xstep, -1*\ystep) {$v$};

\draw[post, bend right=10] (u) to node[below]{\MyColorOne{$\alpha_1$}} (x);
\draw[post, bend right=10] (x) to node[above]{\MyColorTwo{$\beta_1$}} (u);
\draw[post, bend right=0] (u) to node[left]{\MyColorOne{$\alpha_2$}} (y);
\draw[post] (y) to[out=160, in=200, looseness=10] node[left]{ \MyColorTwo{$\ov{\beta}_1$}} (y);
\draw[post, bend right=0] (y) to node[below]{\MyColorOne{$\ov{\alpha}_2$}} (v);
\draw[post] (v) to[out=70, in=110, looseness=10] node[above]{ \MyColorOne{$\ov{\alpha}_1$}} (v);


\end{scope}

\begin{scope}[shift={(6.1*\xstep,0)}]

\renewcommand{\xstep}{2.1}
\renewcommand{\ystep}{2.2}

\node[node] (s) at (-1.6*\xstep, 0*\ystep) {$s$};
\node[node] (u) at (0*\xstep, 0*\ystep) {$u$};
\node[node] (x) at (1*\xstep, 0*\ystep) {$x$};
\node[node] (y) at (0*\xstep, -1*\ystep) {$y$};
\node[node] (v) at (1*\xstep, -1*\ystep) {$v$};
\node[node] (t) at (1.5*\xstep, -1*\ystep) {$t$};

\draw[post] (s) to node[above]{\MyColorTwo{$\nu$}} (u);
\draw[post] (v) to node[above]{\MyColorTwo{$\ov{\nu}$}} (t);

\draw[post, bend right=10] (u) to node[below]{\MyColorOne{$(\ov{u,x})$}, \MyColorTwo{$\alpha_1$}} (x);
\draw[post, bend right=10] (x) to node[above]{\MyColorOne{$(\ov{x,u})$}, \MyColorTwo{$\epsilon$}} (u);
\draw[post, bend right=0] (u) to node[left]{\MyColorOne{$(\ov{u,y})$}, \MyColorTwo{$\alpha_2$}} (y);
\draw[post] (y) to[out=160, in=200, looseness=10] node[left]{\MyColorOne{$(\ov{y,y})$}, \MyColorTwo{$\epsilon$}} (y);
\draw[post, bend right=0] (y) to node[above]{\MyColorOne{$(\ov{y,v})$}, \MyColorTwo{$\ov{\alpha}_2$}} (v);
\draw[post] (v) to[out=70, in=110, looseness=10] node[above]{\MyColorOne{$(\ov{v,v})$}, \MyColorTwo{$\ov{\alpha}_1$}} (v);


\def\inangle{35}
\def\outangle{20}
\def\firstangle{25}

\draw[post] (s) to[out=\firstangle+0*\outangle, in=\firstangle+0*\outangle+1*\inangle, looseness=25] node[above]{\MyColorOne{$(u,x)$}, \MyColorTwo{$\epsilon$}} (s);

\draw[post] (s) to[out=\firstangle+1*\inangle+1*\outangle, in=\firstangle+2*\inangle+1*\outangle, looseness=25] node[above]{\MyColorOne{$(x,u)$}, \MyColorTwo{$\ov{\beta}_1$}} (s);

\draw[post] (s) to[out=\firstangle+2*\inangle+2*\outangle, in=\firstangle+3*\inangle+2*\outangle, looseness=25] node[left]{\MyColorOne{$(u,y)$}, \MyColorTwo{$\epsilon$}} (s);

\draw[post] (s) to[out=\firstangle+3*\inangle+3*\outangle, in=\firstangle+4*\inangle+3*\outangle, looseness=25] node[left]{\MyColorOne{$(y,y)$}, \MyColorTwo{$\beta_1$}} (s);

\draw[post] (s) to[out=\firstangle+4*\inangle+4*\outangle, in=\firstangle+5*\inangle+4*\outangle, looseness=25] node[below]{\MyColorOne{$(y,v)$}, \MyColorTwo{$\epsilon$}} (s);

\draw[post] (s) to[out=\firstangle+5*\inangle+5*\outangle, in=\firstangle+6*\inangle+5*\outangle, looseness=25] node[below]{\MyColorOne{$(v,v)$}, \MyColorTwo{$\epsilon$}} (s);

\end{scope}

\end{tikzpicture}
\caption{
\emph{(Left):}~An interleaved Dyck graph $G$ with a reachability question on $u,v$.
\emph{(Right):}~The bidirected interleaved Dyck graph $G'$ constructed in our reduction, with a reachability question on $s,t$.
}
\label{fig:undecidability}
\end{figure}

%% file: experiments.tex
\section{Experiments}\label{sec:experiments}

In this section we report on the experimental evaluation of our algorithms for $\Dyck_1\odot \Dyck_1$ reachability (\cref{thm:d1d1_upper_bounds}) and $\Dyck_k\odot \Dyck_1$  reachability with linearly-bounded counters (\cref{thm:dkd1_quadratic_upper_bound}).
We first describe some straightforward optimizations to the baseline algorithms (\cref{subsec:experimental_algos}) and then present the experimental results (\cref{subsec:experimental_results}).

\input{experimental_algos}

\input{experimental_results}

%% file: experimental_algos.tex
\subsection{Experimental Algorithms}\label{subsec:experimental_algos}

We describe three straightforward optimizations.

\Paragraph{1.~Under-approximating with $\Dyck_k$.}
Our first optimization is based on the simple observation that $\Dyck_k\odot \Dyck_k$ reachability wrt two alphabets $\Alphabet_1$ and $\Alphabet_2$ can be under-approximated by performing $\Dyck_k$ reachability on the union alphabet $\Alphabet=\Alphabet_1\cup \Alphabet_2$.
Thus, as a first step, we perform bidirected $\Dyck_k$ reachability on the input graph, and reduce the graph by merging pairs that are $\Dyck_k$ reachable.

\input{figures/experimental_algos_hub_removal}

\Paragraph{2.~Removing doubly-self-looped nodes.}
Our final optimization concerns only $\Dyck_1\odot\Dyck_1$ reachability, and is an effective procedure for removing doubly-self-looped nodes. 
Indeed, for any node $x$ that has a self loop on each counter, we can make the following observations.
\begin{compactenum}
\item Any witness path $P\colon u\DODOPath v$ that goes through $x$ implies also the existence of reachability paths $u\DODOPath x$ and $x\DODOPath v$, by self-looping on $x$ a sufficient number of times.
Hence, we can focus on whether $x$ reaches any other node, and if not, remove $x$.
\item For any witness path $Q\colon u\DODOPath x$, wlog $Q$ is a path that never exits $x$, i.e., $Q$ has the form $Q=u\dots x\dots x$.
Hence, in order to decide whether $x$ is reachable from any other node, it suffices to compute $\Dyck_1\odot \Dyck_1$ reachability locally on each connected component that is connected to the rest of the graph only via $x$.
In practice, we have found that when doubly self-looped nodes are present, they are articulation points that separate many small connected components.
\end{compactenum}
Thus, performing the above process on $x$ allows us to either
(i)~merge $x$ with another node that reaches $x$ and repeat, or 
(ii)~remove $x$ from the graph.
In the first case we have reduced the size of the graph, while in the second case we proceed recursively on the small connected components that are created.
\cref{fig:double_self_loop} illustrates this process on a small example. 

\input{figures/experimental_algos_trimming}

\Paragraph{3.~Node trimming.}
Our second optimization is based on identifying simple motifs in the graph which guarantee that a node $x$
(i)~is not reachable from any other node, and
(ii)~if there is a path $P$ that witnesses reachability and goes through $x$, there is a path $Q$ that witnesses the same reachability without going through $x$.
Applying this process repeatedly can prune away many such ``isolated'' nodes and thus reduce the effective size of the graph. 
\cref{fig:trimming} illustrates two simple cases that allow successive trimming on nodes $z$ and $x$.

In our experiments, we found that all heuristics were applicable in all benchmarks.
In particular, steps (1) and (2) were important for significantly reducing the input graph to small, while step (3) helped with running time and also with reducing the graph further in some cases.

%% file: figures/experimental_algos_hub_removal.tex
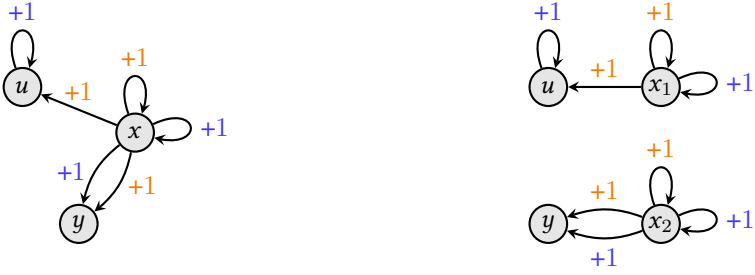
\begin{figure}
\newcommand{\xtstep}{0.9}
\newcommand{\ytstep}{0.5}
\newcommand{\xstep}{1.5}
\newcommand{\ystep}{1.0}

\centering
\begin{tikzpicture}[thick,
pre/.style={<-, thick, shorten >=0cm,shorten <=0cm},
post/.style={->, thick, shorten >=-0cm,shorten <=-0cm},
seqtrace/.style={->, line width=2},
postpath/.style={->, thick, decorate, decoration={zigzag,amplitude=1pt,segment length=2mm,pre=lineto,pre length=2pt, post=lineto, post length=4pt}},
node/.style={thick, draw=black, circle, inner sep = 0, minimum size=5mm, fill=gray!20},
]

\begin{scope}[shift={(0,0)}]

\node[node] (u) at (0*\xstep, 0.6*\ystep) {$u$};
\node[node] (x) at (1*\xstep, 0*\ystep) {$x$};
\node[node] (y) at (0.5*\xstep, -1.2*\ystep) {$y$};

\draw[post, loop above, out=110, in=70, looseness=10] (u) to node[above]{\MyColorOne{$+1$}} (u);
\draw[post] (x) to node[above]{\MyColorTwo{$+1$}} (u);
\draw[post, loop above, out=110, in=70, looseness=10] (x) to node[above]{\MyColorTwo{$+1$}} (x);
\draw[post, loop right, out=25, looseness=10] (x) to node[right]{\MyColorOne{$+1$}} (x);

\draw[post, bend left=20] (x) to node[right]{\MyColorTwo{$+1$}} (y);
\draw[post, bend right=20] (x) to node[left]{\MyColorOne{$+1$}} (y);

\end{scope}

\begin{scope}[shift={(7,0)}]

\node[node] (u) at (0*\xstep, 0.6*\ystep) {$u$};
\node[node] (y) at (0*\xstep, -1.2*\ystep) {$y$};
\node[node] (x1) at (1*\xstep, 0.6*\ystep) {$x_1$};
\node[node] (x2) at (1*\xstep, -1.2*\ystep) {$x_2$};

\draw[post, loop above, out=110, in=70, looseness=10] (u) to node[above]{\MyColorOne{$+1$}} (u);
\draw[post] (x1) to node[above]{\MyColorTwo{$+1$}} (u);
\draw[post, loop above, out=110, in=70, looseness=10] (x1) to node[above]{\MyColorTwo{$+1$}} (x1);
\draw[post, loop right, out=25, looseness=10] (x1) to node[right]{\MyColorOne{$+1$}} (x1);
\draw[post, loop above, out=110, in=70, looseness=10] (x2) to node[above]{\MyColorTwo{$+1$}} (x2);
\draw[post, loop right, out=25, looseness=10] (x2) to node[right]{\MyColorOne{$+1$}} (x2);

\draw[post, bend right=20] (x2) to node[above]{\MyColorTwo{$+1$}} (y);
\draw[post, bend left=20] (x2) to node[below]{\MyColorOne{$+1$}} (y);

\end{scope}

\end{tikzpicture}
\caption{
\emph{(Left):~} A  bidirected, interleaved Dyck graph $G$ with a doubly-self-looped node $x$.
\emph{(Right):}~ The two connected components on which we perform $\Dyck_1\odot \Dyck_1$ reachability separately in order to infer whether $x$ is reachable from any other node.
}
\label{fig:double_self_loop}
\end{figure}

%% file: figures/experimental_algos_trimming.tex
\begin{figure}
\newcommand{\xtstep}{0.9}
\newcommand{\ytstep}{0.5}
\newcommand{\xstep}{1.5}
\newcommand{\ystep}{1}

\centering
\begin{tikzpicture}[thick,
pre/.style={<-, thick, shorten >=0cm,shorten <=0cm},
post/.style={->, thick, shorten >=-0cm,shorten <=-0cm},
seqtrace/.style={->, line width=2},
postpath/.style={->, thick, decorate, decoration={zigzag,amplitude=1pt,segment length=2mm,pre=lineto,pre length=2pt, post=lineto, post length=4pt}},
node/.style={thick, draw=black, circle, inner sep = 0, minimum size=5mm, fill=gray!20},
node2/.style={thick, draw=black, circle, inner sep = 0, minimum size=5mm, fill=gray!50},
]

\begin{scope}[shift={(0,0)}]

\node[node] (u) at (0*\xstep, 1*\ystep) {$u$};
\node[node] (v) at (0*\xstep, -1*\ystep) {$v$};
\node[node] (x) at (1*\xstep, -1.2*\ystep) {$x$};
\node[node] (y) at (1*\xstep, 0.4*\ystep) {$y$};
\node[node2] (z) at (2*\xstep, -1.2*\ystep) {$z$};

\node (a) at (-0.5*\xstep, 0.5*\ystep) {$ $};
\node (b) at (0.3*\xstep, 1.7*\ystep) {$ $};
\node (c) at (1.1*\xstep, 1.2*\ystep) {$ $};
\node (d) at (-0.5*\xstep, -0.7*\ystep) {$ $};
\node (e) at (-0.5*\xstep, 1.5*\ystep) {$ $};
\draw (u) to (a);
\draw (v) to (d);
\draw (u) to (e);

\draw[post] (y) to node[right]{\MyColorTwo{$\beta$}} (x);
\draw[post] (z) to node[above]{\MyColorOne{$\alpha$}} (x);
\draw[post] (u) to node[above]{\MyColorOne{$\alpha$}} (y);
\draw[post] (y) to node[above]{\MyColorOne{$\alpha$}} (v);
\draw[post] (v) to node[left]{\MyColorTwo{$\beta$}} (u);

\end{scope}

\begin{scope}[shift={(5,0)}]

\node[node] (u) at (0*\xstep, 1*\ystep) {$u$};
\node[node] (v) at (0*\xstep, -1*\ystep) {$v$};
\node[node2] (x) at (1*\xstep, -1.2*\ystep) {$x$};
\node[node] (y) at (1*\xstep, 0.4*\ystep) {$y$};

\node (a) at (-0.5*\xstep, 0.5*\ystep) {$ $};
\node (b) at (0.3*\xstep, 1.7*\ystep) {$ $};
\node (c) at (1.1*\xstep, 1.2*\ystep) {$ $};
\node (d) at (-0.5*\xstep, -0.7*\ystep) {$ $};
\node (e) at (-0.5*\xstep, 1.5*\ystep) {$ $};
\draw (u) to (a);
\draw (v) to (d);
\draw (u) to (e);

\draw[post] (y) to node[right]{\MyColorTwo{$\beta$}} (x);
\draw[post] (u) to node[above]{\MyColorOne{$\alpha$}} (y);
\draw[post] (y) to node[above]{\MyColorOne{$\alpha$}} (v);
\draw[post] (v) to node[left]{\MyColorTwo{$\beta$}} (u);

\end{scope}

\begin{scope}[shift={(10,0)}]

\node[node] (u) at (0*\xstep, 1*\ystep) {$u$};
\node[node] (v) at (0*\xstep, -1*\ystep) {$v$};
\node[node] (y) at (1*\xstep, 0.4*\ystep) {$y$};

\node (a) at (-0.5*\xstep, 0.5*\ystep) {$ $};
\node (b) at (0.3*\xstep, 1.7*\ystep) {$ $};
\node (c) at (1.1*\xstep, 1.2*\ystep) {$ $};
\node (d) at (-0.5*\xstep, -0.7*\ystep) {$ $};
\node (e) at (-0.5*\xstep, 1.5*\ystep) {$ $};
\draw (u) to (a);
\draw (v) to (d);
\draw (u) to (e);

\draw[post] (u) to node[above]{\MyColorOne{$\alpha$}} (y);
\draw[post] (y) to node[above]{\MyColorOne{$\alpha$}} (v);
\draw[post] (v) to node[left]{\MyColorTwo{$\beta$}} (u);

\end{scope}

\end{tikzpicture}
\caption{
\emph{(Left):~} A sub-graph of a bidirected, interleaved Dyck graph $G$. 
By construction, any witness path that goes through $z$  must contain the sub-path $xzx$, which can be substituted by $x$. 
Moreover, no path starting in $z$ can go beyond $x$, as going to $y$ must pop a $\beta$. Thus $z$ can safely be removed.
\emph{(Middle):~} By construction, any witness path that goes through $x$ must contain the sub-path $yxy$, which can be substituted for $y$. 
Moreover, any path leading to $x$ will have a stack word of at least one $\beta$, meaning it is unreachable and can be safely removed.
\emph{(Right):~} The trimmed sub-graph.
}
\label{fig:trimming}
\end{figure}
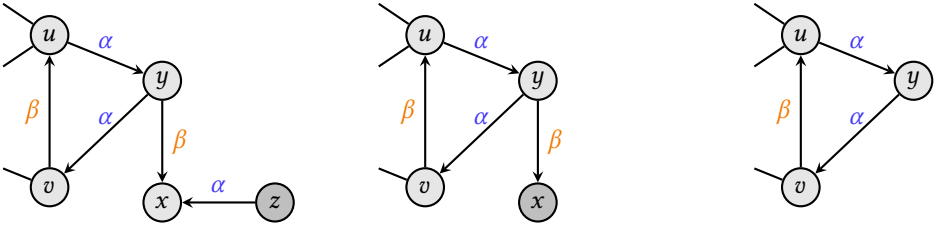

%% file: experimental_results.tex
\subsection{Experimental Results}\label{subsec:experimental_results}

We are now ready to report on our experimental results.

\Paragraph{Experimental setup.}
We have used the DaCapo benchmarks~\cite{Blackburn06} as in earlier works~\cite{Zhang2017,Li2021}, for performing context-sensitive and field-sensitive alias analysis.
Each benchmark provided one interleaved bidirected Dyck graph that models context and field sensitivity, and creates an instance of $\Dyck_k\odot \Dyck_k$ reachability.
In all cases needed, we projected a language $\Dyck_k$ to $\Dyck_1$ by projecting the parenthesis alphabet to a unique parenthesis symbol.
For $\Dyck_k\odot \Dyck_1$ reachability we used a counter bound of $n$.

Each of the above graphs was given as input to our algorithm.
This is in small difference to the procedure of~\cite{Li2021}, where the graph was first reduced using a recent fast simplification technique~\cite{Li2020}.
Our algorithms were implemented in C++ 11 without any compiler optimizations, and were executed on a conventional laptop with a Core i7 CPU with 8 GB of memory running Ubuntu 20.10.
The results are reported in \cref{tab:results_merged}.

\input{tables/table}

\Paragraph{Results on $\Dyck_1\odot\Dyck_1$.}
We first discuss $\Dyck_1\odot\Dyck_1$ reachability.
Recall that $\Dyck_k$ reachability (on the union alphabet) serves as an under-approximation of $\Dyck_1\odot \Dyck_1$ reachability.
The table shows that $\Dyck_k$ reachability already discovers most of the connected components, with $\Dyck_1\odot \Dyck_1$ refining this results with a few remaining components that are hard to detect.
This observation is in alignment to the study of~\cite{Li2021}, where it is reported that the number of new pairs discovered by $\Dyck_1\odot \Dyck_1$ reachability is only about 1\% more than those discovered by $\Dyck_k$ reachability.

Regarding running time, we see that our algorithm handles each benchmark within seconds, while the total benchmark set is processed in less than 5 minutes.
As a comparison, the recent method of~\cite{Li2021} reported more than 3 days ($\sim$ 74 hours)  for the same benchmark set when run on a relatively big machine, and it does so when the input graphs are simplified by a preprocessing step~\cite{Li2020}.
As a sanity check, we have verified that our algorithm also handles the simplified graphs within seconds,
while our own implementation of the technique of~\cite{Li2021} does not finish on our machine after 5 hours even for small graphs.

\Paragraph{Results on $\Dyck_k\odot\Dyck_1$.}
We now turn our attention to $\Dyck_k\odot\Dyck_1$ reachability.
Interestingly, we see that staying faithful to one Dyck language (not projecting $\Dyck_k$ to $\Dyck_1$) increases the number of connected components (i.e., there are fewer reachable pairs).
In principle, this could be due to the theoretical incompleteness arising from the bound on the counter.
However, we believe that our reports do not miss any reachable pairs, and thus the reduced reachability relationships (compared to $\Dyck_1\odot \Dyck_1$) are true negatives and thus increase the analysis precision.
This is further supported by our experience that all reported reachability relationships under $\Dyck_k\odot \Dyck_1$ were already formed with counter values much smaller than $n$.
Moreover, we find that performing $\Dyck_k$ reachability (on the union alphabet) provides an under-approximation that is often very close to (and sometimes matches) the final result.

Finally, regarding time, we again see that our algorithm handles each benchmark within seconds, and the whole benchmark set is again processed in less than 5 minutes.
These times indicate that our algorithms are suitable for static analysis tools.

%% file: tables/table.tex
\begin{table}[]
\label{tab:results_merged}
\caption{
Experimental results on $\Dyck_1\odot \Dyck_1$ and $\Dyck_k\odot \Dyck_1$ reachability.
In each case, ID-CCs denotes the number of connected components wrt the interleaved Dyck language,
while D-CCs denotes the number of connected components wrt the under-approximating Dyck language on the union alphabet.
}
\begin{tabular}{| c | c || c | c | c || c | c | c |}
\hline
\textbf{Benchmark}  & \textbf{$n$} & \multicolumn{3}{c||}{$\Dyck_1 \odot \Dyck_1$} & \multicolumn{3}{c|}{$\Dyck_k \odot \Dyck_1$} \\ \hline
& & \textbf{ID-CCs} & \textbf{D-CCs} & \textbf{Time (s)} & \textbf{ID-CCs} & \textbf{D-CCs} & \textbf{Time (s)} \\ \hline \hline
antlr & 29831 & 26793 & 26825 & 40.5 & 27014 & 27015 & 9.1\\ \hline
bloat & 36181 & 32693 & 32725 & 16.2 & 32946 & 32946 & 12.8\\ \hline
chart & 67535 & 60787 & 60844 & 32.8 & 61158 & 61158 & 65.1\\ \hline
eclipse & 30981 & 27812 & 27840 & 16.0 & 27999 & 28018 & 12.3\\ \hline
fop & 61016 & 54671 & 54723 & 29.8 & 55012 & 55012 & 58.4\\ \hline
hsqldb & 27494 & 24584 & 24610 & 15.5 & 24775 & 24776 & 8.8\\ \hline
jython & 36162 & 31811 & 31845 & 26.3 & 32060 & 32062 & 21.5\\ \hline
luindex & 28595 & 25610 & 25636 & 15.6 & 25809 & 25810 & 7.9\\ \hline
lusearch & 29530 & 26417 & 26447 & 17.8 & 26655 & 26655 & 9.1\\ \hline
pmd & 31333 & 28064 & 28093 & 18.2 & 28296 & 28297 & 9.9\\ \hline
xalan & 27358 & 24498 & 24523 & 15.2 & 24689 & 24690 & 7.4\\ \hline
\end{tabular}
\end{table}

%% file: related_work.tex
\section{Related Work}\label{sec:related_work}

\Paragraph{Complexity of Dyck reachability.}
The immense importance of Dyck reachability in static analyses has lead to a systematic study of its complexity in various settings.
General Dyck reachability can  be solved in $O(n^3)$ time~\cite{Yannakakis90} using an extension of the CYK algorithm to graphs.
The bound is believed to be tight as the problem is 2NPDA-hard~\cite{Heintze97}, while the combinatorial cubic hardness persists even on constant-treewidth graphs~\cite{Chatterjee18}.
Sub-cubic algorithms do exist, but they only offer logarithmic speedups~\cite{Chaudhuri2008}.
When the underlying graph is a Recursive State Machine (RSM) with constant entries and exits,
treewidth has been shown to lead to fast on-demand reachability queries~\cite{Chatterjee2015,Chatterjee2020}.
Despite the cubic hardness of the general problem, it is known to have sub-cubic certificates for both positive and negative instances~\cite{Chistikov2021}.
Dyck reachability over a single parenthesis symbol (aka one-counter systems)
has been shown to admit a sub-cubic bound~\cite{Bradford2018}.
The best current bound for this setting is $O(n^{\omega}\cdot \log^2n)$~\cite{Mathiasen2021}, where $\omega$ is the matrix-multiplication exponent,
while it is also known that the problem has a (conditional) $\Omega(n^{\omega})$ lower bound even for the single-pair question~\cite{Hansen2021}.

The algorithmic benefit of bidirectedness was highlighted in~\cite{Yuan09}, 
where an $O(n\cdot \log n)$ algorithm was presented when the underlying graph is a bidirected tree.
Later this bound was improved to $O(n)$ for trees, 
while the problem was shown to take $O(n^2)$ time (and $O(n\cdot \log n)$ expected time) on general bidirected graphs, thereby breaking below the cubic bound~\cite{Zhang13}.
This sequence of improvements ended with an $O(n\cdot \alpha(n))$ algorithm on general bidirected graphs (and $O(n)$ expected time), where $\alpha(n)$ is the inverse Ackermann function, and the bound was also shown to be tight~\cite{Chatterjee18}.

\Paragraph{Interleaved Dyck reachability.}
The modeling power of interleaved Dyck reachability, as well as its undecidability, were illustrated in~\cite{Reps00}.
As the problem is of vital importance to static analyses, various approximations have been developed.
The most basic, but also widely used approach is to approximate one Dyck language with a regular language, e.g., by bounding the recursion depth~\cite{Sridharan2006,Lerch2015}.
Other approximate techniques involve linear conjunctive language reachability~\cite{Zhang2017}, synchronized pushdown systems~\cite{Spath2019}, and CEGAR-style techniques~\cite{Ferles2021}.
The hardness of the problem stems from the fact that the underlying pushdown automata is operating on multiple stacks.
The problem also arises in distributed models of pushdown systems~\cite{Madhusudan2011}, and has also been addressed with approximations based on parameterization~\cite{Kahlon2008} and bounded-context switching~\cite{Qadeer2005}.

\Paragraph{Vector addition systems.}
When one or both Dyck languages are over one parenthesis symbol,
interleaved Dyck reachability falls in the class of vector addition systems.
In particular, $\Dyck_1\odot \Dyck_1$ yields a unary vector addition system with states in two dimensions,
while $\Dyck_k\odot \Dyck_1$ yields a unary pushdown vector addition system in one dimension.
Reachability for the first case is $\NL$-complete~\cite{Englert2016},
while the decidability of reachability in the second case is open as for now (to our knowledge, see, e.g.,~\cite[Table~1]{Schmitz2019}).
Nevertheless, the latter model enjoys some nice properties, such as decidability of coverability~\cite{Leroux2015}
and decidability of the boundedness of the reachable set~\cite{Leroux2015b}.
Independently of our work, bidirectedness was recently studied in vector additions systems in~\cite{Ganardi2021}.

%% file: conclusion.tex
\section{Conclusion}\label{sec:conclusion}
In this work we have addressed the decidability and complexity of interleaved Dyck reachability on bidirected graphs, inspired by the recent work of~\cite{Li2021}.
We have developed an efficient, nearly cubic-time algorithm for the $\Dyck_1\odot \Dyck_1$ case, while we have shown that the $\Dyck_k\odot \Dyck_1$ is decidable.
As a means to more tractable solutions, we have shown a nearly quadratic bound for $\Dyck_k\odot \Dyck_1$ when restricting witnesses to linearly-bounded counters, and we have shown that this quadratic bound is tight.
Finally, we have shown that general $\Dyck_k\odot\Dyck_k$ reachability remains undecidable on bidirected graphs.
Our results cover an important missing piece in the decidability and complexity landscape of static analyses modeled as Dyck/CFL reachability.
Moreover, our experiments show that the new algorithms can handle big benchmarks within seconds, making them suitable for static analysis tools.

%% file: appendix.tex
\input{proofs}

%% file: proofs.tex
\section{Proofs}
\input{proofs_secd1d1}

\input{proofs_secdkd1}

%% file: proofs_secd1d1.tex
\subsection{Proofs of \cref{sec:d1d1}}\label{sec:proofs_d1d1}

\lemmatchingpaircounter*
\begin{proof}
Since $P$ is a shallow path, for all values $c$ of the first counter such that $\gamma\leq c \leq \delta$, 
we have $\Counter_2(P[:l_c])\leq 6\cdot n$.
Observe that there exist $6\cdot n^2+1$ distinct such values for $c$.
Thus we have at least $6\cdot n^2 +1$ matching pairs $(x_c, y_c)$ with $c$ taking values in the above interval.
It follows that there exist two distinct matching pairs $(x_{c_1}, y_{c_1})$  and $(x_{c_2}, y_{c_2})$ such that $\gamma \leq c_1 < c_2\leq \delta$ which agree both on the first node and the value of the second counter that $P$ has when reaching that node.
That is, $x_{c_1}=x_{c_2}$ and $\Counter_{2}(P[\colon l_{c_1}])=\Counter_{2}(P[\colon l_{c_2}])$.

The case for $(x_{c_3}, y_{c_3})$ and $(x_{c_4}, y_{c_4})$ follows by repeating the same argument for $y_c$.
\end{proof}

\lempathdeflation*
\begin{proof}
We first apply \cref{lem:matching_pair_counter} to obtain the matching pairs $(x_{c_j}, y_{c_j})_{1\leq j \leq 4}$ with the properties stated in the lemma, where $\gamma=12\cdot n^2 + 6\cdot n$.
Consider the following paths.
\[
P_1=P[\colon l_{c_1}]\qquad 
P_2=P[l_{c_1}\colon l_{c_2}]\qquad
P_3=P[l_{c_2}\colon r_{c_3}]\qquad
P_4=P[r_{c_3}\colon r_{c_4}]\qquad
P_5=P[r_{c_4}\colon]
\]
We construct the path $Q$ as
\[
Q=P_1\circ P_3 \circ P_4 \circ \ov{P}_3 \circ P_2 \circ P_3 \circ P_5
\]
Note that $Q$ respects the edges of $G$ due to \cref{item:mathcing_pair_counter2} of \cref{lem:matching_pair_counter}, while $\Counter_i(Q)=0$ for each $i\in[2]$.

We first argue that $Q$ is a valid path.
Due to \cref{item:mathcing_pair_counter3} of \cref{lem:matching_pair_counter}, the value of the second counter is identical in $P$ and $Q$ as we traverse the corresponding segments of $P$ and $Q$, and thus stays non-negative along $Q$.
This also proves \cref{item:path_deflation1}  of the lemma.
We now show that the the first counter also stays non-negative along $Q$.
First note that
\[
c_2-c_1\leq \delta-\gamma =  18\cdot n^2 + 6\cdot n + 1 - (12\cdot n^2 + 6\cdot n) = 6\cdot n^2 + 1 < \gamma/2
\numberthis\label{eq:path_counters}
\]
while $c_3-c_4 < \gamma/2$ by a similar analysis.
We now consider the value of the first counter along the various sub-paths of $Q$.
\begin{compactenum}
\item The value of the counter is identical in the $P_1$ and $P_5$ sub-paths of $Q$ and $P$, and thus non-negative along these sub-paths $Q$.
\item While traversing $P_3$ for the first time, the value of the counter is at least $\min(c_2, c_3) - (c_2-c_1)\geq 0$,
using \cref{eq:path_counters} and the fact that $c_2, c_3\geq \gamma$.
\item While traversing $P_4$, the value of the counter is at least $c_4-(c_2-c_1)\geq 0$, using \cref{eq:path_counters} and the fact that $c_4\geq \gamma$.
\item While traversing $P_3$ for the second time (that is, when $Q$ traverses $\ov{P}_3$),
the value of the counter is at least $\min(c_2, c_3) - (c_2-c_1) - (c_3-c_4)\geq 0$,
using \cref{eq:path_counters} and the fact that $c_2, c_3\geq \gamma$.
\item While traversing $P_2$, the value of the counter is at least $c_1-(c_3-c_4)\geq 0$,
using \cref{eq:path_counters} and the fact that $c_1\geq \gamma$.
\item Finally, while traversing $P_3$ for the third time, the value of the counter is at least $\min(c_2, c_3) - (c_3-c_4)\geq 0$, using \cref{eq:path_counters} and the fact that $c_2, c_3\geq \gamma$.
\end{compactenum}
Thus both counters stay non-negative along $Q$, and hence $Q$ is a valid path.

We now argue that $Q$ is a shallow path.
As we argued in the previous paragraph, the value of the second counter is identical in $P$ and $Q$ as we traverse the corresponding segments of $P$ and $Q$.
Moreover, it is straightforward that the value of the first counter is no larger in $Q$ than in $P$ as we traverse the corresponding sub-paths.
Hence $Q$ is a shallow path.

Finally, we argue that $|\CounterIndex_1^c(Q)|<|\CounterIndex_1^c(P)|$.
It is straightforward to verify that any point in which the counter reaches the value $c$ while traversing the sub-path $P_5$ of $Q$ also corresponds to a unique point where the counter attains the same value along the sub-path $P_5$ of $P$.
We argue that the first counter does not reach the value $c$ in any earlier sub-path of $Q$.
Since the counter reaches the value $c$ in the sub-path $P_3$ of $P$, this concludes that $|\CounterIndex_1^c(Q)|<|\CounterIndex_1^c(P)|$.
\begin{compactenum}
\item Since the counter does not reach $c$ while traversing $P_1$ in $P$, the same holds while traversing $P_1$ in $Q$.
\item In the first traversal of $P_3$ in $Q$, the first counter has decreased by $c_2-c_1>0$.
Hence it cannot reach the value $c$ in this sub-path.
\item The same holds for the sub-path $P_4$.
\item In the second traversal of $P_3$ in $Q$ (in particular, when $Q$ traverses $\ov{P}_3$), the first counter has decreased by $c_2-c_1 + c_3-c_4 >0$.
Hence it cannot reach the value $c$ in this sub-path.
\item Since the counter does not reach $c$ while traversing $P_2$ in $P$, the same holds while traversing $P_2$ in $Q$.
\item In the third traversal of $P_3$, the first counter has decreased by $c_3-c_4>0$.
Hence it cannot reach the value $c$ in this sub-path.
\end{compactenum}

The desired result follows.
\end{proof}

%% file: proofs_secdkd1.tex
\subsection{Proofs of \cref{sec:dkd1}}\label{sec:proofs_dkd1}

\lemmatchingpair*
\begin{proof}
There are at most $n^2$ distinct matching pairs.
Since $\MaxStackHeight(P)\geq n^2$, we have at least $(n+1)^2$ matching pairs.
By the pigeonhole principle, there is a matching pair that appears in two different stack heights $h$ and $h'$.
The desired result follows.
\end{proof}

\lemboundedstackheight*
\begin{proof}
Assume that $\MaxStackHeight(P)\geq 2\cdot n^2+1$, otherwise we can take $Q=P$ and we are done.
Let $i_{\max}$ be the first point where $P$ attains its maximum stack height, and consider the matching node pairs of $P$.
By \cref{lem:matching_pair}, $P$ has two matching node pairs $(x_h, y_h)$ and $x_{h'}, y_{h'}$ such that
(i)~$h< h'\leq n^2$, (ii)~$x_h=x_{h'}$, and 
(iii)~$y_h=y_{h'}$.
Let $x=x_h=x_{h'}$ and $y=y_{h}=y_{h'}$.
Let $K\colon x\DKPath y$ be a shortest path witnessing the $\Dyck_k$-reachability of $y$ from $x$ in $G\Project \Alphabet_1$.
It is known that $\MaxStackHeight(K)\leq n^2$~\cite{Pierre1992}.
Let $P_{x}=P[l_h\colon l_{h'}]$ and $P_{y}=P[r_{h'}\colon r_h]$, i.e. these are the sub-paths of $P$ from $x_h$ to $x_{h'}$, and from $y_{h'}$ to $y_h$, respectively.
We construct the path
\[
R=P[\colon l_{h'}] \circ K \circ P[r_{h'}\colon r_{h}] \circ \ov{K} \circ P[l_{h'}\colon r_{h'}] \circ P[r_{h}\colon]
\]
In words, once we reach $x_{h'}$, we take the path $K$ to reach $y_{h'}$. At that point we traverse the sub-path $P[r_{h'}\colon l_{h}]$ in order to pop the stack that $P$ has accumulated between $x_h$ and $x_{h'}$.
Then we take $\ov{K}$ back to $x_{h'}$ and continue normally, but skip the previously traversed sub-path $P[r_{h'}\colon l_{h}]$.

Observe that either $\MaxStackHeight(R)<\MaxStackHeight(P)$, or $R$ has one less point in which it attains its maximum stack height compared to $P$.
This follows from the following two observations.
\begin{compactenum}
\item Since $\StackHeight(K)\leq n^2$, we have $\MaxStackHeight(P[\colon l_{h'}] \circ K)\leq 2\cdot n^2$, while the stack height of $P[\colon l_{h'}] \circ K \circ P[r_{h'}\colon r_{h}] \circ \ov{K}$ during the traversal of $\ov{K}$ is strictly below $2\cdot n^2$.
\item The number of times $R$ attains the stack height $\MaxStackHeight(P)$ along the \ColorOne $P[r_{h'}\colon r_{h}]$ is equal to the number of times that $P$ attains its maximum stack height along the prefix $P[\colon r_h]$ minus one,  as while traversing the sub-path $P[l_{h'}\colon r_{h'}]$, the stack height of $R$ is strictly smaller as compared to the stack height of $P$ during $P[l_{h'}\colon r_{h'}]$.
\end{compactenum}
Moreover, observe that $\Counter(R)=\Counter(P)$, as the only additional edges that $R$ has occur along the paths $K$ and $\ov{K}$, and thus cancel each-other out.

Finally, we can repeat the above process on $R$ instead of $P$ recursively, until we arrive at a path $Q$ with the stated properties.

The desired result follows.
\end{proof}

\lemboundedstackheightwitness*
\begin{proof}
Let $P\colon u\DKDOPath v$ be an arbitrary path that witnesses the reachability of $v$ from $u$.
Hence we have $P\colon u\DKPath v$ in $G\Project \Alphabet_1$.
By \cref{lem:bounded_stack_height}, there is a path $Q\colon u\DKPath v$ in $G\Project \Alphabet_1$ such that
(i)~$\MaxStackHeight(Q)\leq 2\cdot n^2$, and 
(ii)~$\Counter(Q)=\Counter(P)$.
Note that $Q$ might not be a valid path because the counter becomes negative along $Q$.
Let $\ell$ be the smallest integer such that $\ell + \Counter(Q')\geq 0$ for every $Q'\in \Prefixes(Q)$.
Let $\ell_u$ (resp., $\ell_v$) be the smallest integer such that $\ell_u\cdot \Counter(F_u) \geq \ell$
(resp., $\ell_v\cdot \Counter(F_v) \geq \ell$).
We construct $T$ as the path
\[
T=F_u^{\ell_u}\circ Q \circ F_v^{\ell_v}\circ \ov{Q} \circ \ov{F}^{\ell_u}_u \circ Q\circ \ov{F}^{\ell_v}_v
\]
where $F_u^{\ell_u}$ (resp., $F_v^{\ell_v}$) denotes $\ell_u$ (resp., $\ell_v$) iterations of $F_u$ (resp., $F_v$).
It is easy to verify that $\Counter(T)=0$ and $\Stack(T)=\epsilon$, while the counter remains non-negative along $T$.
Finally, at the end of each of the distinct sub-paths above, the stack of $T$ is empty.
Hence $\MaxStackHeight(T)\leq \max(\zeta, \MaxStackHeight(Q))= \max(\zeta, 2\cdot n^2)$.

The desired result follows.
\end{proof}

\lemboundedlengthpath*
\begin{proof}
We first establish that there is a path $Q$ with $\MaxCounter(Q)\leq n^2\cdot k^{2\cdot \delta}$ and $\MaxStackHeight(Q)\leq \MaxStackHeight(P)$.
Afterwards we argue that such a $Q$ satisfies the properties of the lemma.

Assume that $\MaxCounter(P)> n^2\cdot k^{2\cdot \delta}$, otherwise we are done.
Let $i_{\max}$ be the first index of $P$ in which the counter attains its maximum value.
For every $c\in [\MaxCounter(P)]$, we let
\[
l_{c}=\max(\{i\colon i\leq i_{\max} \text{ and } \Counter(P[\colon i])=c \}) \quad \text{and}\quad
r_{c}=\min(\{i\colon i\geq i_{\max} \text{ and }  \Counter(P[\colon i])=c \})
\]
Let $x_c$ (resp., $y_c$) be the last node of $P[\colon l_c]$ (resp., $P[\colon r_c]$), and $p_c=\Stack(P[\colon l_c])$ and $q_c=\Stack(P[\colon r_c])$ the corresponding stacks.
Note that, since there are $n$ nodes and $\MaxStackHeight(P)\leq \delta$, there are at most $n^2\cdot k^{2\cdot \delta}$ distinct pairs $(x,p)$ where $x$ is a node and $p$ is a stack that $P$ attains at some point.
Since $\MaxCounter(P)>n^2\cdot k^{2\cdot \delta}$, there exist $c, c'$ such that 
(i)~$c< c'$, (ii)~$x_c=x_{c'}$, and 
(iii)~$y_c=y_{c'}$.
Then we can construct another path $P'\colon u\DKDOPath v$ as
\[
P'=P[\colon l_{c}] \circ P[l_{c'}\colon r_{c'}] \circ P[r_{c}\colon]
\]
Note that $|P'|<|P|$, while $\MaxCounter(P')\leq \MaxCounter(P)$ and $\MaxStackHeight(P')\leq \MaxStackHeight(P)$.
Applying the above process for $P'$ recursively, we arrive at a path $Q$ with $\MaxCounter(Q)\leq n^2\cdot k^{2\cdot \delta}$ and $\MaxStackHeight(Q)\leq \MaxStackHeight(P)\leq \delta$.

Finally, it is clear that any irreducible path with stack height bounded by $\delta$ and counter bounded by $\gamma= n^2\cdot k^{2\cdot \delta}$ has length at most
\[
n\cdot \gamma\cdot \delta= n\cdot k^{\delta} \cdot n^2\cdot k^{2\cdot \delta} = n^3\cdot k^{3\cdot \delta}
\] 

The desired result follows
\end{proof}

\lemovcorrectness*
\begin{proof}
We separately argue about completeness and soundness.

\SubParagraph{Completeness.}
Assume that there is an orthogonal pair $(x_i, y_j)\in X\times Y$.
We construct the path $P\colon u\DKDOPath v$ as follows.
We first take the unique irreducible path $u\Path w$ that traverses the nodes $x_i^{\ell}$.
Afterwards, we take the unique irreducible pat $w\Path v$ that traverses the unique edge labeled with $y_j[\ell]$  when on node nodes $y_j^{\ell}$.
At this point, the stack of the path is empty but the counter equals $\sum_{\ell}(1-y_j[\ell])\cdot 2^{\ell}$.
Finally, we loop on $v$ reduce the counter to $0$, thereby reaching leading us to $v$.

\SubParagraph{Soundness.}
Assume that there is a (irreducible) path $P\colon u\DKDKPath v$.
Observe that $P$ does not traverse the same edge in both directions.
Indeed, on the way from $u$ to $w$ this is the case as every node $x_i^{\ell}$ has degree $2$.
Thus, when $P$ reaches $w$ for the first time, its stack encodes the bits of some vector $x_i$.
In turn, once the path reaches some node $y_j^{1}$, the stack completely dictates how to traverse to node $y_j^{\ell+1}$ from node $y_j^{\ell}$.
Although the path can transition from some node $y_j^{\ell}$ to $y_j^{\ell-1}$, it cannot do so by traversing the counter gadget $H_{2^{\ell}}$ backwards, as the counter value is bounded by $2^{\ell}-1$.
Hence, if $P$ returns to $w$, its stack will encode a vector (possibly not in $X$) that does not have a $0$ in a coordinate where $x_i$ has a $1$.
It is straightforward to verify that once the path $P$ has reached $v$, it has traversed the nodes $x_i^{\ell}$ and $y_{j}^{\ell}$, for some $i,j\in[n]$, such that $x_i$ and $y_j$ are orthogonal.

The desired result follows.
\end{proof}

%% file: main.bbl

\begin{thebibliography}{55}


\ifx \showCODEN    \undefined \def \showCODEN     #1{\unskip}     \fi
\ifx \showDOI      \undefined \def \showDOI       #1{#1}\fi
\ifx \showISBNx    \undefined \def \showISBNx     #1{\unskip}     \fi
\ifx \showISBNxiii \undefined \def \showISBNxiii  #1{\unskip}     \fi
\ifx \showISSN     \undefined \def \showISSN      #1{\unskip}     \fi
\ifx \showLCCN     \undefined \def \showLCCN      #1{\unskip}     \fi
\ifx \shownote     \undefined \def \shownote      #1{#1}          \fi
\ifx \showarticletitle \undefined \def \showarticletitle #1{#1}   \fi
\ifx \showURL      \undefined \def \showURL       {\relax}        \fi
\providecommand\bibfield[2]{#2}
\providecommand\bibinfo[2]{#2}
\providecommand\natexlab[1]{#1}
\providecommand\showeprint[2][]{arXiv:#2}

\bibitem[\protect\citeauthoryear{??}{Wal}{2003}]%
        {Wala}
 \bibinfo{year}{2003}\natexlab{}.
\newblock \bibinfo{title}{T. J. Watson Libraries for Analysis (WALA)}.
\newblock \bibinfo{howpublished}{https://github.com}.   (\bibinfo{year}{2003}).
\newblock


\bibitem[\protect\citeauthoryear{Alur and Madhusudan}{Alur and
  Madhusudan}{2004}]%
        {Alur2004}
\bibfield{author}{\bibinfo{person}{Rajeev Alur} {and} \bibinfo{person}{P.
  Madhusudan}.} \bibinfo{year}{2004}\natexlab{}.
\newblock \showarticletitle{Visibly Pushdown Languages}. In
  \bibinfo{booktitle}{{\em Proceedings of the Thirty-Sixth Annual ACM Symposium
  on Theory of Computing}} {\em (\bibinfo{series}{STOC '04})}.
  \bibinfo{publisher}{Association for Computing Machinery},
  \bibinfo{address}{New York, NY, USA}, \bibinfo{pages}{202–211}.
\newblock
\showISBNx{1581138520}
\showDOI{%
\url{https://doi.org/10.1145/1007352.1007390}}


\bibitem[\protect\citeauthoryear{Arnold}{Arnold}{1996}]%
        {Arnold96}
\bibfield{author}{\bibinfo{person}{Robert~S. Arnold}.}
  \bibinfo{year}{1996}\natexlab{}.
\newblock \bibinfo{booktitle}{{\em Software Change Impact Analysis}}.
\newblock \bibinfo{publisher}{IEEE Computer Society Press},
  \bibinfo{address}{Los Alamitos, CA, USA}.
\newblock
\showISBNx{0818673842}


\bibitem[\protect\citeauthoryear{Blackburn}{Blackburn}{2006}]%
        {Blackburn06}
\bibfield{author}{\bibinfo{person}{Stephen M. et~al. Blackburn}.}
  \bibinfo{year}{2006}\natexlab{}.
\newblock \showarticletitle{The DaCapo Benchmarks: Java Benchmarking
  Development and Analysis}. In \bibinfo{booktitle}{{\em OOPSLA}}.
\newblock


\bibitem[\protect\citeauthoryear{Bodden}{Bodden}{2012}]%
        {Bodden12}
\bibfield{author}{\bibinfo{person}{Eric Bodden}.}
  \bibinfo{year}{2012}\natexlab{}.
\newblock \showarticletitle{Inter-procedural Data-flow Analysis with IFDS/IDE
  and Soot}. In \bibinfo{booktitle}{{\em SOAP}}. \bibinfo{publisher}{ACM},
  \bibinfo{address}{New York, NY, USA}.
\newblock


\bibitem[\protect\citeauthoryear{Bradford}{Bradford}{2018}]%
        {Bradford2018}
\bibfield{author}{\bibinfo{person}{Phillip~G. Bradford}.}
  \bibinfo{year}{2018}\natexlab{}.
\newblock \bibinfo{title}{Efficient Exact Paths For Dyck and semi-Dyck Labeled
  Path Reachability}.
\newblock   (\bibinfo{year}{2018}).
\newblock
\showeprint[arxiv]{cs.DS/1802.05239}


\bibitem[\protect\citeauthoryear{{Cetti Hansen}, {Husted Kjelstrøm}, and
  Pavlogiannis}{{Cetti Hansen} et~al\mbox{.}}{2021}]%
        {Hansen2021}
\bibfield{author}{\bibinfo{person}{Jakob {Cetti Hansen}}, \bibinfo{person}{Adam
  {Husted Kjelstrøm}}, {and} \bibinfo{person}{Andreas Pavlogiannis}.}
  \bibinfo{year}{2021}\natexlab{}.
\newblock \showarticletitle{Tight bounds for reachability problems on
  one-counter and pushdown systems}.
\newblock \bibinfo{journal}{{\it Inform. Process. Lett.}}
  \bibinfo{volume}{171} (\bibinfo{year}{2021}), \bibinfo{pages}{106135}.
\newblock
\showISSN{0020-0190}
\showDOI{%
\url{https://doi.org/10.1016/j.ipl.2021.106135}}


\bibitem[\protect\citeauthoryear{Chatterjee, Choudhary, and
  Pavlogiannis}{Chatterjee et~al\mbox{.}}{2018}]%
        {Chatterjee18}
\bibfield{author}{\bibinfo{person}{Krishnendu Chatterjee},
  \bibinfo{person}{Bhavya Choudhary}, {and} \bibinfo{person}{Andreas
  Pavlogiannis}.} \bibinfo{year}{2018}\natexlab{}.
\newblock \showarticletitle{Optimal Dyck Reachability for Data-Dependence and
  Alias Analysis}.
\newblock \bibinfo{journal}{{\em Proc. ACM Program. Lang.\/}}
  \bibinfo{volume}{2}, \bibinfo{number}{POPL}, Article
  \bibinfo{articleno}{Article 30} (\bibinfo{date}{Dec.} \bibinfo{year}{2018}),
  \bibinfo{numpages}{30}~pages.
\newblock


\bibitem[\protect\citeauthoryear{Chatterjee, Goharshady, Ibsen-Jensen, and
  Pavlogiannis}{Chatterjee et~al\mbox{.}}{2020}]%
        {Chatterjee2020}
\bibfield{author}{\bibinfo{person}{Krishnendu Chatterjee},
  \bibinfo{person}{Amir~Kafshdar Goharshady}, \bibinfo{person}{Rasmus
  Ibsen-Jensen}, {and} \bibinfo{person}{Andreas Pavlogiannis}.}
  \bibinfo{year}{2020}\natexlab{}.
\newblock \showarticletitle{Optimal and Perfectly Parallel Algorithms for
  On-demand Data-Flow Analysis}. In \bibinfo{booktitle}{{\em Programming
  Languages and Systems}}, \bibfield{editor}{\bibinfo{person}{Peter
  M{\"u}ller}} (Ed.). \bibinfo{publisher}{Springer International Publishing},
  \bibinfo{address}{Cham}, \bibinfo{pages}{112--140}.
\newblock
\showISBNx{978-3-030-44914-8}


\bibitem[\protect\citeauthoryear{Chatterjee, Ibsen-Jensen, Pavlogiannis, and
  Goyal}{Chatterjee et~al\mbox{.}}{2015}]%
        {Chatterjee2015}
\bibfield{author}{\bibinfo{person}{Krishnendu Chatterjee},
  \bibinfo{person}{Rasmus Ibsen-Jensen}, \bibinfo{person}{Andreas
  Pavlogiannis}, {and} \bibinfo{person}{Prateesh Goyal}.}
  \bibinfo{year}{2015}\natexlab{}.
\newblock \showarticletitle{Faster Algorithms for Algebraic Path Properties in
  Recursive State Machines with Constant Treewidth}. In
  \bibinfo{booktitle}{{\em Proceedings of the 42nd Annual ACM SIGPLAN-SIGACT
  Symposium on Principles of Programming Languages}} {\em
  (\bibinfo{series}{POPL '15})}. \bibinfo{publisher}{Association for Computing
  Machinery}, \bibinfo{address}{New York, NY, USA}, \bibinfo{pages}{97–109}.
\newblock
\showISBNx{9781450333009}
\showDOI{%
\url{https://doi.org/10.1145/2676726.2676979}}


\bibitem[\protect\citeauthoryear{Chatterjee, Kragl, Mishra, and
  Pavlogiannis}{Chatterjee et~al\mbox{.}}{2017}]%
        {ESOP17}
\bibfield{author}{\bibinfo{person}{Krishnendu Chatterjee},
  \bibinfo{person}{Bernhard Kragl}, \bibinfo{person}{Samarth Mishra}, {and}
  \bibinfo{person}{Andreas Pavlogiannis}.} \bibinfo{year}{2017}\natexlab{}.
\newblock \bibinfo{booktitle}{{\em Faster Algorithms for Weighted Recursive
  State Machines}}.
\newblock \bibinfo{publisher}{Springer Berlin Heidelberg},
  \bibinfo{address}{Berlin, Heidelberg}, \bibinfo{pages}{287--313}.
\newblock
\showISBNx{978-3-662-54434-1}
\showDOI{%
\url{https://doi.org/10.1007/978-3-662-54434-1_11}}


\bibitem[\protect\citeauthoryear{Chaudhuri}{Chaudhuri}{2008}]%
        {Chaudhuri2008}
\bibfield{author}{\bibinfo{person}{Swarat Chaudhuri}.}
  \bibinfo{year}{2008}\natexlab{}.
\newblock \showarticletitle{Subcubic Algorithms for Recursive State Machines}.
\newblock \bibinfo{journal}{{\em SIGPLAN Not.\/}} \bibinfo{volume}{43},
  \bibinfo{number}{1} (\bibinfo{date}{Jan.} \bibinfo{year}{2008}),
  \bibinfo{pages}{159–169}.
\newblock
\showISSN{0362-1340}
\showDOI{%
\url{https://doi.org/10.1145/1328897.1328460}}


\bibitem[\protect\citeauthoryear{Chistikov, Majumdar, and Schepper}{Chistikov
  et~al\mbox{.}}{2021}]%
        {Chistikov2021}
\bibfield{author}{\bibinfo{person}{Dmitry Chistikov}, \bibinfo{person}{Rupak
  Majumdar}, {and} \bibinfo{person}{Philipp Schepper}.}
  \bibinfo{year}{2021}\natexlab{}.
\newblock \bibinfo{title}{Subcubic Certificates for CFL Reachability}.
\newblock   (\bibinfo{year}{2021}).
\newblock
\showeprint[arxiv]{cs.FL/2102.13095}


\bibitem[\protect\citeauthoryear{Deutsch}{Deutsch}{1994}]%
        {Deutsch1994}
\bibfield{author}{\bibinfo{person}{Alain Deutsch}.}
  \bibinfo{year}{1994}\natexlab{}.
\newblock \showarticletitle{Interprocedural May-Alias Analysis for Pointers:
  Beyond <i>k</i>-Limiting}.
\newblock \bibinfo{journal}{{\em SIGPLAN Not.\/}} \bibinfo{volume}{29},
  \bibinfo{number}{6} (\bibinfo{date}{June} \bibinfo{year}{1994}),
  \bibinfo{pages}{230–241}.
\newblock
\showISSN{0362-1340}
\showDOI{%
\url{https://doi.org/10.1145/773473.178263}}


\bibitem[\protect\citeauthoryear{Englert, Lazi\'{c}, and Totzke}{Englert
  et~al\mbox{.}}{2016}]%
        {Englert2016}
\bibfield{author}{\bibinfo{person}{Matthias Englert}, \bibinfo{person}{Ranko
  Lazi\'{c}}, {and} \bibinfo{person}{Patrick Totzke}.}
  \bibinfo{year}{2016}\natexlab{}.
\newblock \showarticletitle{Reachability in Two-Dimensional Unary Vector
  Addition Systems with States is NL-Complete}. In \bibinfo{booktitle}{{\em
  Proceedings of the 31st Annual ACM/IEEE Symposium on Logic in Computer
  Science}} {\em (\bibinfo{series}{LICS '16})}. \bibinfo{publisher}{Association
  for Computing Machinery}, \bibinfo{address}{New York, NY, USA},
  \bibinfo{pages}{477–484}.
\newblock
\showISBNx{9781450343916}
\showDOI{%
\url{https://doi.org/10.1145/2933575.2933577}}


\bibitem[\protect\citeauthoryear{Ferles, Stephens, and Dillig}{Ferles
  et~al\mbox{.}}{2021}]%
        {Ferles2021}
\bibfield{author}{\bibinfo{person}{Kostas Ferles}, \bibinfo{person}{Jon
  Stephens}, {and} \bibinfo{person}{Isil Dillig}.}
  \bibinfo{year}{2021}\natexlab{}.
\newblock \showarticletitle{Verifying Correct Usage of Context-Free API
  Protocols}.
\newblock \bibinfo{journal}{{\em Proc. ACM Program. Lang.\/}}
  \bibinfo{volume}{5}, \bibinfo{number}{POPL}, Article \bibinfo{articleno}{17}
  (\bibinfo{date}{Jan.} \bibinfo{year}{2021}), \bibinfo{numpages}{30}~pages.
\newblock
\showDOI{%
\url{https://doi.org/10.1145/3434298}}


\bibitem[\protect\citeauthoryear{Ganardi, Majumdar, and Zetzsche}{Ganardi
  et~al\mbox{.}}{2021}]%
        {Ganardi2021}
\bibfield{author}{\bibinfo{person}{Moses Ganardi}, \bibinfo{person}{Rupak
  Majumdar}, {and} \bibinfo{person}{Georg Zetzsche}.}
  \bibinfo{year}{2021}\natexlab{}.
\newblock \bibinfo{title}{The complexity of bidirected reachability in valence
  systems}.
\newblock   (\bibinfo{year}{2021}).
\newblock
\showeprint[arxiv]{cs.FL/2110.03654}


\bibitem[\protect\citeauthoryear{Heintze and McAllester}{Heintze and
  McAllester}{1997}]%
        {Heintze97}
\bibfield{author}{\bibinfo{person}{Nevin Heintze} {and} \bibinfo{person}{David
  McAllester}.} \bibinfo{year}{1997}\natexlab{}.
\newblock \showarticletitle{On the Cubic Bottleneck in Subtyping and Flow
  Analysis}. In \bibinfo{booktitle}{{\em Proceedings of the 12th Annual IEEE
  Symposium on Logic in Computer Science}} {\em (\bibinfo{series}{LICS '97})}.
  \bibinfo{publisher}{IEEE Computer Society}, \bibinfo{address}{Washington, DC,
  USA}, \bibinfo{pages}{342--}.
\newblock
\showISBNx{0-8186-7925-5}
\showURL{%
\url{http://dl.acm.org/citation.cfm?id=788019.788876}}


\bibitem[\protect\citeauthoryear{Huang, Dong, Milanova, and Dolby}{Huang
  et~al\mbox{.}}{2015}]%
        {Huang2015}
\bibfield{author}{\bibinfo{person}{Wei Huang}, \bibinfo{person}{Yao Dong},
  \bibinfo{person}{Ana Milanova}, {and} \bibinfo{person}{Julian Dolby}.}
  \bibinfo{year}{2015}\natexlab{}.
\newblock \showarticletitle{Scalable and Precise Taint Analysis for Android}.
  In \bibinfo{booktitle}{{\em Proceedings of the 2015 International Symposium
  on Software Testing and Analysis}} {\em (\bibinfo{series}{ISSTA 2015})}.
  \bibinfo{publisher}{Association for Computing Machinery},
  \bibinfo{address}{New York, NY, USA}, \bibinfo{pages}{106–117}.
\newblock
\showISBNx{9781450336208}
\showDOI{%
\url{https://doi.org/10.1145/2771783.2771803}}


\bibitem[\protect\citeauthoryear{Kahlon}{Kahlon}{2008}]%
        {Kahlon2008}
\bibfield{author}{\bibinfo{person}{Vineet Kahlon}.}
  \bibinfo{year}{2008}\natexlab{}.
\newblock \showarticletitle{Parameterization as Abstraction: A Tractable
  Approach to the Dataflow Analysis of Concurrent Programs}. In
  \bibinfo{booktitle}{{\em Proceedings of the Twenty-Third Annual IEEE
  Symposium on Logic in Computer Science (LICS 2008)}}.
  \bibinfo{publisher}{IEEE Computer Society Press}, \bibinfo{pages}{181--192}.
\newblock


\bibitem[\protect\citeauthoryear{Lerch, Sp\"{a}th, Bodden, and Mezini}{Lerch
  et~al\mbox{.}}{2015}]%
        {Lerch2015}
\bibfield{author}{\bibinfo{person}{Johannes Lerch}, \bibinfo{person}{Johannes
  Sp\"{a}th}, \bibinfo{person}{Eric Bodden}, {and} \bibinfo{person}{Mira
  Mezini}.} \bibinfo{year}{2015}\natexlab{}.
\newblock \showarticletitle{Access-Path Abstraction: Scaling Field-Sensitive
  Data-Flow Analysis with Unbounded Access Paths}. In \bibinfo{booktitle}{{\em
  Proceedings of the 30th IEEE/ACM International Conference on Automated
  Software Engineering}} {\em (\bibinfo{series}{ASE '15})}.
  \bibinfo{publisher}{IEEE Press}, \bibinfo{pages}{619–629}.
\newblock
\showISBNx{9781509000241}
\showDOI{%
\url{https://doi.org/10.1109/ASE.2015.9}}


\bibitem[\protect\citeauthoryear{Leroux, Sutre, and Totzke}{Leroux
  et~al\mbox{.}}{2015a}]%
        {Leroux2015b}
\bibfield{author}{\bibinfo{person}{J{\'e}r{\^o}me Leroux},
  \bibinfo{person}{Gr{\'e}goire Sutre}, {and} \bibinfo{person}{Patrick
  Totzke}.} \bibinfo{year}{2015}\natexlab{a}.
\newblock \showarticletitle{On Boundedness Problems for Pushdown Vector
  Addition Systems}. In \bibinfo{booktitle}{{\em Reachability Problems}},
  \bibfield{editor}{\bibinfo{person}{Mikolai Bojanczyk},
  \bibinfo{person}{Slawomir Lasota}, {and} \bibinfo{person}{Igor Potapov}}
  (Eds.). \bibinfo{publisher}{Springer International Publishing},
  \bibinfo{address}{Cham}, \bibinfo{pages}{101--113}.
\newblock
\showISBNx{978-3-319-24537-9}


\bibitem[\protect\citeauthoryear{Leroux, Sutre, and Totzke}{Leroux
  et~al\mbox{.}}{2015b}]%
        {Leroux2015}
\bibfield{author}{\bibinfo{person}{J{\'e}r{\^o}me Leroux},
  \bibinfo{person}{Gr{\'e}goire Sutre}, {and} \bibinfo{person}{Patrick
  Totzke}.} \bibinfo{year}{2015}\natexlab{b}.
\newblock \showarticletitle{On the Coverability Problem for Pushdown Vector
  Addition Systems in One Dimension}. In \bibinfo{booktitle}{{\em Automata,
  Languages, and Programming}},
  \bibfield{editor}{\bibinfo{person}{Magn{\'u}s~M. Halld{\'o}rsson},
  \bibinfo{person}{Kazuo Iwama}, \bibinfo{person}{Naoki Kobayashi}, {and}
  \bibinfo{person}{Bettina Speckmann}} (Eds.). \bibinfo{publisher}{Springer
  Berlin Heidelberg}, \bibinfo{address}{Berlin, Heidelberg},
  \bibinfo{pages}{324--336}.
\newblock
\showISBNx{978-3-662-47666-6}


\bibitem[\protect\citeauthoryear{Lhot\'{a}k and Hendren}{Lhot\'{a}k and
  Hendren}{2006}]%
        {Lhotak06}
\bibfield{author}{\bibinfo{person}{Ond\v{r}ej Lhot\'{a}k} {and}
  \bibinfo{person}{Laurie Hendren}.} \bibinfo{year}{2006}\natexlab{}.
\newblock \showarticletitle{Context-Sensitive Points-to Analysis: Is It Worth
  It?}. In \bibinfo{booktitle}{{\em Proceedings of the 15th International
  Conference on Compiler Construction}} {\em (\bibinfo{series}{CC})}.
  \bibinfo{pages}{47--64}.
\newblock


\bibitem[\protect\citeauthoryear{Li, Zhang, and Reps}{Li et~al\mbox{.}}{2020}]%
        {Li2020}
\bibfield{author}{\bibinfo{person}{Yuanbo Li}, \bibinfo{person}{Qirun Zhang},
  {and} \bibinfo{person}{Thomas Reps}.} \bibinfo{year}{2020}\natexlab{}.
\newblock \showarticletitle{Fast Graph Simplification for Interleaved
  Dyck-Reachability}. In \bibinfo{booktitle}{{\em Proceedings of the 41st ACM
  SIGPLAN Conference on Programming Language Design and Implementation}} {\em
  (\bibinfo{series}{PLDI 2020})}. \bibinfo{publisher}{Association for Computing
  Machinery}, \bibinfo{address}{New York, NY, USA}, \bibinfo{pages}{780–793}.
\newblock
\showISBNx{9781450376136}
\showDOI{%
\url{https://doi.org/10.1145/3385412.3386021}}


\bibitem[\protect\citeauthoryear{Li, Zhang, and Reps}{Li et~al\mbox{.}}{2021}]%
        {Li2021}
\bibfield{author}{\bibinfo{person}{Yuanbo Li}, \bibinfo{person}{Qirun Zhang},
  {and} \bibinfo{person}{Thomas Reps}.} \bibinfo{year}{2021}\natexlab{}.
\newblock \showarticletitle{On the Complexity of Bidirected Interleaved
  Dyck-Reachability}.
\newblock \bibinfo{journal}{{\em Proc. ACM Program. Lang.\/}}
  \bibinfo{volume}{5}, \bibinfo{number}{POPL}, Article \bibinfo{articleno}{59}
  (\bibinfo{date}{Jan.} \bibinfo{year}{2021}), \bibinfo{numpages}{28}~pages.
\newblock
\showDOI{%
\url{https://doi.org/10.1145/3434340}}


\bibitem[\protect\citeauthoryear{Lu and Xue}{Lu and Xue}{2019}]%
        {Lu2019}
\bibfield{author}{\bibinfo{person}{Jingbo Lu} {and} \bibinfo{person}{Jingling
  Xue}.} \bibinfo{year}{2019}\natexlab{}.
\newblock \showarticletitle{Precision-Preserving yet Fast Object-Sensitive
  Pointer Analysis with Partial Context Sensitivity}.
\newblock \bibinfo{journal}{{\em Proc. ACM Program. Lang.\/}}
  \bibinfo{volume}{3}, \bibinfo{number}{OOPSLA}, Article
  \bibinfo{articleno}{148} (\bibinfo{date}{Oct.} \bibinfo{year}{2019}),
  \bibinfo{numpages}{29}~pages.
\newblock
\showDOI{%
\url{https://doi.org/10.1145/3360574}}


\bibitem[\protect\citeauthoryear{Madhusudan and Parlato}{Madhusudan and
  Parlato}{2011}]%
        {Madhusudan2011}
\bibfield{author}{\bibinfo{person}{P. Madhusudan} {and}
  \bibinfo{person}{Gennaro Parlato}.} \bibinfo{year}{2011}\natexlab{}.
\newblock \showarticletitle{The Tree Width of Auxiliary Storage}. In
  \bibinfo{booktitle}{{\em Proceedings of the 38th Annual ACM SIGPLAN-SIGACT
  Symposium on Principles of Programming Languages}} {\em
  (\bibinfo{series}{POPL '11})}. \bibinfo{publisher}{Association for Computing
  Machinery}, \bibinfo{address}{New York, NY, USA}, \bibinfo{pages}{283–294}.
\newblock
\showISBNx{9781450304900}
\showDOI{%
\url{https://doi.org/10.1145/1926385.1926419}}


\bibitem[\protect\citeauthoryear{Mathiasen and Pavlogiannis}{Mathiasen and
  Pavlogiannis}{2021}]%
        {Mathiasen2021}
\bibfield{author}{\bibinfo{person}{Anders~Alnor Mathiasen} {and}
  \bibinfo{person}{Andreas Pavlogiannis}.} \bibinfo{year}{2021}\natexlab{}.
\newblock \showarticletitle{The Fine-Grained and Parallel Complexity of
  Andersen’s Pointer Analysis}.
\newblock \bibinfo{journal}{{\em Proc. ACM Program. Lang.\/}}
  \bibinfo{volume}{5}, \bibinfo{number}{POPL}, Article \bibinfo{articleno}{34}
  (\bibinfo{date}{Jan.} \bibinfo{year}{2021}), \bibinfo{numpages}{29}~pages.
\newblock
\showDOI{%
\url{https://doi.org/10.1145/3434315}}


\bibitem[\protect\citeauthoryear{Milanova}{Milanova}{2020}]%
        {Milanova2020}
\bibfield{author}{\bibinfo{person}{Ana Milanova}.}
  \bibinfo{year}{2020}\natexlab{}.
\newblock \showarticletitle{FlowCFL: Generalized Type-Based Reachability
  Analysis: Graph Reduction and Equivalence of CFL-Based and Type-Based
  Reachability}.
\newblock \bibinfo{journal}{{\em Proc. ACM Program. Lang.\/}}
  \bibinfo{volume}{4}, \bibinfo{number}{OOPSLA}, Article
  \bibinfo{articleno}{178} (\bibinfo{date}{Nov.} \bibinfo{year}{2020}),
  \bibinfo{numpages}{29}~pages.
\newblock
\showDOI{%
\url{https://doi.org/10.1145/3428246}}


\bibitem[\protect\citeauthoryear{Pierre}{Pierre}{1992}]%
        {Pierre1992}
\bibfield{author}{\bibinfo{person}{Laurent Pierre}.}
  \bibinfo{year}{1992}\natexlab{}.
\newblock \showarticletitle{Rational indexes of generators of the cone of
  context-free languages}.
\newblock \bibinfo{journal}{{\em Theoretical Computer Science\/}}
  \bibinfo{volume}{95}, \bibinfo{number}{2} (\bibinfo{year}{1992}),
  \bibinfo{pages}{279 -- 305}.
\newblock
\showISSN{0304-3975}
\showDOI{%
\url{https://doi.org/10.1016/0304-3975(92)90269-L}}


\bibitem[\protect\citeauthoryear{Qadeer and Rehof}{Qadeer and Rehof}{2005}]%
        {Qadeer2005}
\bibfield{author}{\bibinfo{person}{Shaz Qadeer} {and} \bibinfo{person}{Jakob
  Rehof}.} \bibinfo{year}{2005}\natexlab{}.
\newblock \showarticletitle{Context-Bounded Model Checking of Concurrent
  Software}. In \bibinfo{booktitle}{{\em Proceedings of the 11th International
  Conference on Tools and Algorithms for the Construction and Analysis of
  Systems}} {\em (\bibinfo{series}{TACAS'05})}.
  \bibinfo{publisher}{Springer-Verlag}, \bibinfo{address}{Berlin, Heidelberg},
  \bibinfo{pages}{93–107}.
\newblock
\showISBNx{3540253335}
\showDOI{%
\url{https://doi.org/10.1007/978-3-540-31980-1_7}}


\bibitem[\protect\citeauthoryear{Rehof and F\"{a}hndrich}{Rehof and
  F\"{a}hndrich}{2001}]%
        {Rehof01}
\bibfield{author}{\bibinfo{person}{Jakob Rehof} {and} \bibinfo{person}{Manuel
  F\"{a}hndrich}.} \bibinfo{year}{2001}\natexlab{}.
\newblock \showarticletitle{Type-base Flow Analysis: From Polymorphic Subtyping
  to CFL-reachability}. In \bibinfo{booktitle}{{\em Proceedings of the 28th ACM
  SIGPLAN-SIGACT Symposium on Principles of Programming Languages}} {\em
  (\bibinfo{series}{POPL})}. \bibinfo{pages}{54--66}.
\newblock


\bibitem[\protect\citeauthoryear{Reps}{Reps}{1995}]%
        {Reps1995b}
\bibfield{author}{\bibinfo{person}{Thomas Reps}.}
  \bibinfo{year}{1995}\natexlab{}.
\newblock \showarticletitle{Shape Analysis As a Generalized Path Problem}. In
  \bibinfo{booktitle}{{\em Proceedings of the 1995 ACM SIGPLAN Symposium on
  Partial Evaluation and Semantics-based Program Manipulation}} {\em
  (\bibinfo{series}{PEPM '95})}. \bibinfo{publisher}{ACM},
  \bibinfo{pages}{1--11}.
\newblock


\bibitem[\protect\citeauthoryear{Reps}{Reps}{1997}]%
        {Reps97}
\bibfield{author}{\bibinfo{person}{Thomas Reps}.}
  \bibinfo{year}{1997}\natexlab{}.
\newblock \showarticletitle{Program Analysis via Graph Reachability}. In
  \bibinfo{booktitle}{{\em Proceedings of the 1997 International Symposium on
  Logic Programming}} {\em (\bibinfo{series}{ILPS})}. \bibinfo{pages}{5--19}.
\newblock


\bibitem[\protect\citeauthoryear{Reps}{Reps}{2000}]%
        {Reps00}
\bibfield{author}{\bibinfo{person}{Thomas Reps}.}
  \bibinfo{year}{2000}\natexlab{}.
\newblock \showarticletitle{Undecidability of Context-sensitive Data-dependence
  Analysis}.
\newblock \bibinfo{journal}{{\em ACM Trans. Program. Lang. Syst.\/}}
  \bibinfo{volume}{22}, \bibinfo{number}{1} (\bibinfo{year}{2000}),
  \bibinfo{pages}{162--186}.
\newblock


\bibitem[\protect\citeauthoryear{Reps, Horwitz, and Sagiv}{Reps
  et~al\mbox{.}}{1995}]%
        {Reps95}
\bibfield{author}{\bibinfo{person}{Thomas Reps}, \bibinfo{person}{Susan
  Horwitz}, {and} \bibinfo{person}{Mooly Sagiv}.}
  \bibinfo{year}{1995}\natexlab{}.
\newblock \showarticletitle{Precise Interprocedural Dataflow Analysis via Graph
  Reachability}. In \bibinfo{booktitle}{{\em POPL}}. \bibinfo{publisher}{ACM},
  \bibinfo{address}{New York, NY, USA}.
\newblock


\bibitem[\protect\citeauthoryear{Reps, Horwitz, Sagiv, and Rosay}{Reps
  et~al\mbox{.}}{1994}]%
        {Reps94}
\bibfield{author}{\bibinfo{person}{Thomas Reps}, \bibinfo{person}{Susan
  Horwitz}, \bibinfo{person}{Mooly Sagiv}, {and} \bibinfo{person}{Genevieve
  Rosay}.} \bibinfo{year}{1994}\natexlab{}.
\newblock \showarticletitle{Speeding Up Slicing}.
\newblock \bibinfo{journal}{{\em SIGSOFT Softw. Eng. Notes\/}}
  \bibinfo{volume}{19}, \bibinfo{number}{5} (\bibinfo{year}{1994}),
  \bibinfo{pages}{11--20}.
\newblock


\bibitem[\protect\citeauthoryear{Schmitz and Zetzsche}{Schmitz and
  Zetzsche}{2019}]%
        {Schmitz2019}
\bibfield{author}{\bibinfo{person}{Sylvain Schmitz} {and}
  \bibinfo{person}{Georg Zetzsche}.} \bibinfo{year}{2019}\natexlab{}.
\newblock \showarticletitle{Coverability Is Undecidable in One-Dimensional
  Pushdown Vector Addition Systems with Resets}. In \bibinfo{booktitle}{{\em
  Reachability Problems}}, \bibfield{editor}{\bibinfo{person}{Emmanuel Filiot},
  \bibinfo{person}{Rapha{\"e}l Jungers}, {and} \bibinfo{person}{Igor Potapov}}
  (Eds.). \bibinfo{publisher}{Springer International Publishing},
  \bibinfo{address}{Cham}, \bibinfo{pages}{193--201}.
\newblock
\showISBNx{978-3-030-30806-3}


\bibitem[\protect\citeauthoryear{Shang, Xie, and Xue}{Shang
  et~al\mbox{.}}{2012}]%
        {Shang2012}
\bibfield{author}{\bibinfo{person}{Lei Shang}, \bibinfo{person}{Xinwei Xie},
  {and} \bibinfo{person}{Jingling Xue}.} \bibinfo{year}{2012}\natexlab{}.
\newblock \showarticletitle{On-demand Dynamic Summary-based Points-to
  Analysis}. In \bibinfo{booktitle}{{\em Proceedings of the Tenth International
  Symposium on Code Generation and Optimization}} {\em (\bibinfo{series}{CGO
  '12})}. \bibinfo{publisher}{ACM}, \bibinfo{pages}{264--274}.
\newblock


\bibitem[\protect\citeauthoryear{Shivers}{Shivers}{1991}]%
        {Shivers1991}
\bibfield{author}{\bibinfo{person}{Olin~Grigsby Shivers}.}
  \bibinfo{year}{1991}\natexlab{}.
\newblock {\em \bibinfo{title}{Control-Flow Analysis of Higher-Order Languages
  of Taming Lambda}}.
\newblock \bibinfo{thesistype}{Ph.D. Dissertation}. \bibinfo{address}{USA}.
\newblock
\newblock
\shownote{UMI Order No. GAX91-26964.}


\bibitem[\protect\citeauthoryear{Sp\"{a}th, Ali, and Bodden}{Sp\"{a}th
  et~al\mbox{.}}{2019}]%
        {Spath2019}
\bibfield{author}{\bibinfo{person}{Johannes Sp\"{a}th}, \bibinfo{person}{Karim
  Ali}, {and} \bibinfo{person}{Eric Bodden}.} \bibinfo{year}{2019}\natexlab{}.
\newblock \showarticletitle{Context-, Flow-, and Field-Sensitive Data-Flow
  Analysis Using Synchronized Pushdown Systems}.
\newblock \bibinfo{journal}{{\em Proc. ACM Program. Lang.\/}}
  \bibinfo{volume}{3}, \bibinfo{number}{POPL}, Article \bibinfo{articleno}{48}
  (\bibinfo{date}{Jan.} \bibinfo{year}{2019}), \bibinfo{numpages}{29}~pages.
\newblock
\showDOI{%
\url{https://doi.org/10.1145/3290361}}


\bibitem[\protect\citeauthoryear{Sridharan and Bod\'{\i}k}{Sridharan and
  Bod\'{\i}k}{2006}]%
        {Sridharan2006}
\bibfield{author}{\bibinfo{person}{Manu Sridharan} {and}
  \bibinfo{person}{Rastislav Bod\'{\i}k}.} \bibinfo{year}{2006}\natexlab{}.
\newblock \showarticletitle{Refinement-based Context-sensitive Points-to
  Analysis for Java}.
\newblock \bibinfo{journal}{{\em SIGPLAN Not.\/}} \bibinfo{volume}{41},
  \bibinfo{number}{6} (\bibinfo{year}{2006}), \bibinfo{pages}{387--400}.
\newblock


\bibitem[\protect\citeauthoryear{Sridharan, Gopan, Shan, and
  Bod\'{\i}k}{Sridharan et~al\mbox{.}}{2005}]%
        {Sridharan2005}
\bibfield{author}{\bibinfo{person}{Manu Sridharan}, \bibinfo{person}{Denis
  Gopan}, \bibinfo{person}{Lexin Shan}, {and} \bibinfo{person}{Rastislav
  Bod\'{\i}k}.} \bibinfo{year}{2005}\natexlab{}.
\newblock \showarticletitle{Demand-driven Points-to Analysis for Java}. In
  \bibinfo{booktitle}{{\em OOPSLA}}.
\newblock


\bibitem[\protect\citeauthoryear{Tang, Wang, Xiong, Zhang, Wang, and
  Zhang}{Tang et~al\mbox{.}}{2017}]%
        {Tang2017}
\bibfield{author}{\bibinfo{person}{Hao Tang}, \bibinfo{person}{Di Wang},
  \bibinfo{person}{Yingfei Xiong}, \bibinfo{person}{Lingming Zhang},
  \bibinfo{person}{Xiaoyin Wang}, {and} \bibinfo{person}{Lu Zhang}.}
  \bibinfo{year}{2017}\natexlab{}.
\newblock \showarticletitle{Conditional Dyck-CFL Reachability Analysis for
  Complete and Efficient Library Summarization}. In \bibinfo{booktitle}{{\em
  Programming Languages and Systems}},
  \bibfield{editor}{\bibinfo{person}{Hongseok Yang}} (Ed.).
  \bibinfo{publisher}{Springer Berlin Heidelberg}, \bibinfo{address}{Berlin,
  Heidelberg}, \bibinfo{pages}{880--908}.
\newblock
\showISBNx{978-3-662-54434-1}


\bibitem[\protect\citeauthoryear{Vedurada and Nandivada}{Vedurada and
  Nandivada}{2019}]%
        {Vedurada2019}
\bibfield{author}{\bibinfo{person}{Jyothi Vedurada} {and}
  \bibinfo{person}{V.~Krishna Nandivada}.} \bibinfo{year}{2019}\natexlab{}.
\newblock \showarticletitle{Batch Alias Analysis}. In \bibinfo{booktitle}{{\em
  Proceedings of the 34th IEEE/ACM International Conference on Automated
  Software Engineering}} {\em (\bibinfo{series}{ASE '19})}.
  \bibinfo{publisher}{IEEE Press}, \bibinfo{pages}{936–948}.
\newblock
\showISBNx{9781728125084}
\showDOI{%
\url{https://doi.org/10.1109/ASE.2019.00091}}


\bibitem[\protect\citeauthoryear{Williams}{Williams}{2005}]%
        {Williams05}
\bibfield{author}{\bibinfo{person}{Ryan Williams}.}
  \bibinfo{year}{2005}\natexlab{}.
\newblock \showarticletitle{A New Algorithm for Optimal 2-Constraint
  Satisfaction and Its Implications}.
\newblock \bibinfo{journal}{{\em Theor. Comput. Sci.\/}} \bibinfo{volume}{348},
  \bibinfo{number}{2} (\bibinfo{date}{Dec.} \bibinfo{year}{2005}),
  \bibinfo{pages}{357–365}.
\newblock
\showISSN{0304-3975}
\showDOI{%
\url{https://doi.org/10.1016/j.tcs.2005.09.023}}


\bibitem[\protect\citeauthoryear{Williams}{Williams}{2019}]%
        {Williams19}
\bibfield{author}{\bibinfo{person}{Virginia~Vassilevska Williams}.}
  \bibinfo{year}{2019}\natexlab{}.
\newblock \bibinfo{booktitle}{{\em On some fine-grained questions in algorithms
  and complexity}}.
\newblock \bibinfo{type}{{T}echnical {R}eport}.
\newblock


\bibitem[\protect\citeauthoryear{Xu, Rountev, and Sridharan}{Xu
  et~al\mbox{.}}{2009}]%
        {Xu09}
\bibfield{author}{\bibinfo{person}{Guoqing Xu}, \bibinfo{person}{Atanas
  Rountev}, {and} \bibinfo{person}{Manu Sridharan}.}
  \bibinfo{year}{2009}\natexlab{}.
\newblock \showarticletitle{Scaling CFL-Reachability-Based Points-To Analysis
  Using Context-Sensitive Must-Not-Alias Analysis}. In \bibinfo{booktitle}{{\em
  Proceedings of the 23rd European Conference on ECOOP 2009 --- Object-Oriented
  Programming}} {\em (\bibinfo{series}{Genoa})}. \bibinfo{pages}{98--122}.
\newblock


\bibitem[\protect\citeauthoryear{Yan, Xu, and Rountev}{Yan
  et~al\mbox{.}}{2011}]%
        {Yan11}
\bibfield{author}{\bibinfo{person}{Dacong Yan}, \bibinfo{person}{Guoqing Xu},
  {and} \bibinfo{person}{Atanas Rountev}.} \bibinfo{year}{2011}\natexlab{}.
\newblock \showarticletitle{Demand-driven Context-sensitive Alias Analysis for
  Java}. In \bibinfo{booktitle}{{\em Proceedings of the 2011 International
  Symposium on Software Testing and Analysis}} {\em (\bibinfo{series}{ISSTA})}.
  \bibinfo{pages}{155--165}.
\newblock


\bibitem[\protect\citeauthoryear{Yannakakis}{Yannakakis}{1990}]%
        {Yannakakis90}
\bibfield{author}{\bibinfo{person}{Mihalis Yannakakis}.}
  \bibinfo{year}{1990}\natexlab{}.
\newblock \showarticletitle{Graph-theoretic Methods in Database Theory}. In
  \bibinfo{booktitle}{{\em Proceedings of the Ninth ACM SIGACT-SIGMOD-SIGART
  Symposium on Principles of Database Systems}} {\em (\bibinfo{series}{PODS})}.
  \bibinfo{pages}{230--242}.
\newblock


\bibitem[\protect\citeauthoryear{Yuan and Eugster}{Yuan and Eugster}{2009}]%
        {Yuan09}
\bibfield{author}{\bibinfo{person}{Hao Yuan} {and} \bibinfo{person}{Patrick
  Eugster}.} \bibinfo{year}{2009}\natexlab{}.
\newblock \showarticletitle{An Efficient Algorithm for Solving the Dyck-CFL
  Reachability Problem on Trees}. In \bibinfo{booktitle}{{\em Proceedings of
  the 18th European Symposium on Programming Languages and Systems: Held As
  Part of the Joint European Conferences on Theory and Practice of Software,
  ETAPS 2009}} {\em (\bibinfo{series}{ESOP})}. \bibinfo{pages}{175--189}.
\newblock


\bibitem[\protect\citeauthoryear{Zhang, Lyu, Yuan, and Su}{Zhang
  et~al\mbox{.}}{2013}]%
        {Zhang13}
\bibfield{author}{\bibinfo{person}{Qirun Zhang}, \bibinfo{person}{Michael~R.
  Lyu}, \bibinfo{person}{Hao Yuan}, {and} \bibinfo{person}{Zhendong Su}.}
  \bibinfo{year}{2013}\natexlab{}.
\newblock \showarticletitle{Fast Algorithms for Dyck-CFL-reachability with
  Applications to Alias Analysis} {\em (\bibinfo{series}{PLDI})}.
  \bibinfo{publisher}{ACM}.
\newblock


\bibitem[\protect\citeauthoryear{Zhang and Su}{Zhang and Su}{2017}]%
        {Zhang2017}
\bibfield{author}{\bibinfo{person}{Qirun Zhang} {and} \bibinfo{person}{Zhendong
  Su}.} \bibinfo{year}{2017}\natexlab{}.
\newblock \showarticletitle{Context-Sensitive Data-Dependence Analysis via
  Linear Conjunctive Language Reachability}.
\newblock \bibinfo{journal}{{\em SIGPLAN Not.\/}} \bibinfo{volume}{52},
  \bibinfo{number}{1} (\bibinfo{date}{Jan.} \bibinfo{year}{2017}),
  \bibinfo{pages}{344–358}.
\newblock
\showISSN{0362-1340}
\showDOI{%
\url{https://doi.org/10.1145/3093333.3009848}}


\bibitem[\protect\citeauthoryear{Zheng and Rugina}{Zheng and Rugina}{2008}]%
        {Zheng2008}
\bibfield{author}{\bibinfo{person}{Xin Zheng} {and} \bibinfo{person}{Radu
  Rugina}.} \bibinfo{year}{2008}\natexlab{}.
\newblock \showarticletitle{Demand-driven Alias Analysis for C}. In
  \bibinfo{booktitle}{{\em Proceedings of the 35th Annual ACM SIGPLAN-SIGACT
  Symposium on Principles of Programming Languages}} {\em
  (\bibinfo{series}{POPL '08})}. \bibinfo{publisher}{ACM},
  \bibinfo{pages}{197--208}.
\newblock


\end{thebibliography}
